\pdfoutput=1
\makeatletter\def\input@path{{FIGURES/}}\makeatother
\documentclass[10pt,a4paper,reqno]{amsart}
\usepackage[utf8]{inputenc}
\usepackage[T1]{fontenc}
\usepackage[a4paper,margin=20mm]{geometry}
\usepackage{amsmath,amssymb}
\usepackage[normalem]{ulem}
\usepackage{esint} % Pour \fint
\usepackage{graphicx}
\usepackage{float}% Include figure files
\usepackage{tikz}

\makeatletter
\let\article@ordinaryfootnotetext\@footnotetext
\makeatother
\usepackage[para]{footmisc}

\usepackage{etoolbox}
\usepackage{bm}% bold math
\usepackage{lmodern}
\usepackage{xcolor}
\usepackage{slashed}
\usepackage{mathtools,hyperref}
\usepackage{mathrsfs}

\let\origaligned\aligned   \renewcommand*{\aligned}{\origaligned\relax}

\newcommand \Hess {\operatorname{Hess}} 

\newcommand \Ric {{\operatorfont Ric}}

\newcommand \ebf {\mathbf{e}}
\newcommand \del {\partial}

\newcommand \RR {\mathbb{R}}
\newcommand \eps {\epsilon}
\newcommand \be {\begin{equation}\mathopen{}} % avoids an issue with hyperref+amsmath
\newcommand \ee {\end{equation}}
\newcommand \bse {\begin{subequations}}
\newcommand \ese {\end{subequations}}

\DeclareMathOperator{\supp}{supp}

\DeclareMathOperator \Tr {\bf Tr}  
\newcommand \Deltaslash {\slashed\Delta{}}
\newcommand \nablaslash {\slashed\nabla{}}
\newcommand \gslash {\slashed{g}{}}
\newcommand \Acal {\mathcal{A}}
\newcommand{\ssAAUX}{\slashed{\Acal}}

\declareslashed{}{/}{.0}{0}{\Acal}
\declareslashed{}{/}{.15}{0}{\ssAAUX}
\newcommand{\ssrmAAUX}{\slashed{\mathrm{A}}}

\declareslashed{}{/}{.1}{0}{\mathrm{A}}
\declareslashed{}{/}{.3}{0}{\ssrmAAUX}

\newcommand \seedg {g_{\textnormal{\emph{s}}}}
\newcommand \seedh {h_{\textnormal{\emph{s}}}}
\newcommand \normal {\textbf{sil}}

\newtheorem{theorem}{Theorem}[section] 
 
\newtheorem{proposition}[theorem]{Proposition}
\newtheorem{lemma}[theorem]{Lemma}
\newtheorem{corollary}[theorem]{Corollary}

\theoremstyle{remark}
\newtheorem*{remark}{Remark}
\theoremstyle{plain}
\makeatletter
\@ifundefined{proof}{\newenvironment{proof}{\par\noindent\emph{Proof.} }{\hfill$\square$\par\medskip}}{}
\makeatother

\newcommand{\divs}{\operatorname{div}_{\gslash}}
\newcommand{\trs}{\operatorname{tr}_{\gslash}}
\newcommand{\Eappr}[1]{E_{\textbf{appr},#1}}

\newcommand{\Crad}{c_{\mathrm{rad}}}
\newcommand{\Obig}{\mathcal{O}}
\newcommand{\Qsc}{\mathcal{Q}_R}

\numberwithin{equation}{section}
\numberwithin{table}{section}
\numberwithin{figure}{section}
\makeatletter
\renewcommand{\footnoterule}{\kern-3pt\hrule width5pc\kern2.6pt}
\def\@captionfont{\normalfont\small}
\long\def\@makefntext#1{\leavevmode\fontfamily{qtm}\fontsize{7.5}{9.5}\selectfont
  \raisebox{.55ex}{\fontsize{5.5}{6}\selectfont\@thefnmark}%
  \nobreak\hskip.12em#1}
\makeatother

\allowdisplaybreaks[2]
\hypersetup{bookmarksdepth=2,colorlinks=false}
\hypersetup{hidelinks}
\newenvironment{acknowledgments}{\section*{Acknowledgments}}{}

\makeatletter
\patchcmd{\@setauthors}{\MakeUppercase{\authors}}%
  {{\normalfont\large\authors}}{}%
  {\PackageError{article-trois}{Author font patch failed}{Check the AMS author macro.}}
\makeatother

\makeatletter
\def\@setaddresses{%
  \begingroup
  \normalfont\fontsize{8}{9.5}\selectfont
  \parindent\z@ \parskip\z@
  \interlinepenalty\@M
  \def\author##1{\par\nobreak\noindent}%
  \def\\{\unskip,\space\ignorespaces}%
  \def\address##1##2{%
    \@ifnotempty{##1}{(\ignorespaces##1\unskip)\space}%
    \ignorespaces##2\unskip}%
  \def\email##1##2{%
    \@ifnotempty{##2}{%
      \unskip\space Email address%
      \@ifnotempty{##1}{, \ignorespaces##1\unskip}%
      :\space\href{mailto:##2}{##2}}}%
  \addresses
  \unskip
  \endgroup}
\makeatother

\makeatletter
\patchcmd{\FN@makefootnoteparagraph}{\@parboxrestore}{\@parboxrestore\ifnum\c@page=3\relax\parfillskip=0pt\rightskip=44pt\fi}{}{\PackageError{article}{Footnote paragraph patch failed}{}}
\let\article@paragraphfootnotes\FN@makefootnoteparagraph
\def\FN@makefootnoteparagraph{%
  \ifnum\c@page=\@ne
    \unvbox\footins
  \else
    \article@paragraphfootnotes
  \fi}
\let\article@adminfootnotes\@adminfootnotes
\def\@adminfootnotes{%
  \begingroup
  \let\@footnotetext\article@ordinaryfootnotetext
  \long\def\@makefntext##1{%
    \normalfont\fontsize{8}{9.5}\selectfont\noindent##1}%
  \interfootnotelinepenalty\@M
  \article@adminfootnotes
  \ifx\@empty\addresses\else
    \@footnotetext{\@setaddresses}%
  \fi
  \endgroup
  \global\let\addresses\@empty}
\makeatother

\makeatletter
\apptocmd{\@adjustvertspacing}{%
  \advance\abovedisplayskip by -1pt
  \advance\belowdisplayskip by -1pt}{}{\PackageError{article}{Display spacing adjustment failed}{}}
\makeatother
\AtBeginDocument{\normalsize}

\begin{document}
\title[The seeds and silhouettes of optimal gravitational shielding]{The seeds and silhouettes of optimal gravitational shielding}
\author[Bruno Le Floch]{Bruno Le Floch} 
 \address{Laboratoire de Physique Théorique et Hautes Énergies \& Centre National de la Recherche Scientifique, Sorbonne Universit\'e, 4~place Jussieu, 75252 Paris, France.}
\email{bruno@le-floch.fr}
\author[Philippe G. LeFloch]{Philippe G. LeFloch}
\address{Laboratoire Jacques-Louis Lions \& Centre National de la Recherche Scientifique, Sorbonne Universit\'e, 4~place Jussieu, 75252 Paris, France.}
\email{contact@philippelefloch.org}

\date{October 2026}
\begin{abstract}
We construct complete vacuum initial data sets with prescribed asymptotically flat exterior geometry, including the gluing of two Schwarzschild exterior regions with small masses (which need not be positive). The construction allows for arbitrarily slow leading decay and nonzero second fundamental form. For leading terms with a fixed radial power, we first construct localized seed data sets whose constraint errors decay to any prescribed finite order, then apply the projection and stability theorems of our companion papers. The resulting vacuum data sets differ from seeds by harmonic terms controlled by Hamiltonian and momentum silhouettes, and by super-harmonic corrections.
This difference of harmonic terms is measured by a relative notion of ADM energy-momentum, which is defined even for slowly-decaying data sets that do not admit an ADM energy-momentum.  Upon adjusting seed data sets by either harmonic-decaying terms or compact transverse-traceless terms, we freely control this relative ADM charge and, when it exists, the total ADM charge.  In particular, when gluing Schwarzschild metrics the total ADM mass can be specified independently of the Schwarzschild mass parameters.
Together with our companion papers, these results establish a constructive theory of optimal localization beyond harmonic decay.
\end{abstract}
\maketitle
\markboth{THE SEEDS AND SILHOUETTES OF OPTIMAL GRAVITATIONAL SHIELDING}{THE SEEDS AND SILHOUETTES OF OPTIMAL GRAVITATIONAL SHIELDING}
\tableofcontents

%============================================================================

\section{Introduction}\label{section-1}

%-----------------------------------------------------------------------------------------------------------------------------------------

\subsection{Aim of this paper}\label{section:introduction-aim}
Einstein's vacuum constraint equations determine which Riemannian metrics and second fundamental forms can occur on a spacelike hypersurface. We study, in three space dimensions, the construction of asymptotically Euclidean initial data sets with prescribed geometry in selected exterior regions. The basic example is Schwarzschild--Minkowski gluing: the exact data sets must feature a Schwarzschild metric in one angular region and the Euclidean metric in another, outside a fixed ball. Our theory of optimal gravitational shielding began with \emph{Optimal shielding for Einstein gravity}~\cite{BLF--PLF-shielding}, published in 2024. In that work, we introduced the localized seed-to-solution method and the silhouettes that describe the angular form of the harmonic corrections.

We introduced in~\cite{BLF--PLF} a localized seed-to-solution projection operator with harmonic and super-harmonic estimates. We proved in~\cite{BLF--PLF-stability} that the required stability conditions on the localization function do hold as a consequence of Poincar\'e, Korn, and Hardy's inequalities. The problem addressed here is to \emph{construct seed data sets} that meet all the hypotheses of the projection and stability theory, while retaining prescribed exterior values and allowing nonzero second fundamental forms.

A direct angular interpolation between the prescribed data sets leaves a leading constraint error that only decays slowly at radial infinity. We remove successive errors by radial equations that preserve angular support. This construction separates the \emph{geometry exponent} $p_G$, which describes the leading decay of the seed data set, from the \emph{accuracy exponent} $p_A$, which describes the decay of its constraint error. We fix the geometry, profiles, radius, and exponents before choosing sufficiently small masses (or other amplitudes) and analyzing the gluing.

We also describe the harmonic corrections through the Hamiltonian and momentum silhouettes of our companion papers~\cite{BLF--PLF,BLF--PLF-stability}. For fixed localization geometry and variational weights, these angular profiles remain fixed as the seed data sets vary. Their amplitudes are the relative energy and momentum between the exact solution and its seed data set, which we select through parameters in the seed data sets.

%-----------------------------------------------------------------------------------------------------------------------------------------

\subsection{Principal contributions}\label{section-1-1}
Table~\ref{tab:principal-theorems} lists the principal theorems and their main conclusions.
Our main results concern the following two questions.
\begin{itemize}
\item \emph{Prescribed exterior geometry.} We construct complete vacuum data sets agreeing with two Schwarzschild metrics in prescribed exterior regions, provided both masses are sufficiently small. More generally, we find solutions with non-vanishing second fundamental form and \emph{any leading exponent} $0<p_G<1$. The difference between the constructed exact data set and its prescribed seed data set consists of harmonic energy-momentum terms and a super-harmonic remainder. Cf. Theorems~\ref{thm:fixed-radius} and~\ref{thm:coupled-projection}, below.

\item \emph{Choosing relative and total energy-momentum.} We also prove that we can pick four harmonic parameters to \emph{prescribe the (relative) energy-momentum} of an exact solution relative to its seed data set. Namely, compactly supported tensor perturbations (with zero Euclidean trace and divergence) allow us to prescribe sufficiently small total ADM vectors near each fixed future timelike direction. The total ADM mass can for instance be equal to the positive Schwarzschild mass parameter retained in an exterior region. Cf. Theorems~\ref{thm:vector-selection} and~\ref{thm:effective-ADM}, and Proposition~\ref{prop:Sch-mass-prescribed}, below.
\end{itemize}

The proofs also give explicit results on the seed data sets. We solve the scalar-curvature and momentum equations to arbitrary finite order, including resonant logarithmic terms (Propositions~\ref{prop:finite-expansion} and~\ref{prop:coupled-expansion}). We calculate how compact perturbations change the ADM energy and momentum to quadratic order, and show that all four components can be varied independently (Propositions~\ref{lem:TT-quadratic} and~\ref{lem:TT-packets}). We also determine the solutions of the linearized constraints with a fixed radial power and prescribed angular support, and identify their angular compatibility conditions.

Our procedure for choosing the scalar harmonic parameter extends the construction of LeFloch and Nguyen~\cite[Sections~6.2--6.3]{LeFlochNguyen}. We distinguish throughout the prescribed exterior Schwarzschild parameters, the relative energy-momentum correction, and the total ADM vector. We can also vary the total ADM vector of a nonzero background while keeping its leading profiles and exterior values. For $1/2<p_G<1$ and sufficiently small future timelike ADM vectors, it suffices that the background norm, after rescaling to unit radius, be sufficiently small compared with the square root of the rescaled target energy (cf. Proposition~\ref{prop:fixed-background-ADM}, below).

\begin{table}
\centering
\begingroup\small\setlength{\tabcolsep}{0pt}\renewcommand{\arraystretch}{1.2}
\begin{tabular}{@{}l@{\hspace{1.4em}}l@{}}
\hline
Theorem~\ref{thm:fixed-radius} & Complete scalar-flat gluing of two Schwarzschild exterior regions. \\
Theorem~\ref{thm:coupled-projection} & Complete vacuum data with slow leading decay and super-harmonic remainders. \\
Theorem~\ref{thm:vector-selection} & Choice of relative energy-momentum with fixed exterior data. \\
Theorem~\ref{thm:effective-ADM} & Choice of small total ADM vectors near a fixed future timelike direction. \\
\hline
\end{tabular}
\endgroup
\caption{The five principal theorems.}
\label{tab:principal-theorems}
\end{table}

%-----------------------------------------------------------------------------------------------------------------------------------------

\subsection{The Einstein constraints and asymptotic parameters}\label{section-1-2}
A vacuum data set $(g,h)$ consists of a Riemannian metric and a symmetric contravariant tensor satisfying the vacuum Hamiltonian and momentum constraints
\be
\label{eq:full-constraints}
\mathcal H(g,h)\coloneqq R_g+\tfrac12(\Tr_g h)^2-|h|_g^2=0,\qquad \mathcal C^j(g,h)\coloneqq\nabla_i h^{ij}=0.
\ee
Here $R_g$ denotes scalar curvature and $\nabla$ is the Levi-Civita connection of $g$, traces and norms are taken with $g$, and we use the momentum variable~$h$ as in~\cite{BLF--PLF}: if $K$ is the second fundamental form, we set
\be
\label{eq:full-conventions}
h^{ij}\coloneqq K^{ij}-(\Tr_g K)g^{ij},\qquad K_{ij}=g_{ia}g_{jb}h^{ab}-\tfrac12(g_{ab}h^{ab})g_{ij}.
\ee
The momentum operator in~\cite{BLF--PLF} is $\mathcal M_i=g_{ij}\mathcal C^j$. Raising its index gives equivalent weighted norms near the Euclidean metric $\ebf\coloneqq(\delta_{ij})$. We write $\Hess$, $\Delta$, and $\operatorname{div}$ for the Euclidean Hessian, Laplacian, and divergence.

The approach to construct and parametrize exact vacuum data sets in~\cite{BLF--PLF} is to start from a \emph{seed data set} $(\seedg,\seedh)$ satisfying the constraints only approximately, and project it to an exact vacuum data set $(g,h)$.
The \emph{seed-to-solution projection} is based on the linearization of Einstein's constraints~\eqref{eq:full-constraints} at a reference data set which we take to be $(g_0,h_0)=(\ebf,0)$ in the present work, and $(g,h)$ is defined as the (locally unique) vacuum solution of the form\footnote{Contrarily to the general case with non-trivial~$h_0$, the Hamiltonian and momentum constraints linearized around $(\ebf,0)$ are decoupled.}
\be
\label{eq:seed-to-solution}
g = \seedg + \lambda^P r^{n/2-p} d\mathcal H^*[u] , \qquad h = \seedh + \lambda^P r^{n/2-p-1} d\mathcal M^*[Z]
\ee
for some scalar field~$u$ and vector field~$Z$ on~$\RR^3$, with further notation described as needed below.
In particular, the construction of $(g,h)$ requires a \emph{localization function} $\lambda$ subject to several \emph{stability conditions}.

Time-symmetric data sets provide the first construction in Section~\ref{section-2}, below, while Section~\ref{section:full-constraints} extends it to the full Einstein constraints with nonzero second fundamental forms. For $h=0$, the Hamiltonian constraint reduces to the scalar curvature constraint
\be
\label{eq-aHM}
R_g=0,\qquad h=0.
\ee
The induced metric of a standard Minkowski time slice is the Euclidean metric $\ebf=(\delta_{ij})$. In isotropic coordinates, a Schwarzschild time slice of mass $m$ has metric
\be
\label{eq:schwarzschild}
g_m^{\mathrm{Sch}}\coloneqq\Bigl(1+\frac{m}{2r}\Bigr)^4\ebf,\qquad r\coloneqq|x|.
\ee
Section~\ref{section-2} concerns metrics that agree exactly with $g_{m_1}^{\mathrm{Sch}}$ and~$g_{m_2}^{\mathrm{Sch}}$ for small $m_1,m_2\in\RR$ in two angular regions at infinity.

{We use the same notation as in}~\cite{BLF--PLF,BLF--PLF-stability}. The \emph{projection exponent} $p$ determines the variational weights arising in the Corvino--Schoen variational formulation. The \emph{sharp decay exponent} $p_\star$ measures the remainder after subtraction of the harmonic energy-momentum terms, as in~\cite{BLF--PLF}. In dimension three, harmonic decay means metric order $r^{-1}$ and momentum order $r^{-2}$; sub-harmonic and super-harmonic estimates have smaller and larger exponents, respectively. The integer $J$ labels the amplitude order through which the constraints are solved, whereas $N$ counts the controlled derivatives. The weighted norms and admissible exponents are specified in Section~\ref{section-3-1}. We write $\del_i$ for Cartesian derivatives. The notation $\Obig_k(r^{-b})$ describes functions whose order-$j$ derivatives ($0\leq j\leq k$) decay as~$r^{-b-j}$, and likewise for~$o_k(r^{-b})$.

%-----------------------------------------------------------------------------------------------------------------------------------------

\subsection{Related constructions}\label{section-1-3}

Let us briefly review pioneering results in this direction. Corvino and Schoen developed scalar-curvature and constraint gluing~\cite{Corvino-2000,CorvinoSchoen}. Chru\'sciel and Delay~\cite{ChruscielDelay} studied the constraint map in weighted Sobolev spaces and obtained perturbation, gluing, and extension results, including initial data sets exactly Schwarzschild or Kerr near infinity. On compact manifolds, Maxwell identified the changes of free data that make the conformal and conformal thin-sandwich methods equivalent~\cite{MaxwellConformal}. Carlotto and Schoen proved conical localization~\cite{CarlottoSchoen}. Chru\'sciel and Delay further developed scalar-curvature gluing on cone-like regions for asymptotically Euclidean metrics~\cite{ChruscielDelayConical}, and established analogous localization for asymptotically hyperbolic initial data sets~\cite{ChruscielDelayHyperbolic}. Next, in~\cite{LeFlochNguyen}\footnote{\href{https://arxiv.org/abs/1903.00243}{arXiv:1903.00243}} the authors developed the seed-to-solution method and asymptotic localization. Further constructions include characteristic gluing with spacelike applications~\cite{AretakisCzimekRodnianski}\footnote{\href{https://arxiv.org/abs/2107.02441}{arXiv:2107.02441}}, harmonic decay and localization in regions narrower than cones~\cite{MaoTao}\footnote{\href{https://arxiv.org/abs/2210.09437}{arXiv:2210.09437}}, and nonlinear removal of finite-dimensional null-gluing obstructions~\cite{CzimekRodnianski}\footnote{\href{https://arxiv.org/abs/2210.09663}{arXiv:2210.09663}}. Other results give multiple localized time-symmetric systems~\cite{AndersonCorvinoPasqualotto}\footnote{\href{https://arxiv.org/abs/2301.08238}{arXiv:2301.08238}} and spacelike annular gluing by operators with prescribed support~\cite{MaoOhTao}\footnote{\href{https://arxiv.org/abs/2308.13031}{arXiv:2308.13031}}. We developed optimal localization estimates at and beyond harmonic decay in~\cite{BLF--PLF}\footnote{\href{https://arxiv.org/abs/2312.17706}{arXiv:2312.17706}}. Prescribed-decay constructions for Minkowski stability appear in~\cite{Fang-1}\footnote{\href{https://arxiv.org/abs/2401.14353}{arXiv:2401.14353}}. We presented optimal shielding and silhouettes in~\cite{BLF--PLF-shielding}\footnote{\href{https://arxiv.org/abs/2402.17598}{arXiv:2402.17598}}. Initial data sets for black hole stability are constructed in~\cite{Fang-2}\footnote{\href{https://arxiv.org/abs/2405.02071}{arXiv:2405.02071}}. Finally, in~\cite{Sansom-2024}\footnote{\href{https://arxiv.org/abs/2408.10859}{arXiv:2408.10859}}, the author studied $C^3$ null gluing modulo finite-dimensional obstructions.

Recently, Bieri, Garfinkle, Isenberg, Maxwell, and Wheeler use the conformal method to construct maximal vacuum initial data sets with controlled decay, ADM momentum, and anisotropic leading metric terms~\cite{BieriGarfinkleIsenbergMaxwellWheeler}\footnote{\href{https://arxiv.org/abs/2512.21274}{arXiv:2512.21274}}.  Chen and Klainerman~\cite{ChenKlainerman}\footnote{\href{https://arxiv.org/abs/2512.22704}{arXiv:2512.22704}} prescribe four geometric scalars, modulo spherical harmonics of degrees zero and one, and asymptotic charges by an exterior transport--elliptic construction. Their forward construction reaches minimal and borderline decay by effective uniformization~\cite[Theorem~4.3 and Remark~4.4]{ChenKlainermanForward}\footnote{\href{https://arxiv.org/abs/2606.08716}{arXiv:2606.08716}}. We established the stability criteria used in the present construction in~\cite{BLF--PLF-stability}\footnote{\href{https://arxiv.org/abs/2607.28599}{arXiv:2607.28599}}.

Our construction in the present work prepares seed data sets for a seed-to-solution projection, which preserves the prescribed exterior geometries.
We prescribe the components of two tangential trace-free tensors in the seed data set and solve radial equations for the remaining components. Seed-to-solution projection gives complete exact data with prescribed exterior values and unchanged leading profiles.

%-----------------------------------------------------------------------------------------------------------------------------------------

\subsection{Organization of the article}\label{section:introduction-organization}
We first construct seed data sets in Sections~\ref{section-2}--\ref{section:full-constraints}, then verify the hypotheses needed to project them onto exact vacuum solutions in Sections~\ref{section-3}--\ref{section-4}. We consider next their relative and total asymptotic charges and describe the freedom in their leading profiles in Sections~\ref{section-5}--\ref{section-ADM-prescription}.

In Section~\ref{section-2}, we construct seed data sets for the \emph{scalar curvature constraint}, beginning with Schwarzschild--Minkowski data. We obtain exterior expansions of arbitrary finite accuracy for two small Schwarzschild masses of either sign (Proposition~\ref{prop:finite-expansion}). Radial equations preserve the angular support of the corrections, and a radial cutoff extends the seed data sets smoothly through a fixed interior ball.

In Section~\ref{section:full-constraints}, we solve the \emph{Hamiltonian and momentum constraints} together to any prescribed finite order (Proposition~\ref{prop:coupled-expansion}). In the time-symmetric case we construct metrics with \emph{arbitrarily slow decay} whose curvature decays arbitrarily fast.  This is completed by a radial inverse for the momentum equation.  Throughout we include the logarithmic terms arising at resonant exponents. We allow independent amplitudes in the leading metric and momentum profiles. These constructions prepare seed data sets with nonzero second fundamental form for the seed-to-solution projection.

In Section~\ref{section-3}, we verify the weighted pointwise and integral hypotheses of the projection for the Schwarzschild seed data sets constructed in Section~\ref{section-2}. We obtain exact, complete scalar-flat metrics by \emph{gluing two small Schwarzschild masses}, including zero masses or opposite signs (Theorem~\ref{thm:fixed-radius}). In Section~\ref{section:slow-coupled-projection}, we obtain localized solutions with \emph{arbitrarily slow leading metric decay} (Proposition~\ref{thm:slow-leading}) and solve the full constraints with \emph{nonzero second fundamental form} (Theorem~\ref{thm:coupled-projection}). We describe how the Hamiltonian and momentum silhouettes determine the harmonic corrections and distinguish the finite accuracy of the seed data sets from the admissible super-harmonic decay of the remainders.

In Section~\ref{section-4}, we apply these constructions to the \emph{whole-sphere localization problem} and the \emph{thin angular transition problem}. The whole-sphere case gives explicit Hamiltonian and momentum silhouettes. For thin transitions, we explain how angular scaling and the stability results of our companion papers enter the construction. The geometry and width are fixed before the masses are taken sufficiently small; the bounds may depend on these choices.

In Section~\ref{section-5}, we define the \emph{relative ADM energy-momentum} of a solution and its seed data set and prove that each data set admits (absolute) ADM charges if and only if the other does (Proposition~\ref{prop:ADM-relation}). We modify the seed data set by terms with harmonic decay, which enable us to choose the relative energy-momentum (Theorem~\ref{thm:vector-selection}).

Next, in Section~\ref{section-ADM-prescription}, we use compactly-supported transverse-traceless perturbations, and their effect under the seed-to-solution map, to achieve any sufficiently small \emph{total ADM energy-momentum} in any fixed future timelike direction (Theorem~\ref{thm:effective-ADM}); this ADM charge arises from quadratic nonlinearities in the perturbations. For a fixed background vacuum data set, we use such perturbations to choose the relative ADM energy-momentum under an explicit smallness condition relating the background to the target ADM charge; when the metric decay exponent exceeds one half, we can instead choose the total ADM charge (Proposition~\ref{prop:fixed-background-ADM}). For the gluing of Schwarzschild metrics, we retain the prescribed Schwarzschild mass parameters in exterior angular region while achieving any total ADM energy (under some smallness condition); zero or nonzero total momentum can be obtained as described in Proposition~\ref{prop:Sch-mass-prescribed}.

Appendix~\ref{appendix:local-projection} justifies the projection estimates in the local weighted H\"older norms used here and gives the dependence on the seed data sets and their parameters needed for controlling the relative and total ADM charges. Appendix~\ref{appendix:radial-selection} proves the radial divergence and momentum identities. It gives the remaining arguments for choosing the scalar harmonic parameter and for adjusting the total ADM energy when gluing Schwarzschild metrics.
Appendix~\ref{appendix:quadratic-curvature} derives the nonlinear scalar-curvature correction and the quadratic ADM energy-momentum formulas. It also constructs the compact transverse-traceless profiles used to vary the four total charges.

%============================================================================

\section{Schwarzschild seed data sets with Hamiltonian errors of any finite decay order}\label{section-2}

%-----------------------------------------------------------------------------------------------------------------------------------------

\subsection{Schwarzschild--Minkowski seed data sets}\label{section-2-1}
We construct in Proposition~\ref{prop:finite-expansion} asymptotically flat seed metrics~$\seedg$ outside a ball (equivalently, time-symmetric seed data sets) that interpolate between Schwarzschild metrics and whose scalar curvature has any prescribed radial decay $R_{\seedg}=\Obig_k(r^{-p_A-2})$. This is extended to all of~$\RR^3$ in Section~\ref{section-2-4}.
We begin with Schwarzschild--Minkowski data, derive radial correction formulas, and extend the construction to two Schwarzschild regions with small masses of either sign.

We thus start with an explicit interpolation between spatial Schwarzschild and Euclidean metrics. As we prove, an additional mixed tensor cancels the leading Hamiltonian error while retaining both prescribed exterior metrics. We fix $R>0$, a connected smooth angular domain $\Lambda\subset S^2$, whose complement is partitioned into two disjoint closed angular regions $F_0\sqcup F_1 = S^2\setminus\Lambda$ (the construction easily extends to more connected components). We choose the angular interpolation function $\chi$ so that
\be
\label{eq:angular-interpolation}
\chi\in C^\infty(S^2,[0,1]),\quad \chi=1\text{ near }F_1,\quad \chi=0\text{ near }F_0,\quad \supp d\chi\Subset\Lambda.
\ee
On the round sphere, we use a slash for the metric and differential operators: $\gslash$ denotes the round metric, $\nablaslash$ its covariant derivative, and $\Deltaslash$ its Laplacian. For a scalar $f$, $\nablaslash f$ denotes the tangential gradient vector and $(\nablaslash f)^\flat$ its dual one-form. We extend angular tensors to $\RR^3\setminus\{0\}$ by pull-back along the radial projection $\RR^3\setminus\{0\}\to S^2$. For $x\neq0$, we write $\theta\coloneqq x/r\in S^2$, where $r=|x|$. Thus $\theta$ is the radial unit vector, and also the point of $S^2$ at which angular functions are evaluated. Its Euclidean dual one-form is $\theta^\flat\coloneqq\ebf(\theta,\cdot)=dr=\sum_i\theta_i\,dx^i$. The symbols $\flat$ and $\sharp$ denote Euclidean lowering and raising of indices. In Cartesian component calculations we identify covariant and contravariant tensors using $\ebf$; in geometric tensor products we distinguish vectors from one-forms. We then define
\be
\label{eq:angular-mixed}
c_\chi\coloneqq\theta^\flat\otimes(\nablaslash\chi)^\flat+(\nablaslash\chi)^\flat\otimes\theta^\flat.
\ee
For $r\geq R$, we consider the exterior metric of the seed data set
\be
\label{eq:basic-exterior}
g_m^{[1]}\coloneqq\Bigl(1+\frac{m\chi}{2r}\Bigr)^4\ebf+\frac{2m}{r}c_\chi,
\ee
which, by~\eqref{eq:angular-interpolation}, coincides exactly with $g_m^{\mathrm{Sch}}$ near $F_1$ and with $\ebf$ near $F_0$. It is positive definite for $|m|/R$ sufficiently small, with a threshold depending on the fixed function~$\chi$. The term involving $c_\chi$ is supported inside the angular transition region.

The first term of~\eqref{eq:basic-exterior}, without $c_\chi$ correction, has scalar curvature
\be
\label{eq:conformal-error}
R_{(1+m\chi/(2r))^4\ebf} =-\frac{4m}{r^3}\Bigl(1+\frac{m\chi}{2r}\Bigr)^{-5}\Deltaslash\chi.
\ee
Its leading error has order $r^{-3}$ and generally fails the weighted square-integrability condition at accuracy $p_A\geq1$ (Section~\ref{section-3-2}, below). The mixed term in~\eqref{eq:basic-exterior} is chosen to cancel this error.

To find out this cancellation, we write, for a metric perturbation~$\gamma$, the linearized operator~\cite{BLF--PLF,BLF--PLF-stability}
\be
\label{eq:metric-linearization}
L \gamma\coloneqq DR_{\ebf}[\gamma]=\del_i\del_j \gamma_{ij}-\Delta\Tr_{\ebf}\gamma.
\ee
We will prove in Section~\ref{section-2-3} that the first-order metric coefficient $G_1\coloneqq\frac{2}{r}(\chi\ebf+c_\chi)$ satisfies
\be
\label{eq:H1}
LG_1=L\bigl((2/r)(\chi\ebf+c_\chi)\bigr)=0.
\ee
Nonlinear terms are of higher order in~$r$ and quadratic in the mass~$m$.
Consequently, for every fixed derivative order $k$, the exterior metric~\eqref{eq:basic-exterior} satisfies
\be
\label{eq:basic-error}
|\del^k R_{g_m^{[1]}}|\leq C_k m^2r^{-4-k},\qquad r\geq R,
\ee
provided $|m|/R$ lies in a fixed sufficiently small interval. The constants $C_k>0$ depend on derivatives of $\chi$ and on that interval. Thus the first corrected seed data set already allows every accuracy exponent $1<p_A<2$.

%-----------------------------------------------------------------------------------------------------------------------------------------

\subsection{Schwarzschild--Schwarzschild seed data sets}\label{section:hamiltonian-main}
We now state Proposition~\ref{prop:finite-expansion} on seed metrics of arbitrary accuracy~$p_A$ interpolating between two Schwarzschild regions with signed masses. The radial identities in Section~\ref{section-2-3}, below, give its proof and extend the construction to slow leading decay.

We generalize the previous discussion to the case of two Schwarzschild metrics.
With the radius~$R$, angular regions $F_0,F_1$, and function $\chi$ satisfying~\eqref{eq:angular-interpolation} as before,
we introduce the \emph{logarithmic variable} by $\ell\coloneqq\log(r/R)$ and define the two-tensor
\be
\label{eq:radial-corrector}
\Crad\coloneqq\theta^\flat\otimes\theta^\flat-\tfrac12\ebf.
\ee
For $m_1,m_2\in\RR$, we define
\be
\label{eq:two-mass-profile}
\begin{gathered}
M(\theta)\coloneqq m_1+(m_2-m_1)\chi(\theta),\qquad \mu\coloneqq\frac{|m_1|+|m_2|}{R},\\
M=m_1\text{ near }F_0,\qquad M=m_2\text{ near }F_1,\qquad \supp dM\subseteq\supp d\chi\Subset\Lambda.
\end{gathered}
\ee
The first exterior metric is
\be
\label{eq:two-mass-exterior}
g_{m_1,m_2}^{[1]}\coloneqq \Bigl(1+\frac{M}{2r}\Bigr)^4\ebf+\frac{2(m_2-m_1)}r c_\chi.
\ee
It interpolates between $g_{m_1}^{\mathrm{Sch}}$ where $\chi=0$ and $g_{m_2}^{\mathrm{Sch}}$ where $\chi=1$. Its linear term is $2m_1r^{-1}\ebf+(m_2-m_1)G_1$, which lies in the kernel of $L$ by~\eqref{eq:schwarzschild-components}, below, and $\Delta(r^{-1})=0$ away from $r=0$. Thus the leading cancellation holds for any $m_1,m_2\in\RR$, including zero. Setting $m_1=0$ and $m_2=m$ in~\eqref{eq:two-mass-exterior} gives exactly the metric $g_m^{[1]}$ for Schwarzschild--Minkowski gluing defined in~\eqref{eq:basic-exterior}.
We impose the condition $\mu_0<1$ (hence $\mu<1$) to ensure that black hole horizons (for positive mass) and conformal zeros (for negative mass) lie inside the excluded ball.

\begin{proposition}[Finite Hamiltonian accuracy for two Schwarzschild masses]\label{prop:finite-expansion}
Fix a radius $R>0$, an angular function $\chi$ satisfying~\eqref{eq:angular-interpolation}, and an integer $J\geq1$. There exist functions $P_j(m_1,m_2;\ell,\theta)$, homogeneous polynomials of total degree $j$ in the masses and polynomials in $\ell=\log(r/R)$, with smooth angular coefficients supported in the fixed compact set $\supp d\chi\Subset\Lambda$, such that
\be
\label{eq:finite-expansion}
g_{m_1,m_2}^{[J]}\coloneqq g_{m_1,m_2}^{[1]}+ \sum_{j=2}^J r^{-j}P_j(m_1,m_2;\ell,\theta)\Crad
\ee
has the following properties. There exists $\mu_0\in(0,1)$ such that all these metrics are positive definite for $\mu<\mu_0$, agree with the prescribed Schwarzschild metrics outside the transition region, and their scalar curvatures satisfy, for every fixed derivative order $k$,
\be
\label{eq:finite-error}
|\del^kR_{g_{m_1,m_2}^{[J]}}| \leq C_{J,k}(|m_1|+|m_2|)^{J+1}r^{-J-3-k}(1+\ell)^{K_{J,k}}, \qquad r\geq R.
\ee
Here $K_{J,k}$ is a finite nonnegative integer, and the constants are independent of the mass pair in this neighborhood. Each correction vanishes on $m_1=m_2$. For $m_1=0$ and $m_2=m$, one obtains a Schwarzschild--Minkowski family, denoted by $g_m^{[J]}\coloneqq g_{0,m}^{[J]}$, whose coefficients satisfy $P_j(0,m;\ell,\theta)=m^jP_j(0,1;\ell,\theta)$ by homogeneity. For any prescribed finite accuracy exponent $p_A>1$, choose $J$ with $J+1>p_A$. The estimate above then gives $\del^kR_{g_{m_1,m_2}^{[J]}}=o(r^{-p_A-2-k})$ for every fixed $k$, since the strict inequality absorbs the powers of $\log(r/R)$.
\end{proposition}

%-----------------------------------------------------------------------------------------------------------------------------------------

\subsection{Radial equations for the scalar curvature constraint}\label{section-2-3}
We now give the proof of Proposition~\ref{prop:finite-expansion}. The idea is to iteratively solve the linearized equtions with nonlinear terms treated as sources. We express the linearized Hamiltonian operator in radial and tangential components. This decomposition gives radial inverses that preserve angular support, including at harmonic and resonant exponents. These inverses also give data sets with slow leading decay.

\subsubsection{Spherical decomposition and radial exponents.}
For a symmetric covariant metric perturbation $\gamma$, we introduce its trace reversal and its spherical components by
\be
\label{eq:trace-reversal}
A\coloneqq \gamma-(\Tr_{\ebf}\gamma)\ebf =a\,\theta^\flat\otimes\theta^\flat+\theta^\flat\otimes b+b\otimes\theta^\flat+C,
\ee
motivated by the fact that $L\gamma$ is the double divergence of~$A$.
The inverse relation is $\gamma=A-\frac12(\Tr_{\ebf}A)\ebf$. Here $a=A(\theta,\theta)$ is a scalar, $b$ is a tangential one-form, and $C$ is a symmetric covariant two-tensor on $S^2$. For a tangent vector $V$, the mixed component is $b(V)=A(\theta,V)$. The products in~\eqref{eq:trace-reversal} are tensor products of one-forms. Their Cartesian extensions satisfy $\theta_i b_i=\theta_iC_{ij}=0$.

We use the logarithmic derivative $\del_\ell = r\del_r$. Recall that the operators $\divs$, $\nablaslash$, and $\trs$ are the covariant divergence, covariant derivative, and trace of the round metric. They act on tangential fields; in particular, $\divs\divs C$ uses covariant differentiation rather than differentiating Cartesian tensor components as scalar functions.
The proof of Lemma~\ref{lem:double-divergence} is given in \hyperref[appendix:spherical-divergence]{Appendix~\ref*{appendix:radial-selection}}.

\begin{lemma}\label{lem:double-divergence}
For the decomposition~\eqref{eq:trace-reversal}, one has
\be
\label{eq:spherical-L}
r^2L \gamma=(\del_\ell+2)(\del_\ell+1)a+2(\del_\ell+2)\divs b+\divs\divs C-(\del_\ell+1)\trs C.
\ee
\end{lemma}

We use leading tensors with a fixed radial power $r^{-q}=R^{-q}e^{-q\ell}$: their Cartesian components have the form $r^{-q}\gamma_0(\theta)$, with $\gamma_0$ depending only on the angular variable. Polynomial degrees in the masses are determined later on. In dimension three, components $(a,b,C)=r^{-q}(a_0,b_0,C_0)$ satisfy $L \gamma=0$ exactly when
\be
\label{eq:homogeneous-constraint}
(2-q)(1-q)a_0+2(2-q)\divs b_0+\divs\divs C_0+(q-1)\trs C_0=0.
\ee
For $q\notin\{1,2\}$, arbitrary smooth tangential profiles $b_0,C_0$ therefore determine
\be
\label{eq:slow-parametrization}
a_0=-\frac{2(2-q)\divs b_0+\divs\divs C_0+(q-1)\trs C_0}{(2-q)(1-q)}.
\ee
The construction preserves angular support. For $q\in(0,+\infty)\setminus\{1,2\}$, the metric perturbation $\gamma=r^{-q}\gamma_0(\theta)$ with this value of~$a_0$ gives $|\del^k R_{\ebf+t \gamma}|=\Obig(t^2\,r^{-2q-2-k})$ for any derivative order $k\geq 0$.

\subsubsection{Compatibility at the harmonic exponent.} At $q=1$, the radial component $a_0$ is free and the tangential components must satisfy
\be
\label{eq:harmonic-compatibility}
\divs(2b_0+\divs C_0)=0.
\ee
The trace reversal of $r\,G_1$, defined just before the equation~\eqref{eq:H1}, has components
\be
\label{eq:schwarzschild-components}
a_0=-4\chi,\qquad b_0=2(\nablaslash\chi)^\flat,\qquad C_0=-4\chi\gslash.
\ee
Since $\divs(\chi\gslash)=(\nablaslash\chi)^\flat$, these components satisfy~\eqref{eq:harmonic-compatibility}. This proves $LG_1=0$ and hence~\eqref{eq:basic-error}. At $q=2$, the compatibility condition instead reads $\divs\divs C_0+\trs C_0=0$. For these two exponents, if the conditions on $b_0,C_0$ (possibly with an inhomogeneous source) are not satisfied, one must include a logarithmic radial term in the scalar coefficient~$a$.

\subsubsection{Radial inversion and resonant exponents.}
The term $a\theta^\flat\otimes\theta^\flat$ in~\eqref{eq:trace-reversal} is the trace reversal of $a\Crad$.  This motivates the following approach to invert the linearized scalar curvature.
The radial inverse preserves angular support because it is an ordinary differential equation at each angular point.

\begin{lemma}[Radial inversion of the Hamiltonian constraint]\label{lem:Hamiltonian-radial}
Fix a scalar source~$E$ that depends on the radius~$r$ and angle~$\theta$.  The inhomogeneous equation $r^2L\gamma=E$ admits a family of solutions $\gamma=a\Crad$ with a scalar coefficient characterized by
\be
\label{eq:rdr2aE}
(r\del_r+2)(r\del_r+1)a = E , \qquad
a(r,\theta) = \int_R^r \frac{r-r'}{r^2} E(r',\theta) dr' + a_1(\theta) r^{-1} + a_2(\theta) r^{-2} ,
\ee
for some angular functions~$a_1,a_2$.
If $E=r^{-q}E_0(\ell,\theta)$ where $E_0$ is polynomial in $\ell=\log(r/R)$, then $a_1,a_2$ can be chosen such that $a=r^{-q}a_0(\ell,\theta)$ with $a_0$ polynomial in~$\ell$ of the same degree when $q\not\in\{1,2\}$ and otherwise of one higher degree and with vanishing constant coefficient.
\end{lemma}

\begin{proof}
The first relation in~\eqref{eq:rdr2aE} is a direct application of~\eqref{eq:spherical-L} with $b=0$ and $C=0$.
Its general solution is given as the second relation in~\eqref{eq:rdr2aE}: the homogeneous part is obvious while the inhomogeneous part is checked by direct computation,
\be
(r\del_r+2)(r\del_r+1) \int_R^r \frac{r-r'}{r^2} E(r',\theta) dr'
= (r\del_r+2) \int_R^r \frac{r'}{r^2} E(r',\theta) dr'
= E(r,\theta) .
\ee
For the remaining statements, it is most convenient to seek~$a$ directly in the form $a=r^{-q}a_0(\ell,\theta)$ and observe that
\be
\label{eq:radial-inverse}
L\bigl(r^{-q}a_0(\ell,\theta)\Crad\bigr) =r^{-q-2}(\del_\ell+2-q)(\del_\ell+1-q)a_0.
\ee
For $q\not\in\{1,2\}$ the linear operator $P\mapsto (\del_\ell+2-q)(\del_\ell+1-q)P$ on the space of polynomials of any given maximum degree is invertible, while for $q\in\{1,2\}$ it is surjective with a kernel consisting of polynomials of degree zero, namely functions of the angular variables.
\end{proof}

\begin{proof}[Proof of Proposition~\ref{prop:finite-expansion}]
We organize the expansion as a series in powers of the masses~$m_1,m_2$, or equivalently (by dimensional analysis) inverse powers of the radius~$r$.
We use the invertibility in Lemma~\ref{lem:Hamiltonian-radial} to cancel the scalar-curvature error successively at each order.

At total mass degree $j$, expansion of the inverse metric, Christoffel symbols, and scalar curvature gives an error $r^{-j-2}E_j(m_1,m_2;\ell,\theta)$. The coefficient is homogeneous of degree $j$ in the masses and polynomial in $\ell$. The cancellation above starts the induction at degree two. At each further step we solve
\be
\label{eq:polynomial-recursion}
(\del_\ell+2-j)(\del_\ell+1-j)P_j=-E_j.
\ee
At $j=2$, we invert $\del_\ell-1$ on polynomials and integrate~$\del_\ell$, fixing the free constant to zero. This step generates terms linear in~$\ell$; nonlinearities in the expansions of the scalar curvature then lead to higher-degree polynomials.
For $j\geq3$, both constant terms $2-j$ and $1-j$ in~\eqref{eq:polynomial-recursion} are nonzero, giving an inverse on each bounded-degree polynomial space. Applied to each mass monomial, this inverse preserves homogeneity, smoothness, and support, and~\eqref{eq:radial-inverse} cancels the degree-$j$ error. Where $d\chi=0$ locally, the metric is Schwarzschild with constant mass $M$, so all correction terms~$P_j$ vanish.  Continuity then ensures that all corrections are supported in $\supp d\chi$.

With $r/R\geq1$, the initial perturbation is bounded by $C\mu R/r$ and each correction and its scaled derivatives by $C(\mu R/r)^j(1+\ell)^K$. These functions are bounded for $r\geq R$. For fixed maximal order~$J$, a single $\mu_0$ makes every intermediate metric uniformly equivalent to~$\ebf$, for any mass with $\mu=(|m_1|+|m_2|)/R<\mu_0$. Taylor's formula along $\tau\mapsto(\tau m_1,\tau m_2)$ bounds the remainder $\del^kR_{g_{m_1,m_2}^{[J]}}$ by $C_{J,k}(R\mu)^{J+1}r^{-(J+3+k)}(1+\ell)^{K_{J,k}}$. Indeed, each mass factor supplies an inverse radial power, the scalar curvature supplies two derivatives, and differentiating $k$~times supplies~$r^{-k}$. The inverse metric is uniformly bounded throughout the Taylor segment. This proves~\eqref{eq:finite-error}, including when one or both masses vanish.

For $m_1=m_2$, the initial metric is Schwarzschild throughout the exterior, which is an exact solution, so induction gives $P_j(m_1,m_1;\ell,\theta)=0$.
\end{proof}

%-----------------------------------------------------------------------------------------------------------------------------------------

\subsection{Extension of the seed data sets through the interior}\label{section-2-4}
We extend the exterior seed data sets through a fixed ball using a radial cutoff~$\zeta_R$. This cutoff introduces an additional compactly supported constraint error. We fix the radius~$R$, angular region~$\Lambda$, and cutoff functions $\chi$ and~$\zeta_R$ first. The compactly supported constraint error can then be taken arbitrarily small in any given weighted norm by taking sufficiently small masses.

We fix a smooth radial cutoff $\zeta_R$ equal to zero for $r\leq R$ and to one for $r\geq2R$, with $0\leq\zeta_R\leq1$ and $|\del^k\zeta_R|\leq C_kR^{-k}$. The global seed data set is
\be
\label{eq:global-seed}
\seedg(m_1,m_2)\coloneqq\ebf+\zeta_R(g_{m_1,m_2}^{[J]}-\ebf),\qquad \seedh(m_1,m_2)\coloneqq0,
\ee
interpolating between $\ebf$ for $r\leq R$ and $g_{m_1,m_2}^{[J]}$ for $r\geq 2R$. We write $\seedg(m)\coloneqq\seedg(0,m)$ and $\seedh(m)\coloneqq\seedh(0,m)$. We choose a smooth connected domain $\Omega$ that contains the closed ball~$\overline B_{3R}$ and the (open) cone over~$\Lambda$ and coincides with it for $r\geq4R$.  In particular it contains the cone over $\supp d\chi$. This puts every radial cutoff error compactly inside~$\Omega$.

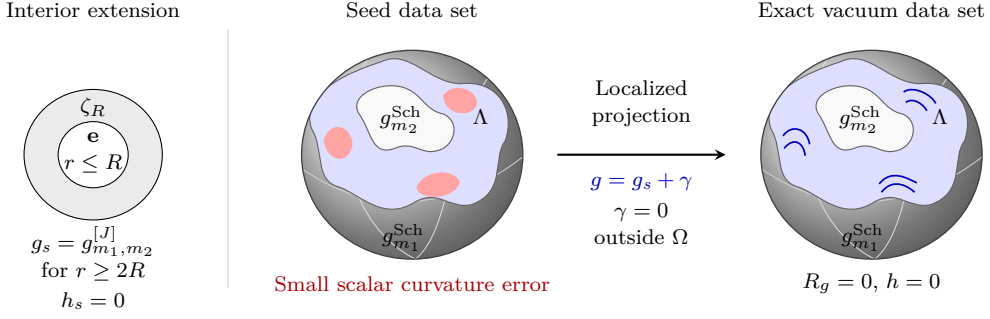
\begin{figure}[tb]
\centering
\begin{tikzpicture}[x=1.06cm,y=1cm,font=\footnotesize,line width=.4pt]
\useasboundingbox (-.03,-2.17) rectangle (12.65,2.12);
\def\outerpatch{(-.91,.25) .. controls (-1.02,.60) and (-.55,.80) .. (-.28,.89)
 .. controls (-.03,1.02) and (.16,.68) .. (.44,.68)
 .. controls (.77,.72) and (.94,.43) .. (.84,.13)
 .. controls (.72,-.03) and (.86,-.36) .. (.53,-.48)
 .. controls (.33,-.52) and (.24,-.33) .. (.03,-.45)
 .. controls (-.26,-.65) and (-.47,-.38) .. (-.66,-.31)
 .. controls (-.89,-.25) and (-.79,.04) .. (-.91,.25) -- cycle}
\def\innerpatch{(-.48,.39) .. controls (-.58,.56) and (-.31,.70) .. (-.10,.65)
 .. controls (.09,.60) and (.10,.43) .. (.28,.39)
 .. controls (.44,.32) and (.39,.08) .. (.23,.04)
 .. controls (.08,-.01) and (.02,.14) .. (-.16,.06)
 .. controls (-.37,-.02) and (-.37,.23) .. (-.48,.39) -- cycle}
\def\angularsphere{%
 \shade[ball color=gray!10] (0,0) circle (1);
 \begin{scope}
  \clip (0,0) circle (1);
  \draw[black!18,line width=.22pt] (-1,-.13) .. controls (-.35,-.63) and (.42,-.51) .. (1,-.02);
  \draw[black!13,line width=.22pt] (-.40,.92) .. controls (-.80,.22) and (-.55,-.70) .. (.03,-1);
  \draw[black!13,line width=.22pt] (.68,.73) .. controls (.31,.14) and (.36,-.71) .. (.03,-1);
  \path[fill=blue!13,draw=black!62,line width=.42pt] \outerpatch;
  \path[fill=gray!5,draw=black!62,line width=.42pt] \innerpatch;
 \end{scope}
 \draw[black!75,line width=.5pt] (0,0) circle (1);
 \node at (-.10,.34) {$g^{\mathrm{Sch}}_{m_2}$};
 \node at (-.08,-.79) {$g^{\mathrm{Sch}}_{m_1}$};
 \node at (.62,.37) {$\Lambda$};
}
\node at (1.24,1.92) {Interior extension};
\begin{scope}[shift={(1.24,0)},scale=1.15]
 \fill[gray!16] (0,0) circle (.75);
 \fill[white] (0,0) circle (.38);
 \draw (0,0) circle (.38);
 \draw (0,0) circle (.75);
 \node at (0,.2) {$\ebf$};
 \node at (0,.55) {$\zeta_R$};
 \node at (0,-.1) {$r\leq R$};
\end{scope}
\node[align=center] at (1.24,-1.30) {$\seedg=g^{[J]}_{m_1,m_2}$\\for $r\geq2R$};
\node at (1.24,-1.93) {$\seedh=0$};
\draw[black!20] (2.92,-1.76)--(2.92,1.65);
\node at (5.22,1.92) {Seed data set};
\begin{scope}[shift={(5.22,0)},scale=1.38]
 \angularsphere
 \begin{scope}
  \clip \outerpatch;
  \path[fill=red!36] (-.80,.14) .. controls (-.74,.34) and (-.55,.28) .. (-.54,.10)
   .. controls (-.57,-.12) and (-.78,-.06) .. (-.80,.14) -- cycle;
  \path[fill=red!36] (.28,.57) .. controls (.39,.72) and (.62,.60) .. (.58,.47)
   .. controls (.52,.34) and (.28,.39) .. (.28,.57) -- cycle;
  \path[fill=red!36] (.04,-.31) .. controls (.10,-.15) and (.37,-.12) .. (.43,-.26)
   .. controls (.40,-.43) and (.05,-.45) .. (.04,-.31) -- cycle;
 \end{scope}
\end{scope}
\node[red!65!black] at (5.22,-1.70) {Small scalar curvature error};
\node at (10.95,1.92) {Exact vacuum data set};
\begin{scope}[shift={(10.95,0)},scale=1.38]
 \angularsphere
 \begin{scope}
  \clip \outerpatch;
  \draw[blue!75!black,line width=.65pt] (-.80,.18) .. controls (-.72,.30) and (-.58,.25) .. (-.56,.10);
  \draw[blue!75!black,line width=.65pt] (-.79,.06) .. controls (-.70,.18) and (-.61,.11) .. (-.60,-.01);
  \draw[blue!75!black,line width=.65pt] (.29,.60) .. controls (.38,.67) and (.55,.57) .. (.58,.49);
  \draw[blue!75!black,line width=.65pt] (.31,.50) .. controls (.38,.56) and (.48,.47) .. (.50,.42);
  \draw[blue!75!black,line width=.65pt] (.07,-.28) .. controls (.18,-.13) and (.33,-.23) .. (.41,-.27);
  \draw[blue!75!black,line width=.65pt] (.09,-.36) .. controls (.20,-.21) and (.34,-.31) .. (.39,-.36);
 \end{scope}
\end{scope}
\node at (10.95,-1.70) {$R_g=0$, $h=0$};
\draw[-stealth,thick] (7.03,0)--(9.13,0);
\node[align=center] at (8.08,.68) {Localized\\projection};
\node[blue!75!black] at (8.08,-.36) {$g=\seedg+\gamma$};
\node[align=center] at (8.08,-.91) {$\gamma=0$\\outside $\Omega$};
\end{tikzpicture}
\caption{Left: radial interpolation between Euclidean metric~$\ebf$ and a finite approximation $g^{[J]}$.  Right: from a seed data set to an exact scalar-flat metric. The spheres show angular directions at infinity, with an irregular transition region $\Lambda$. Red marks stand for angular inhomogeneities in the scalar curvature error, which are reflected in an inhomogeneous correcting tensor, depicted by blue curves. The drawing is schematic.}
\label{fig:seed-geometry}
\end{figure}

The metric of the seed data set is smooth on $\RR^3$, positive definite for $\mu$ small, and satisfies the vacuum constraints exactly on $\RR^3\setminus\Omega$. In the compact annulus $R\leq r\leq2R$, its scalar curvature satisfies
\be
\label{eq:compact-error}
|\del^kR_{\seedg(m_1,m_2)}|\leq C_{J,k}R^{-2-k}\mu.
\ee
\begin{itemize}
\item \emph{Size of the metric perturbation.} On the annulus, the order-$j$ derivatives of $\seedg-\ebf$ are bounded by $C\mu R^{-j}$, including the derivatives of $\zeta_R$. The inverse metric remains uniformly bounded for sufficiently small $\mu$. Positivity follows from the convex interpolation between the positive exterior metric and $\ebf$.

\item \emph{Size of the scalar curvature error.} The terms $g^{-1}g^{-1}\del^2g$ and $g^{-1}g^{-1}g^{-1}\del g\del g$ have sizes $R^{-2}\mu$ and $R^{-2}\mu^2$, respectively. Since $\mu\leq1$, their sum and its differentiated versions satisfy~\eqref{eq:compact-error}. Derivatives of the cutoff create a linear error in $\mu$, even when the exterior error has been reduced to degree $J+1$ as stated in~\eqref{eq:finite-error}. This annular error need only be small in the required weighted norms.

\item \emph{Order of the choices.} We fix the radius~$R$, angular geometry, profiles, cutoff, regularity, and finite accuracy $p_A$, then choose $J$ with $J+1>p_A$. The H\"older norm introduced below to measure the scalar curvature error includes a weight~$r^{p_A+2}$.  On the annulus, the curvature scale $R^{-2}\mu$ multiplied by the radial source weight $R^{p_A+2}$ gives $R^{p_A}\mu$. The relevant weighted $L^2$ norm has the same scale (Section~\ref{section-3-2}, below). We finally choose $|m_1|+|m_2|$ small enough for these norms to be small, as needed for the existence of a projection to an exact vacuum data set $(g,h)$. The threshold thus depends on $R$ and $p_A$ as well as the other fixed choices; the ratio $\mu$ alone does not give a bound uniform in the radius and accuracy.
\end{itemize}
\noindent Figure~\ref{fig:seed-geometry} shows the interior extension and the exterior regions retained by projection. The cutoff gives Euclidean data on $B_R$ and $g^{[J]}_{m_1,m_2}$ for $r\geq2R$. We choose $\Omega$ conical over $\Lambda$ for $r\geq4R$, so the prescribed Schwarzschild regions lie outside $\Omega$. In Theorem~\ref{thm:fixed-radius}, below, we will construct the exact scalar-flat metric $g=\seedg+\gamma$ with $\gamma=0$ outside $\Omega$. The blue correction then cancels the red scalar curvature error and preserves both exterior metrics.

%============================================================================

\section{Seed data sets with errors of any finite decay order in both Einstein constraints}\label{section:full-constraints}

%-----------------------------------------------------------------------------------------------------------------------------------------

\subsection{Finite accuracy for both constraints}\label{section:coupled-accuracy}
We now construct seed data sets with arbitrarily slow decay, satisfying both Einstein constraints to any prescribed finite order. The main result is Proposition~\ref{prop:coupled-expansion}; similarly to the previous section, its proof consists of repeatedly solving linearized constraint equations as radial differential equations that preserve angular support.

We prescribe the leading metric and momentum independently as perturbations of the Euclidean data $(\ebf,0)$, subject to their linearized constraints. The linearization of constraints at $(\ebf,0)$ is the block diagonal operator $\mathcal L_0(\gamma,\pi)\coloneqq(L \gamma,\operatorname{div}_{\ebf}\pi)$.
This set-up is generalized at the end of the section to polyhomogeneous perturbations and perturbations of the seed metrics of Section~\ref{section-2} that interpolate between Schwarzschild metrics.

\begin{proposition}[Finite asymptotic accuracy for the full constraints]\label{prop:coupled-expansion}
Fix $R>0$, $q\in(0,1)$, an integer $J\geq1$, and smooth angular profiles $\gamma_0,\pi_0$ supported in a fixed compact subset of~$\Lambda$, satisfying
\be
\label{eq:coupled-leading-kernel}
L(r^{-q}\gamma_0)=0,\qquad \operatorname{div}_{\ebf}(r^{-q-1}\pi_0)=0.
\ee
There are polynomials $P_j,a_j,b_j$ in~$\ell$, with smooth angular coefficients (respectively scalars, scalars, and tangent vectors) supported in the union of these profile supports, such that
\be
\label{eq:coupled-exterior}
\begin{aligned}
g_t^{[J]}&\coloneqq\ebf+t\,r^{-q}\gamma_0+{\textstyle\sum_{j=2}^J} t^j\,r^{-jq}P_j(\ell,\theta)\Crad,\\
h_t^{[J]}&\coloneqq t\,r^{-q-1}\pi_0+{\textstyle\sum_{j=2}^J} t^j\,r^{-jq-1}
(a_j\theta\otimes\theta+\theta\otimes b_j+b_j\otimes\theta)
\end{aligned}
\ee
has the following properties for $\mu_q=|t|/R^q$ sufficiently small.
The metric $g_t^{[J]}$ is positive definite in the region $s=r/R\geq1$, and for every fixed derivative order $k$ one has
\be
\label{eq:coupled-error}
|\del^k\mathcal H(g_t^{[J]},h_t^{[J]})|+|\del^k\mathcal C(g_t^{[J]},h_t^{[J]})| \leq C R^{-2-k}\mu_q^{J+1}s^{-(J+1)q-2-k}(1+\log s)^K.
\ee
Here $K$ is a finite nonnegative integer, and $C,K$ depend on the fixed profiles, $q,J,k$, and $R$, but not on~$t$. In particular, any prescribed finite accuracy~$p_A$ for both constraints can be achieved by taking $J+1>p_A/q$.
\end{proposition}

%-----------------------------------------------------------------------------------------------------------------------------------------

\subsection{Seed data sets with slow leading decay}\label{section:slow-hamiltonian-seeds}
For clarity, we specialize first to the time-symmetric case, where $\pi_0=0$ and all $a_j,b_j,h_t^{[J]}$ vanish.
The Schwarzschild construction in Section~\ref{section-2} required some care in constructing the leading order term as $q=1$ is a resonant exponent.  Here, $q\in(0,1)$ and the allowed leading term $r^{-q}\gamma_0$ is characterized by~\eqref{eq:slow-parametrization}, which expresses $a_0$ in terms of $b_0,C_0$ with a formula that respects the support of $b_0,C_0$.  As an example, a nonconstant function $\psi\in C_c^\infty(\Lambda)$ gives the divergence-free tangential one-form $b_0\coloneqq(\theta\times\nablaslash\psi)^\flat$, which can be completed with $C_0=0$ and $a_0=0$ to get the nonzero profile $\gamma_0\coloneqq\theta^\flat\otimes b_0+b_0\otimes\theta^\flat$.
We multiply the slow-decaying term~$r^{-q}\gamma_0$ by an amplitude parameter~$t$, whose dimensionless size at radius $R$ is $\mu_q\coloneqq|t|/R^q$.

Our aim is to construct polynomials $P_{j,q}(\ell,\theta)$ in $\ell=\log(r/R)$ with smooth angular coefficients supported in $\supp \gamma_0$ such that the specialization of~\eqref{eq:coupled-error} holds, namely
\be
\label{eq:slow-exterior-error}
|\del^k R_{g_{t,q}^{[J]}}| \leq C_{q,J,k}R^{-2-k}\mu_q^{J+1}s^{-(J+1)q-2-k}(1+\log s)^K,
\ee
where $s=r/R\geq1$ and $K$ is a finite nonnegative integer. The metric equals $\ebf$ outside the cone over $\supp \gamma_0$.

\begin{proof}[Proof of Proposition~\ref{prop:coupled-expansion}: time-symmetric case]
At order $j$ in~$t$, the scalar-curvature error has the form $t^j\,r^{-jq-2}E_{j,q}(\ell,\theta)$. We determine $P_{j,q}$ by solving
\be
\label{eq:slow-polynomial-recursion}
(\del_\ell+2-jq)(\del_\ell+1-jq)P_{j,q}=-E_{j,q},
\ee
using Lemma~\ref{lem:Hamiltonian-radial} with $q$ replaced by~$jq$.
For $jq\notin\{1,2\}$ the differential operator is invertible on spaces of polynomials of any given maximum degree. At either resonance, integration with zero free constant increases the degree by at most one. Thus every $q>0$ is allowed and angular support is preserved. We apply Taylor's formula to $\tau\mapsto R_{g_{\tau t,q}^{[J]}}$, $0\leq\tau\leq1$, as in the last paragraph of the proof of Proposition~\ref{prop:finite-expansion}. Since $s^{-jq}(1+\log s)^K$ is bounded, sufficiently small $|t|$ gives uniform positivity throughout this segment. The recursion cancels all amplitude degrees through $J$. The remainder therefore has at least $J+1$ factors of $t\,r^{-q}$, while scalar curvature and its $k$ derivatives contribute $r^{-2-k}$. The polynomial coefficients give the logarithmic factor in~\eqref{eq:slow-exterior-error}; substituting $|t|\,r^{-q}=\mu_qs^{-q}$ proves the estimate.
\end{proof}

%-----------------------------------------------------------------------------------------------------------------------------------------

\subsection{Radial inversion of the momentum constraint}\label{section-2-coupled}
We first solve the Euclidean momentum equation at a fixed radial order. The polynomial inverse below also covers the two resonant orders and supplies the momentum corrections in Proposition~\ref{prop:coupled-expansion}.

\begin{lemma}[Radial inversion of the momentum constraint]\label{lem:momentum-radial}
Fix $q\in(0,1)$ and a smooth symmetric contravariant tangential tensor $C_0$ compactly supported in~$\Lambda$. Define
\be
\label{eq:momentum-profile}
b_0\coloneqq-\frac{\divs C_0}{2-q},\qquad a_0\coloneqq\frac{\trs C_0-\divs b_0}{1-q},\qquad \pi_0\coloneqq a_0\theta\otimes\theta+\theta\otimes b_0+b_0\otimes\theta+C_0.
\ee
Then $\operatorname{div}_{\ebf}(r^{-q-1}\pi_0)=0$, and $\pi_0$ has the same support as~$C_0$. For every $q>0$ and polynomial vector source $F(\ell,\theta)=F_r\theta+F_{\mathrm{tan}}$, compactly supported in~$\Lambda$, there are scalar and tangential polynomials $a,b$ supported in the angular support of $F$ such that
\be
\label{eq:momentum-inhomogeneous}
\operatorname{div}_{\ebf}\bigl(r^{-q-1}(a\theta\otimes\theta+\theta\otimes b+b\otimes\theta)\bigr) =r^{-q-2}F.
\ee
\end{lemma}
For a tensor whose Cartesian components have the form $r^{-q-1}\pi_0(\theta)$, the radial and tangential divergence coefficients are
\be
\label{eq:momentum-homogeneous}
(1-q)a_0+\divs b_0-\trs C_0,\qquad (2-q)b_0+\divs C_0.
\ee
For the inhomogeneous problem, we set the tangential two-tensor to zero and solve successively
\be
\label{eq:momentum-polynomial-inverse}
(\del_\ell+2-q)b=F_{\mathrm{tan}},\qquad (\del_\ell+1-q)a=F_r-\divs b.
\ee
The proof of Lemma~\ref{lem:momentum-radial} is given in \hyperref[appendix:radial-momentum]{Appendix~\ref*{appendix:radial-selection}}.

%-----------------------------------------------------------------------------------------------------------------------------------------

\subsection{Construction and independent amplitudes}\label{section:coupled-construction}
We use the inverses of linearized scalar-curvature and momentum in Lemmas~\ref{lem:Hamiltonian-radial} and~\ref{lem:momentum-radial} to cancel successive terms of the nonlinear constraints.

\begin{proof}[Proof of Proposition~\ref{prop:coupled-expansion}]
We denote the Christoffel symbols of $g$ by $\Gamma^i{}_{jk}(g)$. In Cartesian coordinates the exact raised momentum equation is
\be
\label{eq:momentum-exact}
\mathcal C^j(g,h)=\del_i h^{ij}+\Gamma^i{}_{i\ell}(g)h^{\ell j}+\Gamma^j{}_{i\ell}(g)h^{i\ell}.
\ee
The Hamiltonian terms involving $h$ are quadratic in~$h$, and the last two terms in~\eqref{eq:momentum-exact} are at least quadratic in $(g-\ebf,h)$. Thus the coefficient at amplitude degree $t^j$, before the degree-$j$ correction $P_j,a_j,b_j$ is included, has the form $\bigl(\mathcal H,\mathcal C\bigr)_{[j]}=t^j\,r^{-jq-2}(E_j,F_j)$, where $E_j,F_j$ are polynomials in~$\ell$ with smooth angular coefficients. This statement follows inductively from~\eqref{eq:full-constraints} and~\eqref{eq:momentum-exact}: an $h$ factor has one additional inverse radial power, and each differentiated metric factor contributes one inverse radial power per derivative. Products therefore have this radial order.

We solve~\eqref{eq:radial-inverse} with radial exponent $jq$ and source $-E_j$, and~\eqref{eq:momentum-polynomial-inverse} with source $-F_j$, using the two invertibility lemmas to select corrections $P_j,a_j,b_j$ of order~$j$ in the amplitude parameter~$t$. The block diagonal linearization cancels both degree-$j$ errors; further nonlinear contributions start at degree $j+1$. The inverses cover every resonance and preserve support, since all error coefficients vanish where the data equal $(\ebf,0)$.

For fixed finite order, $s^{-jq}(1+\log s)^K$ is bounded, and small $\mu$ gives uniform metric invertibility. Taylor's formula for $g^{-1}$ and the constraints~\eqref{eq:full-constraints}, \eqref{eq:momentum-exact} then gives the bounds~\eqref{eq:coupled-error}.
\end{proof}

Observe that the leading metric and momentum amplitudes are independent and moreover the same arguments permit a superposition of a finite number of independent leading profiles with different radial exponents.
Finitely many positive exponents have only finitely many nonnegative integer combinations below each fixed exponent threshold.
For finitely many amplitudes $t_\ell$ and fixed profiles satisfying the linearized constraints at their respective exponents $q_\ell>0$, we set $q_{\min}\coloneqq\min_\ell q_\ell$. At multi-index $\beta$, the metric exponent is $\sum_\ell\beta_\ell q_\ell$ and the momentum exponent is one larger, leading to the expansion
\be
\label{eq:coupled-exterior-multi}
\begin{aligned}
g_t^{[J]}&\coloneqq\ebf+{\textstyle\sum_\ell} t_\ell\,r^{-q_\ell}\gamma_0^\ell+{\textstyle\sum_{2\leq|\beta|\leq J}} {\textstyle\prod_\ell} (t_\ell r^{-q_\ell})^{\beta_\ell} P_\beta(\ell,\theta)\Crad,\\
h_t^{[J]}&\coloneqq {\textstyle\sum_\ell}t_\ell\,r^{-q_\ell-1}\pi_0^\ell+{\textstyle\sum_{2\leq|\beta|\leq J}} {\textstyle\prod_\ell} (t_\ell r^{-q_\ell})^{\beta_\ell} r^{-1} (a_\beta\theta\otimes\theta+\theta\otimes b_\beta+b_\beta\otimes\theta) .
\end{aligned}
\ee
The polynomials $P_\beta,a_\beta,b_\beta$ are obtained by inverting linearized constraint with the non-linear sources as in the single-profile case. Finitely many multi-indices occur at each degree, so~\eqref{eq:coupled-error} holds with $q=q_{\min}$ and $\mu=\sum_\ell|t_\ell|/R^{q_\ell}$.

The same slow leading term can be inserted into the Schwarzschild--Minkowski gluing construction (and more generally the gluing of two Schwarzschild metrics). We start from $g_m^{[1]}+t\,r^{-q}\gamma_0$, with $\supp \gamma_0\Subset\Lambda$, and cancel scalar curvature degree by degree in the two amplitudes. The corrected exterior metric has the finite form
\be
\label{eq:slow-Schwarzschild}
g_{m,t}^{[J]}=g_m^{[1]}+t\,r^{-q}\gamma_0+ \sum_{a=0}^J\sum_{b=\max\{0,2-a\}}^{J-a} t^a m^b r^{-aq-b}P_{a,b}(\ell,\theta)\Crad.
\ee
At a fixed pair $(a,b)$, the equation for $P_{a,b}$ is~\eqref{eq:radial-inverse} with radial exponent $aq+b$; logarithms handle the resonant values $1$ and~$2$. Every coefficient is supported in the union of the fixed angular transition and $\supp \gamma_0$. The metric therefore keeps the prescribed exact Schwarzschild and Minkowski values outside~$\Lambda$.

%-----------------------------------------------------------------------------------------------------------------------------------------

%============================================================================

\section{Complete scalar-flat metrics joining two prescribed Schwarzschild regions}\label{section-3}

%-----------------------------------------------------------------------------------------------------------------------------------------

The seed data sets of Section~\ref{section-2} interpolate between the prescribed Schwarzschild metrics, but their scalar curvature need not vanish in the gluing domain. We now verify the weighted hypotheses of the seed-to-solution projection results in \cite{BLF--PLF} and obtain complete scalar-flat metrics that match both exterior regions exactly.
Section~\ref{section:slow-coupled-projection}, below, deals with the seed data sets of Section~\ref{section:full-constraints}.

\subsection{Localization data sets and admissible exponents}\label{section-3-1}
We specify the localization domain, weights, and exponents required by the projection theorem applicable to the Schwarzschild seed data sets. Their geometric conditions are fixed independently of the amplitudes of the seed data sets. Theorem~\ref{thm:fixed-radius}, below, gives the resulting exact metrics.

We work on $\RR^3$ with reference data $(g_0,h_0)=(\ebf,0)$, fix $R,\Lambda,\chi,\Omega,\zeta_R$ as in Section~\ref{section-2}, fix a smooth radius with $\rho\simeq(R^2+r^2)^{1/2}$, and $\rho=r$ for $r\geq4R$, as well as a localization function $\lambda>0$ in $\Omega$ that vanishes linearly at its boundary, is extended by zero outside~$\Omega$, and depends only on the angular variable for $r\geq4R$. If $\Omega\neq\RR^3$, we require $\lambda\simeq\min(1,\operatorname{dist}(x,\Omega^c)/\rho)$ and $|r^j\del^j\lambda|\lesssim 1$. Since the support of the constraint error is compactly supported inside~$\Omega$, with moreover an angular part compactly supported inside~$\Lambda$, the function~$\lambda$ is bounded below on that support.

We fix H\"older exponents $N\geq3$, $0<\alpha<1$, and localization exponents satisfying the dimension-three specialization of the admissibility conditions in~\cite{BLF--PLF}:
\be
\label{eq:localization-exponents}
\max(N+\alpha,1+P/2)<\underline P<P<\overline P, \qquad \overline P-P=P-\underline P\geq\tfrac72.
\ee
For example, $P=\max(11,N+5)$, $\underline P=P-4$ and $\overline P=P+4$ satisfy these inequalities.
The projection to an exact solution also depends on a projection exponent $p\in(1/2,1)$.

In the companion works~\cite{BLF--PLF,BLF--PLF-stability} we uncovered the notions of (Hamiltonian and momentum) \emph{harmonic, radial, and shell stability} for the localization function~$\lambda$ (together with the exponents~$P,p$).  These six stability conditions take the form of \emph{Poincar\'e--Korn-type inequalities} for scalar and vector fields on~$\Lambda$, weighted by~$\lambda^{2P}$.  For suitable $(\Lambda,\lambda)$, they are known to hold for $p$~sufficiently close to its upper bound, which is~$1$ in dimension three, as explicited in Section~\ref{section-4}.  We henceforth assume these stability conditions.  Appendix~\ref{appendix:local-projection} justifies the use of the seed-to-solution projection in our local weighted H\"older norms and its continuous dependence on the seed data sets and their parameters. These estimates are used in later sections to choose the asymptotic charges.

For the Schwarzschild family, we have $p_G=1$. We fix an accuracy exponent $p_A>1$ and an integer $J$ such that $J+1>p_A$. We denote by $p_{3,p}^{\lambda}>1$ the linear decay bound of~\cite{BLF--PLF}, which depends on $(\Lambda,\lambda)$ but not on the masses~$m_1,m_2$. We choose the sharp decay exponent so that
\be
\label{eq:sharp-range}
1<p_\star\leq p_A,\qquad p_\star<\min\bigl(p_{3,p}^{\lambda},2\bigr).
\ee
The range is nonempty since $p_A,p_{3,p}^{\lambda}>1$. Even for large $p_A$, the second upper bound on~$p_\star$ remains.

With the geometric set-up and exponents being fixed, we then require smallness of the seed data set and its constraint error, with constants depending on these choices.
For the Schwarzschild gluing of interest here, this simply amounts to smallness of the masses~$m_1,m_2$.

%-----------------------------------------------------------------------------------------------------------------------------------------

\subsection{Weighted bounds for the Hamiltonian error}\label{section-3-2}
We verify the pointwise and integral estimates required for the Hamiltonian projection. The strict radial margin below is necessary for the weighted source norm to be finite.

\subsubsection{Weighted norms and exterior integrability.}
We recall the scalar norms involved in the time-symmetric application of~\cite{BLF--PLF}. For radial and angular weights $b,a\in\RR$, respectively, the weighted Lebesgue norm is defined by
\be
\label{eq:weighted-L2}
\|f\|_{L^2_{b,a}(\Omega)}^2 \coloneqq\int_\Omega |f|^2\lambda^{-2a}\rho^{2b-3}\,dx.
\ee
For an integer $l\geq0$, we also use $\|f\|_{H^l_{b,a}}\coloneqq(\sum_{|\beta|\leq l}\|\del^\beta f\|_{L^2_{b+|\beta|,a}}^2)^{1/2}$, with the same angular weight for every derivative. In the H\"older norm $C^{N,\alpha}_{b,a}(\Omega)$ explicited in~\eqref{equa-norm-Hol}, the order-$k$ derivative is multiplied by $\lambda^{-a+k}\rho^{b+k}$, and the top local H\"older seminorm by $\lambda^{-a+N+\alpha}\rho^{b+N+\alpha}$; the latter is evaluated in balls of radius $c\lambda\rho$ around each point, for a small constant $c>0$.
The metric deviation is measured in H\"older norm $C^{N,\alpha}_{p_G}(\RR^3)$, with a radial weight but no angular weight.

The scalar-curvature error is measured in the sum of the localized H\"older and Lebesgue norms
\be
\label{eq:error-norm}
\Eappr{p_A}[g,0] \coloneqq\|R_g\|_{C^{N-2,\alpha}_{p_A+2,\underline P-2}(\Omega)} +\|R_g\|_{L^2_{p_A+2,P}(\Omega)}.
\ee
Because in our present applications the curvature error is supported away from compact and angular boundaries of~$\Omega$, the angular weight $\lambda$ is bounded below and can be omitted from the norms up to changing some smallness constants.

We bound separately the norms~\eqref{eq:error-norm} in the region $r\geq 2R$ and the annulus $R\leq r\leq 2R$, and there is no contribution for $r\leq R$ since $g=\ebf$ there.
We use $\mu=(|m_1|+|m_2|)/R$ from~\eqref{eq:two-mass-profile} and put $s=r/R$. For $r\geq 2R$, the bound~\eqref{eq:finite-error} reads
\be
\label{eq:dimensionless-error}
|\del^kR_{\seedg(m_1,m_2)}| \leq C R^{-2-k}\mu^{J+1}s^{-J-3-k}(1+\log s)^K, \qquad s\geq2.
\ee
The square of the exterior Lebesgue norm is consequently bounded by
\be
\label{eq:radial-integral}
C R^{2p_A}\mu^{2J+2} \int_2^\infty s^{2(p_A-J-1)-1}(1+\log s)^{2K}\,ds.
\ee
The integral converges for $J+1>p_A$. Under the same inequality, the pointwise and H\"older norms feature negative powers of~$r$ and are bounded by $CR^{p_A}\mu^{J+1}$.

On the compact cutoff annulus $r/R\in[1,2]$,~\eqref{eq:compact-error} and the positive lower bound for $\lambda$ give the remaining contribution. Combining the compact and exterior regions, for $\mu$ in a fixed small interval, we obtain
\be
\label{eq:seed-bounds}
\|\seedg(m_1,m_2)-\ebf\|_{C^{N,\alpha}_1}\leq C_G R\mu, \qquad \Eappr{p_A}[\seedg(m_1,m_2),0]\leq C_A R^{p_A}(\mu+\mu^{J+1}).
\ee
The constants $C_G,C_A>0$ are uniform for $\mu<\mu_0$, without sign restriction on the masses. Near $\partial\Omega$, the seed data set is smooth and exact, so its constraint error vanishes. Since $\seedh=0$, its momentum error vanishes.

%-----------------------------------------------------------------------------------------------------------------------------------------

\subsection{Application: Schwarzschild gluing and the Hamiltonian silhouette}\label{section-3-3}

\subsubsection{The silhouette and relative energy.}
With these fixed choices and local norms, we now apply the seed-to-solution projection and harmonic decay theorems of~\cite{BLF--PLF} to the seed data obtained in Section~\ref{section-2}.  The geometric set-up is depicted in Figure~\ref{fig:global-manifold}.

\begin{figure}
\centering
\includegraphics[width=\linewidth]{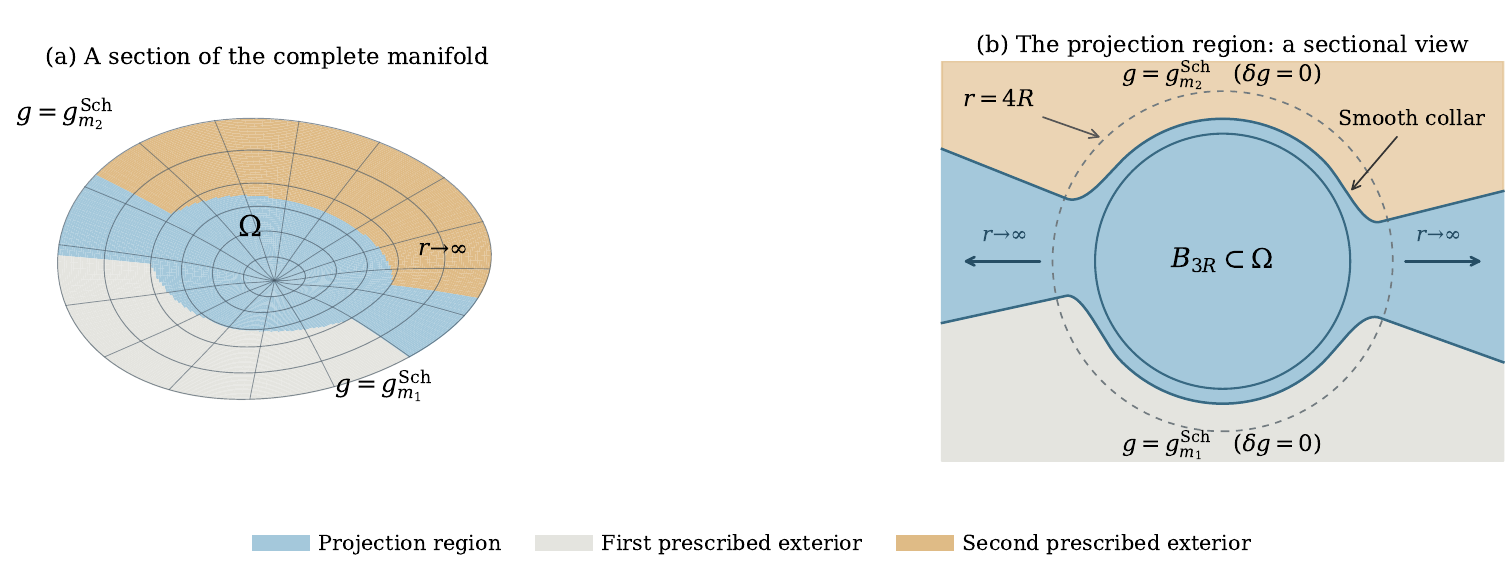}
\caption{The complete manifold~$\RR^3$ and the localization domain~$\Omega$ (blue) where the projection changes the seed data set.  The blue region contains the ball~$B_{3R}$ and is conical over the (not necessarily axisymmetric) angular region~$\Lambda$ outside~$B_{4R}$. The two prescribed Schwarzschild metrics are retained outside this region, namely $\delta g=0$ and $g=g_m^{\mathrm{Sch}}$ with $m=m_1,m_2$.  Right: the solid circle bounds $B_{3R}$, the dashed circle marks $r=4R$.}
\label{fig:global-manifold}
\end{figure}

A silhouette is an angular profile that determines a harmonic correction to the metric or momentum that lies in the kernel of the linearized constraints at infinity. We use the notation $\nu^{\normal}$ for the Hamiltonian silhouette and $\xi^{\normal(j)}$ for the momentum silhouette vector fields, with the normalizations of the companion papers. The Hamiltonian silhouette is uniquely characterized as follows, under the conditions of harmonic and radial stability of $(\Lambda,\lambda)$~\cite{BLF--PLF}: the associated tensor
\be
\label{eq:silhouette-tensor}
\mathcal S^\lambda \coloneqq\lambda^{2P}r^{3-2p} \bigl(\Hess-\ebf\Delta\bigr) \bigl(r^{-2+2p}\nu^{\normal}\bigr),\qquad r\geq4R,
\ee
has radial exponent~$r^{-1}$, unit relative ADM energy, and satisfies $L\mathcal S^\lambda=0$ on the exterior end. A fixed cutoff extends it smoothly into~$\Omega$, and it is extended by zero outside~$\Omega$; the boundary weight in~\eqref{eq:silhouette-tensor} ensures H\"older regularity.
The relative ADM energy is defined as
\be
\label{eq:relative-flux}
\mathcal F(\gamma)\coloneqq\frac1{16\pi}\lim_{r\to\infty} \int_{|x|=r}(\del_j \gamma_{ij}-\del_i \gamma_{jj})\theta_i\,dS, \qquad \mathcal F(\mathcal S^\lambda)=1.
\ee
In the following result, the remainder of order $r^{-p_\star}$ has zero relative ADM energy $\mathcal F(\eta_{g;m_1,m_2})=0$ because $p_\star>1$. Thus the coefficient of $\mathcal S^\lambda$ equals the relative energy between the solution and its seed data set.

\subsubsection{Complete scalar-flat metrics with prescribed sectors.}

\begin{theorem}[Complete scalar-flat gluing of two Schwarzschild regions with small masses of either sign]\label{thm:fixed-radius}
Fix the exponents $(N,\alpha,P,\underline P,\overline P,p)$, localization geometry $(R,\Lambda,\Omega,\lambda)$ satisfying the Hamiltonian stability conditions, and cutoff functions~$(\chi,\zeta_R)$ as in Sections~\ref{section-2} and~\ref{section-3-1}, as well as a geometry exponent $p_G=1$, an accuracy exponent $p_A>1$, an integer~$J$ with $J+1>p_A$, and $p_\star$ satisfying~\eqref{eq:sharp-range}. There then exists $\mu_*\in(0,1)$ such that, for every $(m_1,m_2)\in\RR^2$ with
\be
\label{eq:two-mass-smallness}
|m_1|+|m_2|<R\mu_*,
\ee
the seed data set~\eqref{eq:global-seed} satisfies the hypotheses of the seed-to-solution projection in the local weighted norms. Its projection $g_{m_1,m_2}$ is a complete scalar-flat metric on $\RR^3$ and agrees with the seed metric outside~$\Omega$. For $r\geq4R$ outside $\Lambda$, it is exactly $g_{m_1}^{\mathrm{Sch}}$ where $\chi=0$ and $g_{m_2}^{\mathrm{Sch}}$ where $\chi=1$. It can be written as
\be
\label{eq:projected-expansion}
g_{m_1,m_2}=\seedg(m_1,m_2)+\mathfrak m(m_1,m_2)\mathcal S^\lambda+\eta_{g;m_1,m_2},\quad \|\eta_{g;m_1,m_2}\|_{C^{N,\alpha}_{p_\star,\underline P}(\Omega)}\leq C\Eappr{p_\star}^+[\seedg(m_1,m_2),0]
\ee
for some $\mathfrak m(m_1,m_2)\in\RR$ and constant $C>0$ that is uniform in the masses satisfying~\eqref{eq:two-mass-smallness} regardless of their sign.
\end{theorem}

The notation $g_m$, $\mathfrak m(m)$, and $\eta_{g;m}$ denotes the restriction $(m_1,m_2)=(0,m)$.  For $m_1=m_2=0$ the seed data set and its projection are Euclidean throughout~$\RR^3$.  For a more general mass $m_1=m_2\neq 0$, exterior corrections vanish but the Schwarzschild--Euclidean cutoff still has a compact error. The weighted estimate~\eqref{eq:seed-bounds} contains the term $R^{p_A}\mu=(|m_1|+|m_2|)R^{p_A-1}$ from the compact cutoff error. Thus increasing $R$ at fixed masses need not improve this bound. For each fixed $R$ and $p_A$, taking the masses sufficiently small makes all required norms small; the permitted threshold for $\mu$ depends on these choices. Changing the domain~$\Lambda$, localization function~$\lambda$, or exponents also requires a new threshold.

The mass parameters $m_1,m_2$ may be negative, in which case the ADM energy of $\seedg(m_1,m_2)$ is negative and must be compensated by a sufficiently large positive relative ADM energy $\mathfrak m(m_1,m_2)$.  This is necessary for the ADM energy of the projected metric~$g_{m_1,m_2}$ to be positive, in accordance with the positive mass theorem.  We discuss the ADM energy further in Section~\ref{section-5-ADM}, below.

\emph{In conclusion, for every fixed admissible localization data set, the same threshold gives exact gluing for all sufficiently small masses $m_1,m_2\in\RR$.}

\begin{proof}
It suffices to verify that the complete Euclidean reference and the fixed localization geometry satisfy the conditions of~\cite[Theorem~2.11]{BLF--PLF} for the existence of the seed-to-solution projection and of~\cite[Theorem~4.4]{BLF--PLF} for the improved super-harmonic radial decay rate $p_\star>1$ stated in~\eqref{eq:projected-expansion}. Proposition~\ref{prop:finite-expansion} gives uniform exterior positivity and error bounds for $\mu<\mu_0$. The convex interpolation~\eqref{eq:global-seed} remains positive, smooth, and vacuum away from~$\Omega$ for either mass sign. Moreover, $\mu<1$ gives $1+M/(2r)\geq1/2$ on $r\geq R$; negative-mass Schwarzschild singularities lie inside the replaced region.
The equation~\eqref{eq:seed-bounds} verifies the source and near-reference hypotheses uniformly in $m_1,m_2$. With $\eps_G,\eps_A>0$ sufficiently small (as a function of exponents and localization geometry) as selected by \cite[Theorem~4.4]{BLF--PLF}, we choose $\mu_*>0$ such that
\be
\label{eq:last-smallness}
C_G R\mu_*\leq\eps_G,\qquad \mu_*\leq\mu_0,\qquad \bigl(2C_A R^{p_A}+C_+(R^{p_\star}+R)\bigr)\mu_*\leq\eps_A.
\ee
This ensures applicability of the results in~\cite{BLF--PLF}.
Since the reference data and seed data sets are time-symmetric ($h_0=\seedh=0$), the momentum constraint is exactly obeyed by the seed data set, so that the seed-to-solution projection~\eqref{eq:seed-to-solution} has a vanishing vector field~$Z$ and only involves a single scalar field~$u$, hence $h=\seedh=0$.
Harmonic decay gives~\eqref{eq:projected-expansion}, and localization preserves the exterior metrics.

In the case $m_1=m_2=0$, the seed data is trivially $(\seedg,\seedh)=(\ebf,0)$ and one has $u=0$ and $(g,h)=(\seedg,\seedh)$.
\end{proof}

%-----------------------------------------------------------------------------------------------------------------------------------------

%============================================================================

\section{Vacuum data with slow leading decay and super-harmonic remainders}\label{section:slow-coupled-projection}
We now apply the projection to the more general seed data sets of Section~\ref{section:full-constraints}. We first construct scalar flat metrics with slow leading metric decay (time-symmetric case $h=\seedh=0$), then solve both constraints with nonzero second fundamental form and describe the Hamiltonian and momentum silhouettes.

%-----------------------------------------------------------------------------------------------------------------------------------------

\subsection{Slow leading decay and super-harmonic remainders}\label{section-3-4}
We prescribe a nonzero slow leading term with radial exponent $q\in(0,1)$. The accuracy exponent $p_A$ can be chosen to be arbitrarily large by retaining sufficiently many asymptotic corrections.
However, the sharp decay exponent~$p_\star$ (with which the exact metric and seed metric are compared) is limited in~\eqref{eq:slow-sharp-range} by further effects such as non-linear terms involving the leading $r^{-q}$ decay and the harmonic corrections~$\mathcal S^\lambda$.
The threshold~$t_*$ degenerates as $p_G\to0^+$ or $p_A\to\infty$.

\begin{proposition}[Slow leading term and high-accuracy seed data sets]\label{thm:slow-leading}
Fix the exponents $(N,\alpha,P,\underline P,\overline P,p)$, localization geometry $(R,\Lambda,\Omega,\lambda)$ satisfying the stability conditions, and radial cutoff function~$\zeta_R$ as in Sections~\ref{section-2} and~\ref{section-3-1}, as well as an exponent $q\in(0,1)$, an accuracy exponent $p_A>1$, and an integer~$J$ with $(J+1)q>p_A$, and
\be
\label{eq:slow-sharp-range}
1<p_\star\leq p_A,\qquad p_\star<\min\bigl(p_{3,p}^{\lambda},1+q\bigr).
\ee
Also fix a nonzero smooth angular tensor~$\gamma_0$, compactly supported in $\Lambda$, with $L(r^{-q}\gamma_0)=0$.
The choice of gluing domain $\Omega$ in Section~\ref{section-2-4} ensures that the cone over $\supp \gamma_0$ is contained in~$\Omega$. With the norms of Section~\ref{section-3-2}, there exists $t_*>0$, depending on these fixed choices, such that every $|t|<t_*$ defines a global seed data set
\be
\label{eq:slow-global-seed}
\seedg^{(q)}(t)\coloneqq\ebf+\zeta_R\bigl(g_{t,q}^{[J]}-\ebf\bigr),\qquad \seedh^{(q)}(t)\coloneqq0,
\ee
satisfying all the seed data set hypotheses of~\cite{BLF--PLF} with $p_G=q$ and the prescribed~$p_A$. Its seed-to-solution projection $(g_q(t),0)$ is scalar flat, equals $(\ebf,0)$ on $\RR^3\setminus\Omega$, and satisfies
\be
\label{eq:slow-projected}
g_q(t)=\seedg^{(q)}(t)+\mathfrak m_q(t)\mathcal S^\lambda+\eta_{g;q}(t), \qquad \eta_{g;q}(t)\in C^{N,\alpha}_{p_\star,\underline P}(\Omega).
\ee
In the fixed asymptotic Cartesian chart, with angular convergence in $C^2(S^2)$, one has
\be
\label{eq:slow-leading-limit}
\lim_{r\to\infty}r^q\bigl(g_q(t)(r\theta)-\ebf\bigr)=t \gamma_0(\theta).
\ee
Consequently, for $t\neq0$, the metric deviation has leading exponent exactly $q$: it does not belong to $C^0_{q+\eta}$ for any $\eta>0$.
\end{proposition}

\begin{proof}\bse
The smooth extension through the interior has its constraint error supported in $\Omega$. On the cutoff annulus it satisfies $|\del^kR_{\seedg^{(q)}(t)}|\leq C R^{-2-k}\mu_q$. On the exterior, Proposition~\ref{prop:coupled-expansion} specializes to~\eqref{eq:slow-exterior-error} in the time-symmetric case. The squared $L^2_{p_A+2,P}$ norm restricted to the exterior region $r\geq 2R$ is bounded by
\be
\label{eq:slow-radial-integral}
C R^{2p_A}\mu_q^{2J+2} \int_2^\infty s^{2(p_A-(J+1)q)-1}(1+\log s)^{2K}\,ds.
\ee
This integral is finite exactly when $p_A<(J+1)q$. The weighted local H\"older estimates follow from the differentiated bounds as in Section~\ref{section-3-2}. We obtain
\be
\label{eq:slow-seed-bounds}
\|\seedg^{(q)}(t)-\ebf\|_{C^{N,\alpha}_q}\leq C_G R^q\mu_q, \qquad \Eappr{p_A}[\seedg^{(q)}(t),0]\leq C_A R^{p_A}(\mu_q+\mu_q^{J+1}).
\ee
All constants are independent of $t$. We choose $\mu_{q,*}<1$ so that positivity holds and~\eqref{eq:slow-seed-bounds} meet the thresholds $\eps_G,\eps_A$ in \cite[Theorem~4.4]{BLF--PLF}, then set $t_*\coloneqq R^q\mu_{q,*}$. The super-harmonic decay proven there gives~\eqref{eq:slow-projected}, with zero momentum as in Theorem~\ref{thm:fixed-radius}.

Every nonlinear term in the seed data set has exponent $jq>q$, so $r^q$ times that term, together with its first two angular derivatives, tends to zero despite the finite logarithmic factors. The projected harmonic term satisfies $r^q\mathcal S^\lambda=\Obig(r^{q-1})\to0$, and the remainder satisfies $r^q\eta_{g;q}(t)=\Obig(r^{q-p_\star})\to0$. This proves~\eqref{eq:slow-leading-limit}. A nonzero value of $\gamma_0$ in one angular direction rules out each stronger bound $\Obig(r^{-q-\eta})$ when $t\neq0$.
\ese\end{proof}

Observe that the conclusions apply equally well to the seed metrics~\eqref{eq:slow-Schwarzschild} obtained by adding slow-decaying terms of amplitude~$t$ to Schwarzschild--Minkowski interpolations of mass~$m$.
With $\mu\coloneqq|t|/R^q+|m|/R$, each remainder monomial of degree at least $J+1$ has radial exponent at least $(J+1)q$, since $aq+b\geq(a+b)q$. The exterior and norm estimates therefore hold with $\mu_q$ replaced by~$\mu$. Extension through the interior using the cutoff~$\zeta_R$ gives a suitable global seed data set provided $t,m$ are chosen small after the geometry, profiles, exponents, and order.
All conclusions of Proposition~\ref{thm:slow-leading} go through, with in particular a harmonic correction $\mathcal S^\lambda$ of the same $r^{-1}$~decay order as the leading Schwarzschild--Minkowski interpolation, and \eqref{eq:slow-leading-limit}~is only sensitive to the slow-decaying terms of amplitude~$t$ and not the mass parameter~$m$.

%-----------------------------------------------------------------------------------------------------------------------------------------

\subsection{Full Einstein constraints and momentum silhouettes}\label{section-3-coupled}
We apply the projection to seed data sets with nonzero second fundamental form. The expansion contains the Hamiltonian and momentum silhouettes, whose coefficients are the relative ADM energy and momentum.

For the full constraints, we write $\mathcal G\coloneqq(\mathcal H,\mathcal C)$ and extend~\eqref{eq:error-norm} by
\be
\label{eq:full-error-norm}
\Eappr{b}[\seedg,\seedh]\coloneqq\|\mathcal H(\seedg,\seedh)\|_{C^{N-2,\alpha}_{b+2,\underline P-2}}+\|\mathcal C(\seedg,\seedh)\|_{C^{N-1,\alpha}_{b+2,\underline P-1}}+\|\mathcal H(\seedg,\seedh)\|_{L^2_{b+2,P}}+\|\mathcal C(\seedg,\seedh)\|_{L^2_{b+2,P}}.
\ee
All localized norms are over~$\Omega$. For $b=1$ an additional integral term would be needed, but it is bounded by $\Eappr{b}[\seedg,\seedh]$ for the case $b>1$ of interest to us here. Multiplication by metrics or inverse metrics preserves these spaces uniformly in our small families, by the local product rule.

We use the momentum silhouettes with the normalization of~\cite{BLF--PLF}. If $\xi^{\normal(j)}$ is the normalized vector kernel element, the associated momentum tensor is
\be
\label{eq:momentum-silhouette}
(\mathcal T_j^\lambda)^{ik}\coloneqq-\tfrac12\lambda^{2P}r^{1-2p} \bigl(\del_i(r^{-2+2p}\xi_k^{\normal(j)})+\del_k(r^{-2+2p}\xi_i^{\normal(j)})\bigr).
\ee
Its Cartesian components are $r^{-2}$ times angular functions, and we choose a fixed compact extension as for~$\mathcal S^\lambda$. Thus $\nu^{\normal}$ and $\xi^{\normal(j)}$ are angular profiles, whereas $\mathcal S^\lambda$ and $\mathcal T_j^\lambda$ are the resulting metric and momentum tensors. On the exterior end, the harmonic contributions denoted by $g^\infty,h^\infty$ in~\cite{BLF--PLF} are $\mathfrak m\mathcal S^\lambda$ and $\sum_j\mathfrak J_j\mathcal T_j^\lambda$, respectively.
The same construction and conclusions hold, starting from the seed data sets \eqref{eq:coupled-exterior-multi} or~\eqref{eq:slow-Schwarzschild} involving finitely many amplitudes or perturbations of Schwarzschild--Minkowski data, with smallness of amplitudes and masses imposed last.

\begin{theorem}[Localized vacuum data with slow leading decay and super-harmonic remainders]\label{thm:coupled-projection}
Fix the exponents $(N,\alpha,P,\underline P,\overline P,p)$, localization geometry $(R,\Lambda,\Omega,\lambda)$ satisfying the stability conditions, and radial cutoff function~$\zeta_R$ as in Sections~\ref{section-2} and~\ref{section-3-1}, as well as an integer~$J$ and exponents $q\in(0,1)$, $p_A>1$, and $p_\star$~satisfying $(J+1)q>p_A$ and~\eqref{eq:slow-sharp-range}.
Also fix angular profiles $\gamma_0,\pi_0$ as in Proposition~\ref{prop:coupled-expansion}. Then, for all sufficiently small $|t|$, the global seed data set
\be
\label{eq:coupled-global}
\seedg(t)\coloneqq\ebf+\zeta_R(g_t^{[J]}-\ebf),\qquad \seedh(t)\coloneqq\zeta_R h_t^{[J]}
\ee
has accuracy $p_A$ for both constraints. Its seed-to-solution projection $(g(t),h(t))$ is a complete vacuum initial data set, equals $(\ebf,0)$ outside~$\Omega$, and has the asymptotic expansion
\be
\label{eq:coupled-projected}
\begin{aligned}
g(t)&=\seedg(t)+\mathfrak m(t)\mathcal S^\lambda+\eta_g(t),
&\eta_g(t)&\in C^{N,\alpha}_{p_\star,\underline P},\\
h(t)&=\seedh(t)+\sum_{j=1}^3\mathfrak J_j(t)\mathcal T_j^\lambda+\eta_h(t),
&\eta_h(t)&\in C^{N,\alpha}_{p_\star+1,\underline P}.
\end{aligned}
\ee
The tensors $\mathcal T_j^\lambda$ are $r^{-2}$ times angular functions outside a ball and have unit relative ADM momentum in direction~$j$. With $C^2(S^2)$ convergence, the leading limits are
\be
\label{eq:coupled-leading-limits}
r^q(g(t)-\ebf)\longrightarrow t \gamma_0,\qquad r^{q+1}h(t)\longrightarrow t\pi_0.
\ee
For $t\neq0$, $\pi_0\neq0$, the second fundamental form has nonzero leading coefficient $t(\pi_0-\tfrac12\Tr_{\ebf}\pi_0\,\ebf)$ at order~$r^{-q-1}$.
\end{theorem}
\begin{proof}
Both compact cutoff errors are bounded by $CR^{-2-k}\mu$, where $\mu=|t|/R^q$. Differentiating $\zeta_R h_t^{[J]}$ gives this scale because $h_t^{[J]}=\Obig(\mu/R)$ on the annulus. Its quadratic Hamiltonian contribution is bounded by $CR^{-2-k}\mu^2$, which is smaller. On the exterior we use~\eqref{eq:coupled-error}. The supports stay away from the boundary, so the proof of~\eqref{eq:slow-seed-bounds} applies to both sources and gives
\be
\label{eq:coupled-seed-norms}
\|\seedg(t)-\ebf\|_{C^{N,\alpha}_q}+\|\seedh(t)\|_{C^{N,\alpha}_{q+1}}\leq C R^q\mu,
\qquad
\Eappr{p_A}[\seedg(t),\seedh(t)]\leq C R^{p_A}(\mu+\mu^{J+1}).
\ee
For sufficiently small~$\mu$, these bounds are strong enough to apply the super-harmonic theorem of~\cite{BLF--PLF}, which gives~\eqref{eq:coupled-projected}.

Since $\underline P>N+\alpha$, both corrections extend by zero in $C^{N,\alpha}$ (Appendix~\ref{appendix:local-projection}). Smallness gives $c\ebf\leq g(t)\leq C\ebf$, hence completeness by Euclidean length comparison. Higher-order terms in the seed data set and the harmonic corrections, of orders $r^{-1},r^{-2}$, decay faster than the leading pair, proving~\eqref{eq:coupled-leading-limits}. The equation~\eqref{eq:full-conventions} gives the leading coefficient of $K$; its coefficient is nonzero for $\pi_0\neq0$ because Euclidean trace reversal is invertible. The same bounds with the sum of dimensionless amplitudes prove the assertions for several amplitudes.
\end{proof}

\subsubsection{Silhouettes as angular profiles of the correction.}
We now describe briefly what the silhouettes determine in the exact solution. The scalar profile $\nu^{\normal}$ and the vector profiles $\xi^{\normal(j)}$ determine the metric and momentum tensors $\mathcal S^\lambda$ and $\mathcal T_j^\lambda$ through~\eqref{eq:silhouette-tensor} and~\eqref{eq:momentum-silhouette}, which are four symmetric two-tensors.
In fixed asymptotic Cartesian components, we define the angular tensor profiles by
\be
\label{eq:angular-silhouette-tensors}
\mathcal S^\lambda_\infty(\theta)\coloneqq r\mathcal S^\lambda(r\theta),\qquad \mathcal T^\lambda_{\infty,j}(\theta)\coloneqq r^2\mathcal T_j^\lambda(r\theta),\qquad r\geq4R.
\ee
Homogeneity makes these definitions independent of~$r$. For the solutions of Theorem~\ref{thm:coupled-projection}, the expansion~\eqref{eq:coupled-projected} gives
\be
\label{eq:silhouette-radial-limits}
\begin{gathered}
r\bigl(g(t)(r\theta)-\seedg(t)(r\theta)\bigr)\longrightarrow\mathfrak m(t)\mathcal S^\lambda_\infty(\theta),\\
r^2\bigl(h(t)(r\theta)-\seedh(t)(r\theta)\bigr)\longrightarrow\sum_{j=1}^3\mathfrak J_j(t)\mathcal T^\lambda_{\infty,j}(\theta)
\end{gathered}
\ee
in $C^2(S^2)$, with zero extension outside~$\Lambda$. Indeed, the weighted remainder bounds imply
\be
\label{eq:silhouette-limit-error}
\|r\eta_g(t)(r\,\cdot)\|_{C^2(S^2)}+\|r^2\eta_h(t)(r\,\cdot)\|_{C^2(S^2)}\leq C r^{1-p_\star}\Bigl(\|\eta_g(t)\|_{C^{N,\alpha}_{p_\star,\underline P}}+\|\eta_h(t)\|_{C^{N,\alpha}_{p_\star+1,\underline P}}\Bigr).
\ee
An angular derivative contributes at most one factor~$r$. The angular weights remain bounded for derivatives through order two, and $\underline P>N+\alpha$ allows for extension by zero with no regularity loss. Since $p_\star>1$, the right-hand side tends to zero.

The angular tensor profiles in~\eqref{eq:angular-silhouette-tensors} are normalized to have unit ADM energy and momentum, which identifies their coefficients $\mathfrak m$ and $\mathfrak J_j$ with the relative energy and momentum between the exact solution and its seed data set.  Section~\ref{section-5} explains some freedom in selecting these coefficients.

For $0<q<1$, the harmonic correction is smaller than the leading terms of the seed data set, and can only be seen at higher orders in the expansion in powers of~$r$.  The slow leading profiles in~\eqref{eq:coupled-leading-limits} are therefore not sensitive to the harmonic terms.

As an interesting special case, an axisymmetric angular domain and (stable) localization function give an axisymmetric Hamiltonian silhouette~$\nu^{\normal}$. Indeed, rotation about the axis preserves the kernel equation and the unit energy normalization. Since the kernel is one-dimensional, the rotated normalized element equals the original one. The momentum kernel instead has dimension~$3$, corresponding to three silhouettes~$\xi^{\normal(j)}$, one for each Cartesian momentum component. Rotations about the axis preserve its equation and rotate these components. By uniqueness of the normalized kernel basis, the silhouette associated with momentum along the axis is invariant under the corresponding action on vector fields, while the two transverse silhouettes rotate into each other. Section~\ref{section-4-1}, below, gives the three silhouette explicitly in the case $\Lambda=S^2$ and $\lambda\equiv 1$.

\subsubsection{Finite-dimensional families with prescribed leading profiles.}
An interesting set-up is as follows.
We fix the preceding choices, including the polynomial inverses and the extension through the interior, and an exponent $q\in(0,1)$, as well as pairs $(\gamma_\ell,\pi_\ell)$, $1\leq \ell\leq d$ subject to the linear constraints~\eqref{eq:coupled-leading-kernel}, namely $L(r^{-q}\gamma_\ell)=0$ and $\operatorname{div}_{\ebf}(r^{-q-1}\pi_\ell)=0$.  The construction with several amplitudes gives solutions $X(a)\coloneqq(g_a,h_a)$ for all $a\in\RR^d$ in a sufficiently small ball, with
\be
\label{eq:profile-parametrization}
\lim_{r\to\infty}\bigl(r^q(g_a-\ebf),r^{q+1}h_a\bigr) =\sum_{\ell=1}^d a_\ell(\gamma_\ell,\pi_\ell)\quad\hbox{in }C^2(S^2).
\ee
Uniform recursion bounds give this limit. By linear independence of the leading profiles, this limit determines the parameter vector $a$ uniquely. The family is therefore injective in the fixed coordinates. Compact perturbations (Section~\ref{section-ADM-background}) can change charges while preserving these leading profiles.

%============================================================================

\section{Explicit whole-sphere silhouettes and gluing through narrow angular transitions}\label{section-4}

%-----------------------------------------------------------------------------------------------------------------------------------------

\subsection{Hamiltonian and momentum silhouettes on the whole sphere}\label{section-4-1}
In Sections~\ref{section-3}--\ref{section:slow-coupled-projection} we have assumed that the localization function~$\lambda$ is stable~\cite{BLF--PLF}.
We combine the constructions of seed data sets with the stability results of the companion papers \cite{BLF--PLF,BLF--PLF-stability} in two limiting settings: the whole sphere, and small-aperture domains, which permit thin Schwarzschild transition regions.

We now calculate the harmonic metric and momentum corrections when the projection acts on the whole sphere, namely when $\Omega=\RR^3$, $\lambda\equiv1$, and $\Lambda=S^2$. The transition in the seed data set may remain angularly restricted, while projection acts on the whole sphere with radial weights. Sections~\ref{section:full-constraints} and~\ref{section:slow-coupled-projection} still give global seed data sets. Whole-sphere stability holds in dimension three for $p\in(1/2,1)$ close to~$1$~\cite{BLF--PLF,BLF--PLF-stability}. We fix the geometry, then $p,p_A,J,p_\star$, and finally the amplitude.

The harmonic scalar silhouette is constant on $S^2$. With normalization~\eqref{eq:relative-flux}, we obtain
\be
\label{eq:whole-silhouette}
\nu^{\normal}=\frac1{2(1-p)}, \qquad \mathcal S^1=\frac1r\bigl((4-2p)\theta^\flat\otimes\theta^\flat-(2-2p)\ebf\bigr).
\ee
Indeed, for a radial function $r^{-a}$ with $a=2(1-p)$, its Hessian is $a r^{-a-2}((a+2)\theta^\flat\otimes\theta^\flat-\ebf)$ and its Laplacian is $a(a-1)r^{-a-2}$. Substitution into~\eqref{eq:silhouette-tensor} gives~\eqref{eq:whole-silhouette}. The ADM energy~$\mathcal F$ defined in~\eqref{eq:relative-flux} of $r^{-1}(A\theta^\flat\otimes\theta^\flat+B\ebf)$ is $(A+B)/2$, hence $\mathcal F(\mathcal S^1) = 1$ as stated below~\eqref{eq:silhouette-tensor}.

The momentum silhouettes are also explicit~\cite{BLF--PLF}. Let $a=2(1-p)>0$ as before, and let $e_j$ denote the $j$-th Cartesian unit vector. With the unit momentum normalization recalled in Section~\ref{section-3-coupled}, we obtain
\be
\label{eq:whole-momentum-silhouettes}
\begin{aligned}
\xi^{\normal(j)}&=\frac{6}{a(a+5)}\bigl(3e_j+a\theta_j\theta\bigr),\\
\mathcal T_j^1&=\frac{6}{(a+5)r^2}\bigl(e_j\otimes\theta+\theta\otimes e_j-\theta_j\ebf^{-1}+(a+2)\theta_j\theta\otimes\theta\bigr).
\end{aligned}
\ee
To verify both the equation and the normalization, we take a trial profile $\xi=ce_j+d\theta_j\theta$ and denote by $\mathcal T$ its tensor from~\eqref{eq:momentum-silhouette} with $\lambda=1$. Differentiation gives $\operatorname{div}_{\ebf}\mathcal T=(ac-3d)(e_j-3\theta_j\theta)/(2r^3)$. Thus $ac=3d$ solves the harmonic momentum equation. For these coefficients, $\mathcal T\theta=dr^{-2}(e_j+(a+2)\theta_j\theta)$. Since $\int_{S^2}(e_j+(a+2)\theta_j\theta)_k\,d\theta=4\pi(a+5)\delta_{jk}/3$, its ADM momentum is $d(a+5)e_j/6$. Unit momentum in direction $j$ therefore gives $d=6/(a+5)$ and $c=18/(a(a+5))$, proving~\eqref{eq:whole-momentum-silhouettes}. The three relative momentum components~$\mathfrak J_j$ are the amplitudes of these tensors in~\eqref{eq:coupled-projected}.

Overall, whole-sphere projection from $(\seedg,\seedh)$ to an exact solution~$(g,h)$ introduces harmonic corrections $\mathcal S^1$ and~$\mathcal T^1_j$ that affect all components of the metric and momentum.  This can be compensated by adjusting the seed data as in \cite[Theorems~6.1 and~6.4]{LeFlochNguyen}, further extended in Section~\ref{section-5}.

%-----------------------------------------------------------------------------------------------------------------------------------------

\subsection{Application: small angular apertures}\label{section-4-2}
We switch to the other extreme limit, whose stability was proven in \cite[Section 5]{BLF--PLF-stability}, namely shrinking a fixed reference domain to obtain localization domains~$\Lambda$ with small angular aperture.  This case accommodates all types of seed data sets in Sections~\ref{section-2}--\ref{section:full-constraints} and their projections in Sections~\ref{section-3}--\ref{section:slow-coupled-projection}.
The weighted stability estimates determine the allowed projection exponents $p$ close to~$1$; the masses are chosen only after the angular geometry and these exponents have been fixed.

\subsubsection{Angular scaling and stability.}
 We fix a bounded connected reference domain $\breve\Lambda\Subset B_1(0)\subset\RR^2$ with smooth boundary and a smooth boundary-defining function $\breve\lambda$, positive inside and vanishing linearly at the boundary. For a scale $0<\delta\leq1$, we define
\be
\label{eq:aperture-scaling}
\Phi_\delta(y)\coloneqq(\delta y,\sqrt{1-\delta^2|y|^2}), \qquad \Lambda_\delta\coloneqq\Phi_\delta(\breve\Lambda), \qquad \lambda_\delta\circ\Phi_\delta=\breve\lambda.
\ee
As proven in~\cite{BLF--PLF-stability}, the diameter and weighted Poincar\'e constants (whose precise definition is not used in the present work) of the rescaled domain~$\Lambda_\delta$ satisfy
\be
\label{eq:scaled-Poincare}
\operatorname{diam}(\Lambda_\delta) +C_{\mathrm P}^{\lambda_\delta}+C_{\mathrm P2}^{\lambda_\delta} \lesssim \delta,
\ee
which ensures that for small enough $\delta>0$, Hamiltonian stability holds.
Once the localization domain and function $(\Lambda_\delta,\lambda_\delta)$ are fixed, it is shown in~\cite{BLF--PLF-stability} that for small enough $1-p>0$, momentum stability holds.

We fix the reference domain and weight, angular scale, projection and sharp exponents, radius and cutoffs, and accuracy order before the masses. The geometric inequalities are independent of mass smallness.

%-----------------------------------------------------------------------------------------------------------------------------------------

\subsubsection{Narrow angular transition regions}
We now choose an annulus as our reference domain~$\breve{\Lambda}$, so that the angular transition can be narrow compared with its length. Its relative width $\tau$ is distinct from the overall angular scale~$\delta$: we fix the width first, then choose the angular scale sufficiently small, then $1-p$ sufficiently small, and finally choose sufficiently small masses.

We can take $\breve\Lambda$ to be a smooth annular region separating two nested closed curves in the plane. A smooth function $\breve\chi$ is one on the inner side and zero on the outer side, with its transition supported inside the annulus. Transporting it by~\eqref{eq:aperture-scaling}, and extending it by its constant values, gives a smooth function $\chi_\delta$ on $S^2$. The seed data set~\eqref{eq:two-mass-exterior} then glues Schwarzschild masses $m_2$ and $m_1$ in the inner and outer angular regions, respectively; $m_1=0$ recovers Schwarzschild--Minkowski gluing. The curves and the transition function may be non-axisymmetric.

To describe a narrow annulus, we fix a reference smooth simple closed curve~$\breve\Gamma$ and consider the squared distance function $d(x,\breve\Gamma)^2$, which is smooth in a tubular neighborhood $d(x,\breve\Gamma)<\tau_{\max}$.  We then let $\breve\Lambda$ be a tubular neighborhood of~$\breve\Gamma$ of positive width $\tau<\tau_{\max}$, equipped with the smooth boundary defining function $\breve\lambda(x) = (\tau^2-d(x,\breve\Gamma)^2)/\tau^2\in[0,1]$ for $x\in\breve\Lambda$, extended by zero outside.  For each fixed~$\tau$, Hamiltonian and momentum stability hold for sufficiently small~$\delta$, $1-p$, and masses (chosen in this order). Arbitrarily small width-to-length ratios~$\tau$ are therefore allowed, with quantifiers
\be
\label{eq:two-scale-quantifiers}
\begin{aligned}
& \forall(N,\alpha,P,\underline P,\overline P) \text{ admissible},\quad
\forall\breve\Gamma,\quad
\forall\tau\in(0,\tau_{\max}(\breve\Gamma)),\quad
\exists\delta_*(\breve\Gamma,\tau)>0,\quad \forall\delta\in(0,\delta_*),\\
& \exists q_*(\breve\Gamma,\tau,\delta,P)>0, \quad \forall p \in(1-q_*,1),\quad
\forall R,p_A,J\quad
\exists\mu_*(\breve\Gamma,\tau,\delta,P,p,R,p_A,J)>0.
\end{aligned}
\ee
For each such fixed choice, every signed pair with $|m_1|+|m_2|<R\mu_*$ is permitted, including either or both masses equal to zero.
Here we wrote quantifiers for the Schwarzschild--Minkowski or Schwarzschild--Schwarzschild gluing of Sections~\ref{section-2} and~\ref{section-3}.
For other types of seed data sets and their projection considered in Sections~\ref{section:full-constraints} and~\ref{section:slow-coupled-projection}, one should replace $R,p_A,J$ by $R,p_A,q,J$ with the slow decay exponent $q\in(0,1)$, and the meaning of~$\mu_*$ is modified, namely the bound $|m_1|+|m_2|<R\mu_*$ is replaced by analogous thresholds on the amplitudes~$t$ or~$t_\ell$ of leading profiles.

\subsubsection{Mass thresholds for thin annuli.} The scaling theorem applies with a fixed reference annulus; its constants may depend on~$\tau$. Getting a uniform estimates in the limit $\tau\to 0$ may be feasible provided one can determine the asymptotic behaviour of Poincar\'e constants in that limit. We have not explored this direction.

%============================================================================

\section{Choosing energy-momentum: relative charges and harmonic parameters}\label{section-5}

%-----------------------------------------------------------------------------------------------------------------------------------------

\subsection{Absolute and relative ADM charges}\label{section-5-ADM}
We now calculate the charge carried by the difference between an exact solution and its seed data set. We then select the harmonic parameters to prescribe this relative vector, keeping the leading profiles and the exterior geometry fixed.

We define the absolute and relative ADM charges in a fixed asymptotic chart and prove convergence of the corresponding surface integrals under the relevant decay assumptions. The distinction between absolute and relative is necessary when the individual data sets decay too slowly to have finite ADM charges.

\subsubsection{Surface integrals and convergence.}
We distinguish the energy $E_{\mathrm{ADM}}$ in a fixed asymptotic rest frame, the momentum $P_{\mathrm{ADM}}\in\RR^3$, and the invariant mass $\sqrt{E_{\mathrm{ADM}}^2-|P_{\mathrm{ADM}}|^2}$ when the energy-momentum vector is future causal, which holds for exact vacuum solutions $(g,h)$ by the positive mass theorem but not necessarily for seed data sets $(\seedg,\seedh)$. In time symmetry $P_{\mathrm{ADM}}=0$, so mass is the absolute value of energy, which is non-negative for $(g,0)$ but could be negative for $(\seedg,0)$. The Schwarzschild parameters $m_1,m_2$ prescribed on angular sectors need not equal any of these global quantities.

In the fixed Cartesian chart, we write $S_r\coloneqq\{x:|x|=r\}$ and define the finite-radius linear surface integrals
\be
\label{eq:four-fluxes}
E_r(\gamma)\coloneqq\frac1{16\pi}\int_{S_r}(\del_j \gamma_{ij}-\del_i \gamma_{jj})\theta_i\,dS, \qquad P_r(\pi)^j\coloneqq\frac1{8\pi}\int_{S_r}\pi^{ij}\theta_i\,dS.
\ee
In the standard ADM decay range the ADM charges $E_{\mathrm{ADM}},P_{\mathrm{ADM}}$ are $r\to+\infty$ limits of $E_r(g-\ebf),P_r(h)$, respectively. The relative energy-momentum between two data sets is instead
\be
\label{eq:relative-four-charge}
\mathcal Q_{\mathrm{rel}}((g,h);(g',h')) \coloneqq\lim_{r\to\infty}\bigl(E_r(g-g'),P_r(h-h')\bigr).
\ee
This is the relative ADM definition of~\cite{BLF--PLF}, in a fixed asymptotic chart. Linearity makes it additive in the sense that $\mathcal Q_{\mathrm{rel}}((g,h);(g',h'))+\mathcal Q_{\mathrm{rel}}((g',h');(g'',h''))=\mathcal Q_{\mathrm{rel}}((g,h);(g'',h''))$ whenever two of these limits exist. Positivity need not hold: Theorem~\ref{thm:effective-ADM}, below, supplies two solutions with equal positive energy and opposite nonzero momenta, whose relative ADM vector has zero energy and nonzero momentum.

\begin{proposition}[The charges of a projected seed data set]\label{prop:ADM-relation}
For the full projection in Theorem~\ref{thm:coupled-projection}, set $\mathfrak J\coloneqq(\mathfrak J_1,\mathfrak J_2,\mathfrak J_3)$ and $\mathbf q_{\mathrm{rel}}\coloneqq(\mathfrak m,\mathfrak J)$. One then has
\be
\label{eq:relative-modulator-vector}
\mathcal Q_{\mathrm{rel}}((g,h);(\seedg,\seedh))=\mathbf q_{\mathrm{rel}}.
\ee
The absolute ADM charges of the solution exist (namely the surface integrals $E_r(g-\ebf),P_r(h)$ converge as $r\to+\infty$) if and only if those of its seed data set exist, and in that case
\be
\label{eq:ADM-seed-plus-modulator}
(E(g,h),P(g,h))=(E(\seedg,\seedh),P(\seedg,\seedh))+\mathbf q_{\mathrm{rel}}.
\ee
If $g-\ebf=\Obig_2(r^{-q})$, $h=\Obig_1(r^{-q-1})$ with $q>1/2$, the vacuum solution has finite standard ADM charges. Our seed data sets have finite standard ADM charges in the same range. For every $q>0$, including $q\leq1/2$, the relative limit~\eqref{eq:relative-modulator-vector} still exists.
\end{proposition}
\begin{proof}\bse
The surface integrals of a metric remainder $\Obig_1(r^{-p_\star})$ and of a momentum remainder $\Obig(r^{-p_\star-1})$ are $\Obig(r^{1-p_\star})$, hence tend to zero. The energy and momentum normalizations of $\mathcal S^\lambda$ and $\mathcal T_j^\lambda$ therefore give~\eqref{eq:relative-modulator-vector}. Subtracting these convergent surface integrals from the solution or seed data set proves the equivalence of convergence of $E_r(g-\ebf),P_r(h)$ or $E_r(\seedg-\ebf),P_r(\seedh)$, and the relation~\eqref{eq:ADM-seed-plus-modulator}.

We write $\gamma=g-\ebf$. The exact constraints have the form
\be
\label{eq:full-Q}
\begin{aligned}
\mathcal H(g,h)&=L \gamma+Q_H(\gamma,h),
&Q_H(\gamma,h)&\coloneqq Q(\gamma)+\tfrac12(\Tr_g h)^2-|h|_g^2,\\
\mathcal C^j(g,h)&=\del_i h^{ij}+Q_M(\gamma,h)^j,
&Q_M(\gamma,h)^j&\coloneqq\Gamma^i{}_{i\ell}h^{\ell j}+\Gamma^j{}_{i\ell}h^{i\ell}.
\end{aligned}
\ee
Here $Q(\gamma)\coloneqq R_{\ebf+\gamma}-L \gamma$. Uniform near-Euclidean bounds give
\be
\label{eq:full-Q-bound}
|Q_H|\leq C(|\gamma||\del^2 \gamma|+|\del \gamma|^2+|h|^2),\qquad |Q_M|\leq C|\del \gamma||h|.
\ee
These terms are $\Obig(r^{-2q-2})$, integrable for $q>1/2$. The divergence theorem on annuli makes the vacuum surface integrals Cauchy. For seed data sets, we include the integrable errors of~\eqref{eq:coupled-seed-norms}. Also $R_g=|h|_g^2-\tfrac12(\Tr_g h)^2$ is integrable, as is the scalar curvature of the seed data set by its Hamiltonian error and the same bound.

For comparison with geometric momentum, we use $K_{ij}-(\Tr_g K)g_{ij}=g_{ia}g_{jb}h^{ab}$. Replacing metric factors, normal, and area in $g_{ia}g_{jb}h^{ab}\nu_g^i\del_j\,dS_g/(8\pi)$ by Euclidean ones changes the surface integral by $\Obig(r^2|\gamma||h|)=\Obig(r^{1-2q})$. Thus the definitions agree for $q>1/2$. For relative metric and momentum differences of orders $r^{-1},r^{-2}$, the replacement error is $\Obig(r^{-q})$, which vanishes throughout $q>0$.
\ese\end{proof}

For $q\leq1/2$, the seed data sets of Section~\ref{section:full-constraints} and their projections in Section~\ref{section:slow-coupled-projection} typically fail to admit ADM charges, as we can now test explicitly starting from our finite exterior expansions. We group equal radial exponents and write the trace-reversed metric difference $\seedg-\ebf$ and momentum~$\seedh$, up to faster-decaying terms whose surface integrals tend to zero, as
\be
\label{eq:flux-polynomial-data}
\begin{aligned}
A&=\sum_{0<w\leq1}r^{-w}(a_w\theta^\flat\otimes\theta^\flat+\theta^\flat\otimes b_w+b_w\otimes\theta^\flat+C_w),\\
\seedh&=\sum_{0<w\leq1}r^{-w-1}(\alpha_w\theta\otimes\theta+\theta\otimes\beta_w+\beta_w\otimes\theta+D_w).
\end{aligned}
\ee
All coefficients are polynomials in $\ell=\log(r/R)$. The equations~\eqref{eq:first-divergence} give the polynomials governing these surface integrals
\be
\label{eq:flux-polynomials}
e_w(\ell)\coloneqq\frac1{16\pi}\int_{S^2}\bigl((\del_\ell+2-w)a_w-\trs C_w\bigr)\,d\theta, \quad v_w(\ell)\coloneqq\frac1{8\pi}\int_{S^2}(\alpha_w\theta+\beta_w)\,d\theta.
\ee
The surface integrals $(E_r(\seedg-\ebf),P_r(\seedh))$ are $\sum_w r^{1-w}(e_w,v_w)+o(1)$. Thus, $r\to+\infty$ convergence is equivalent to $(e_w,v_w)=0$ for every $w<1$ and $(e_1,v_1)$ being constant (rather than having positive degree in~$\ell$). The same conclusions hold when replacing the seed data set $(\seedg,\seedh)$ by its exact projection $(g,h)$ as harmonic corrections only add a constant and cannot remove divergences.

\subsubsection{Total charges from nonlinear constraint terms.}
The data sets considered in the present work are complete on~$\RR^3$ with no inner boundary.  Provided $(Q_H,Q_M)$ are integrable, which occurs for $q>1/2$, integrating~\eqref{eq:full-Q} on balls gives the exact identity
\be
\label{eq:total-four-integral}
(E,P)=-\Bigl(\frac1{16\pi}\int_{\RR^3}Q_H\,dx, \frac1{8\pi}\int_{\RR^3}Q_M\,dx\Bigr).
\ee
If $|\del^j \gamma|+\rho|\del^j h|\leq C\mu_q R^q\rho^{-q-j}$ (recall $\rho\simeq\max(R,r)$) for $j\leq2$, then~\eqref{eq:full-Q-bound} bounds the vector by $CR\mu_q^2$, which is quadratic in the dimensionless amplitude~$\mu_q=|t|/R^q$.  Analogously, data sets for Schwarzschild gluing, with $q=1$, have a total ADM energy--momentum that scales quadratically with the dimensionless mass $\mu=(|m_1|+|m_2|)/R$.
The first variation of~\eqref{eq:total-four-integral} at Euclidean data therefore vanishes. When the data set is incomplete, such as in the Schwarzschild metric~\eqref{eq:schwarzschild}, an inner boundary flux must be added to~\eqref{eq:total-four-integral}, which can thus be linear in the amplitude or mass.

%-----------------------------------------------------------------------------------------------------------------------------------------

\subsection{A localized harmonic energy parameter}\label{section-5-1}
We introduce an angularly supported tensor with unit energy flux. Its amplitude changes the harmonic energy coefficient without changing the prescribed Schwarzschild and Minkowski regions.

We fix $f\in C_c^\infty(\Lambda)$ with integral $16\pi$ and define a harmonic tensor on the end by
\be
\label{eq:mass-carrier}
T_f\coloneqq\frac{f(\theta)}r\Crad, \qquad \int_{S^2}f\,d\theta=16\pi.
\ee
The equations~\eqref{eq:spherical-L} and~\eqref{eq:relative-flux} give
\be
\label{eq:carrier-normalization}
LT_f=0,\qquad \mathcal F(T_f)=1.
\ee
Indeed, its trace reversal is $r^{-1}f\theta^\flat\otimes\theta^\flat$ and the radial component of its divergence is $r^{-2}f$. We cut it off inside the same fixed ball as before.

We add $sT_f$ to the leading exterior seed data set and repeat the finite corrections and interior extension, obtaining $\seedg(m,s)$. The recursion~\eqref{eq:polynomial-recursion} has homogeneous polynomial coefficients in $(m,s)$, supported compactly inside $\Lambda$. Thus the Schwarzschild mass $m$ and Euclidean exterior region are unchanged. The estimates of Section~\ref{section-3} hold uniformly with $\mu=(|m|+|s|)/R$ in a fixed small neighborhood.
The projected metric has the expansion
\be
\label{eq:parametric-expansion}
g(m,s)=\seedg(m,s)+\mathfrak m(m,s)\mathcal S^\lambda+\eta_g(m,s), \qquad \eta_g(m,s)\in C^{N,\alpha}_{p_\star,\underline P}(\Omega).
\ee

\begin{remark}[Relation with the construction of LeFloch and Nguyen]
LeFloch and Nguyen select a parameter in the seed data set to remove an unwanted harmonic term~\cite[Sections~6.2--6.3]{LeFlochNguyen}: equations~(6.23)--(6.24) give the mass correction, (6.28) its continuity, and (6.29)--(6.31) the selection; equations~(6.44)--(6.45) treat Schwarzschild--Minkowski data. Presently, we localize this parameter inside the gluing region, retaining the exterior metrics.
In the notation of~\cite[equation~(6.45)]{LeFlochNguyen}, we first set $t\coloneqq-4m/5$ and choose an angular function $\Phi$ equal to one on the prescribed Schwarzschild and Minkowski sectors, with zero spherical average. The remaining exterior Hessian term is canceled by solving
\be
\label{eq:Nguyen-selection}
\frac54s-\widetilde m(s,-4m/5)=0, \qquad \int_{S^2}\Phi\,d\theta=0.
\ee
The factor $5/4$ reflects their Hessian tensor and parametrization. We normalize our parameter by its ADM flux to prescribe the harmonic correction inside~$\Omega$; seed-to-solution projection already preserves the exterior metrics.
\end{remark}

%-----------------------------------------------------------------------------------------------------------------------------------------

\subsection{The energy flux equation and its nonlinear remainder}\label{section-5-2}
We express the harmonic energy coefficient~$\mathfrak m(m,s)$ in~\eqref{eq:parametric-expansion} in terms of the prescribed parameter and an integrable nonlinear remainder. This identity gives the quantitative estimate needed to select~$s$.

For a perturbation $\gamma$ of the Euclidean metric, we write the exact scalar-curvature identity as
\be
\label{eq:exact-Q}
R_{\ebf+\gamma}=L \gamma+Q(\gamma).
\ee
For uniformly small $\gamma$, the function $Q$ consists of terms quadratic or higher in $\gamma$ and its derivatives; in particular, we have
\be
\label{eq:Q-bound}
|Q(\gamma)|\leq C\bigl(|\gamma|\,|\del^2\gamma|+|\del \gamma|^2\bigr).
\ee
This is also the time-symmetric Euclidean specialization of the nonlinear expansion in~\eqref{eq:full-Q-bound}. Since $p>1/2$, the sub-harmonic estimates supplied by the uniform Newton construction imply
\be
\label{eq:solution-subharmonic}
|\del^j(g(m,s)-\ebf)|\leq C\mu R^p\max(R,r)^{-p-j},\qquad 0\leq j\leq2.
\ee
The bound~\eqref{eq:Q-bound} is integrable since $p>1/2$. Continuous dependence of each Newton step follows from the uniform inverse estimates; uniform quadratic convergence gives continuity of the limit in the same norm. Dominated convergence therefore applies to the integral of $Q$.

We denote by $\langle\chi\rangle_{S^2}\coloneqq(4\pi)^{-1}\int_{S^2}\chi\,d\theta$ the unweighted spherical average of the angular cutoff function~$\chi$ involved in the interpolation between Schwarzschild and Minkowski.

\begin{lemma}\label{lem:modulator-equation}
For the preceding family, there is a continuous function $\mathcal N(m,s)$ such that
\be
\label{eq:modulator-equation}
\mathfrak m(m,s)=-m\langle\chi\rangle_{S^2}-s+\mathcal N(m,s), \qquad |\mathcal N(m,s)|\leq \frac C R(|m|+|s|)^2.
\ee
\end{lemma}
\begin{proof}\bse
For a trace-reversed tensor of the form $r^{-1}(a_0\theta^\flat\otimes\theta^\flat+\theta^\flat\otimes b_0+b_0\otimes\theta^\flat+C_0)$, its flux is $(16\pi)^{-1}\int_{S^2}(a_0-\trs C_0)\,d\theta$. The angular divergence integrates to zero. The equations~\eqref{eq:schwarzschild-components} and~\eqref{eq:carrier-normalization} therefore yield
\be
\label{eq:seed-ADM}
\mathcal F(\seedg(m,s)-\ebf)=m\langle\chi\rangle_{S^2}+s.
\ee
The higher-order terms in the seed data set, including $r^{-2}\log(r/R)$, have zero flux. The same holds for $\eta_g(m,s)$ because $p_\star>1$.

We integrate~\eqref{eq:exact-Q} for the scalar-flat metric $g(m,s)$ over a large ball. There is no inner boundary since the metric is $C^2$ on $\RR^3$. The divergence theorem and the limit at infinity give
\be
\label{eq:mass-nonlinear}
\mathcal F(g(m,s)-\ebf) =-\frac1{16\pi}\int_{\RR^3}Q(g(m,s)-\ebf)\,dx \eqqcolon\mathcal N(m,s).
\ee
The flux decomposition~\eqref{eq:parametric-expansion}, followed by~\eqref{eq:seed-ADM}, proves~\eqref{eq:modulator-equation}. The equations~\eqref{eq:Q-bound}--\eqref{eq:solution-subharmonic} bound the integral by $CR\mu^2$. They also give continuity of~$\mathfrak m$ hence of $\mathcal N$ by dominated convergence, using continuous dependence of the projection in the same fixed weighted space.
\ese\end{proof}

%-----------------------------------------------------------------------------------------------------------------------------------------

\subsection{Selection of the relative energy}\label{section-5-3}
We solve the scalar parameter equation to prescribe the relative energy coefficient. The proof of the following statement is given in \hyperref[appendix:scalar-selection]{Appendix~\ref*{appendix:radial-selection}}.
\begin{corollary}[Prescribed scalar modulator]\label{thm:modulator-selection}
Fix the exponents, localization geometry, and cutoffs. For all $m\geq0$ and target coefficients $a\in\RR$ with $(m+|a|)/R$ sufficiently small, there exists a parameter $s_*$ such that
\be
\label{eq:target-modulator}
\mathfrak m(m,s_*)=a,\qquad s_*=-m\langle\chi\rangle_{S^2}-a+\Obig\bigl((m+|a|)^2/R\bigr).
\ee
In particular, for $a=0$ the exact scalar-flat metric differs from its selected seed data set by a super-harmonically decaying tensor throughout the gluing domain, while retaining the prescribed Schwarzschild and Euclidean metrics outside~$\Omega$.
\end{corollary}
The prescribed Schwarzschild mass $m$, the relative coefficient $\mathfrak m$, and the total ADM mass are distinct. The condition $\mathfrak m(m,s_*)=0$ sets the relative coefficient to zero, while the equation~\eqref{eq:mass-nonlinear} gives the total ADM mass as $\mathcal N(m,s_*)$, of order $m^2/R$ in this family.

The construction also applies to slow leading data sets when the difference of nonlinear errors is integrable. For a fixed $0<q<1$, we choose an admissible projection exponent $p$ with $p+q>1$ and add $sT_f$ to the slow seed data set. The finite recursion preserves its slow coefficients: every further term involving $s$ has exponent at least $1+q$, whereas the difference between the two seed data sets has energy flux $s$. The scalar specialization of Theorem~\ref{thm:vector-selection}, proven below, gives
\be
\label{eq:slow-modulator}
\mathfrak m(s)=\mathfrak m(0)-s+\mathcal R(s),\qquad |\mathcal R(s)-\mathcal R(s')|\leq C\varepsilon|s-s'|.
\ee
Here $\varepsilon$ bounds the small slow and harmonic amplitudes. The proof integrates the difference of the nonlinear constraints, whose radial order $r^{-p-q-2}$ is integrable because $p+q>1$; the individual integrals need not converge. We can therefore choose the scalar harmonic parameter throughout $0<q<1$, with uniqueness and Lipschitz dependence supplied by the contraction argument below.

%-----------------------------------------------------------------------------------------------------------------------------------------

\subsection{Selection of relative energy-momentum}\label{section-5-vector}
We select four harmonic parameters to prescribe relative energy-momentum. These parameters leave the slow leading metric and momentum profiles unchanged. We choose $\varphi\in C_c^\infty(\Lambda)$, nonnegative and positive on a nonempty open set, and define
\be
\label{eq:momentum-carriers}
A_{\varphi ij}\coloneqq\int_{S^2}\varphi\theta_i\theta_j\,d\theta,\qquad \varpi_{\varphi j}\coloneqq8\pi\varphi\sum_{k=1}^3(A_\varphi^{-1})_{jk}\theta_k,\qquad U_{\varphi j}\coloneqq r^{-2}\varpi_{\varphi j}\theta\otimes\theta.
\ee
The matrix $A_\varphi$ is positive definite: for $v\neq0$, the integral of $\varphi(v\cdot\theta)^2$ is positive because a nonzero linear function cannot vanish on an open subset of~$S^2$. The equations~\eqref{eq:first-divergence}, with radial exponent~$2$, give $\operatorname{div}_{\ebf}U_{\varphi j}=0$. Direct integration gives $P_r(U_{\varphi j})=e_j$, with the Cartesian unit vectors introduced in Section~\ref{section-4-1}. Together with the energy normalized harmonic tensor~$T_f$, these are four independent flux parameters, with compact angular support.

We write $z\coloneqq(s,v)\in\RR\times\RR^3$, add $(sT_f,\sum_jv_jU_{\varphi j})$ to a fixed slow exterior data set, and perform finite correction and interior extension. The seed data set $S(z)$ equals $(\ebf,0)$ outside~$\Omega$, and its seed-to-solution projection is $X(z)$. Addition of radial exponents in the recursion makes every further $z$-dependent term of metric order at least $r^{-1-q}$ or momentum order at least $r^{-2-q}$, up to logarithms; harmonic products decay faster. The relative ADM energy-momentum of the \emph{seed data sets} is therefore
\be
\label{eq:seed-four-difference}
\mathcal Q_{\mathrm{rel}}(S(z);S(z'))=z-z'.
\ee
Coefficients of slower terms, including exponent-one logarithms from the fixed slow pair, remain unchanged.

\begin{theorem}[Choice of relative energy-momentum with fixed slow leading profiles and exterior geometry]\label{thm:vector-selection}
In the context of Theorem~\ref{thm:coupled-projection}, assume that the slow decay exponent $q\in(0,1)$ and projection exponent~$p$ obey $p+q>1$. For sufficiently small slow amplitudes and a sufficiently small closed parameter ball $|z|\leq r_0$, there is a nonlinear remainder $\mathbf N$ such that the relative energy-momentum vector $\mathbf q_{\mathrm{rel}}(z)=(\mathfrak m(z),\mathfrak J(z))$ satisfies
\be
\label{eq:vector-modulator-equation}
\mathbf q_{\mathrm{rel}}(z)=\mathbf q_{\mathrm{rel}}(0)-z+\mathbf N(z),\qquad \mathbf N(0)=0,\qquad |\mathbf N(z)-\mathbf N(z')|\leq L|z-z'|,
\ee
where $0\leq L<1$. Every target $\mathbf q$ with $|\mathbf q-\mathbf q_{\mathrm{rel}}(0)|\leq(1-L)r_0$ is attained by a unique parameter~$z$ in this ball, namely $\mathbf q_{\mathrm{rel}}(z)=\mathbf q$. The selected parameter~$z$ is Lipschitz in the target~$\mathbf q$ with Lipschitz constant $(1-L)^{-1}$. In particular, after taking the fixed slow amplitudes sufficiently small, all four modulators can be set to $\mathbf q=0$ without changing the leading slow pair or the prescribed exterior values.
\end{theorem}
\begin{proof}\bse
We use the pair space $\mathcal X_p\coloneqq C^{N,\alpha}_{p,\underline P}\times C^{N,\alpha}_{p+1,\underline P}$ with the sum norm as well as the source space $\mathcal Y_p=(C^{N-2,\alpha}_{p+2,\underline P-2}\cap L^2_{p+2,P})\times(C^{N-1,\alpha}_{p+2,\underline P-1}\cap L^2_{p+2,P})$, obtained by replacing $(\mathcal H,\mathcal C)$ in~\eqref{eq:full-error-norm} at $b=p$ by an arbitrary scalar--vector source.

For any $z$ with $|z|\leq r_0$, we consider the seed-to-solution projection from $S(z)$ to $X(z)$, which is given by $X(z)=S(z)+\mathcal DU_z$, with scalar--vector potential $U_z\coloneqq(u_z,Z_z)$ and fixed Euclidean reference adjoint $\mathcal D(u,Z)\coloneqq(\mathcal D_gu,\mathcal D_hZ)$, defined by
\be
\label{eq:full-reference-adjoint}
\mathcal D_gu\coloneqq\lambda^{2P}\rho^{3-2p}(\Hess u-\Delta u\,\ebf),\qquad (\mathcal D_hZ)^{ij}\coloneqq-\tfrac12\lambda^{2P}\rho^{1-2p}(\del_iZ_j+\del_jZ_i).
\ee
We set $w\coloneqq X(z)-X(z')-S(z)+S(z')=\mathcal D(U_z-U_{z'})$.
Let us prove first a bound on the difference of seed data sets,
\be
\label{eq:Sz-Szp-bound}
\|S(z)-S(z')\|_{\mathcal X_p}+\bigl\|(L,\operatorname{div}_{\ebf})(S(z)-S(z'))\bigr\|_{\mathcal Y_p}\leq C|z-z'| .
\ee
All terms in the finite expansions defining seed data sets have compact angular support inside~$\Lambda$ and compact support away from boundaries of~$\Omega$ in the ball~$B_{4R}$, hence are supported in a region where $\lambda$ is uniformly bounded below; as a result, angular weights in the norms can be absorbed into the multiplicative constant~$C$.  The finite expansions of $S(z)$ and $S(z')$ differ by harmonic parameter terms and their nonlinear consequences, whose radial decay is sufficient to make the weighted norms finite.  Finally, the finite expansions defining seed data sets are polynomial in~$z$, and $S(z)-S(z')$ is thus divisible by $z-z'$ (as a polynomial), which provides the $|z-z'|$ factor in the upper bound.

The differences $X(z)-(\ebf,0)$ are small in $C^{N,\alpha}_{q,\underline P}\times C^{N,\alpha}_{q+1,\underline P}$ by some parameter $\varepsilon>0$. The differentiated form of~\eqref{eq:full-Q-bound} gives
\be
\label{eq:full-Q-difference}
\|\mathbf Q(X(z))-\mathbf Q(X(z'))\|_{\mathcal Y_p} \leq C\varepsilon\|X(z)-X(z')\|_{\mathcal X_p}, \qquad \mathbf Q(g,h)\coloneqq(Q_H(g-\ebf,h),Q_M(g-\ebf,h)).
\ee
The worst pointwise products have order $\rho^{-p-q-2}$ and angular factor $\lambda^{2\underline P-2}$. Their weighted square integrals have radial factor $\rho^{-2q-1}$; the angular exponent is positive by~\eqref{eq:localization-exponents}. This proves the integrability required by the space~$\mathcal Y_p$ in~\eqref{eq:full-Q-difference}; local products prove its differentiated and H\"older parts.

Subtracting the two exact constraints and applying~\eqref{eq:full-local-inverse}, below, at the Euclidean seed data set and \eqref{eq:Sz-Szp-bound} to bound $S(z)-S(z')$ yields $\|w\|_{\mathcal X_p}\leq C|z-z'|+C\varepsilon\|w\|_{\mathcal X_p}$. Absorbing the last term gives, for sufficiently small~$\varepsilon$ and a different constant~$C$,
\be
\label{eq:full-projection-difference}
\|X(z)-X(z')\|_{\mathcal X_p}\leq C|z-z'|, \qquad |\mathbf Q(X(z))-\mathbf Q(X(z'))|\leq C\varepsilon|z-z'|\rho^{-p-q-2}.
\ee
The last bound is integrable because $p+q>1$. Integration of the difference of~\eqref{eq:full-Q}, with~\eqref{eq:seed-four-difference} and the normalized harmonic remainders, proves~\eqref{eq:vector-modulator-equation}, with
\be
\label{eq:vector-relative-remainder}
\mathbf N(z)=-\frac1{16\pi}\operatorname{diag}(1,2,2,2) \int_{\RR^3}\bigl(\mathbf Q(X(z))-\mathbf Q(X(0))\bigr)\,dx.
\ee
The same estimate between two parameters gives $L=C\varepsilon<1$ after shrinking the amplitudes.

For $\mathbf q$ sufficiently close to $\mathbf q_{\mathrm{rel}}(0)$ as stated, the map $z\mapsto\mathbf q_{\mathrm{rel}}(0)-\mathbf q+\mathbf N(z)$ preserves the closed parameter ball $|z|\leq r_0$ and contracts by a Lipschitz constant $L<1$. Its fixed point is thus unique. Subtracting the equations for $\mathbf q,\mathbf q'$ gives $(1-L)|z(\mathbf q)-z(\mathbf q')|\leq|\mathbf q-\mathbf q'|$. The ball of harmonic parameters can be fixed first, while $\mathbf q_{\mathrm{rel}}(0)$ tends to zero with the slow amplitudes by harmonic projection; shrinking them permits $\mathbf q=0$.
\ese\end{proof}

%-----------------------------------------------------------------------------------------------------------------------------------------

\subsection{Relative energy-momentum and the total ADM charge map}\label{section-5-prescription}
We compare the selected harmonic parameters~$z$ with the charges of the exact solutions.

Theorem~\ref{thm:vector-selection} prescribes the relative charge between the exact solution and its selected seed data set, while~\eqref{eq:seed-four-difference} gives the relative charge between seeed data sets in the family. By additivity, the relative charge of two \emph{exact} solutions in that family is thus
\be
\label{eq:exact-relative-cancellation}
\mathcal Q_{\mathrm{rel}}(X(z);X(z')) =z+\mathbf q_{\mathrm{rel}}(z)-z'-\mathbf q_{\mathrm{rel}}(z')=\mathbf N(z) - \mathbf N(z').
\ee
The modulator cancels the seed data set's linear charge, regardless of accuracy of the seed (how large $p_A$ is); invertibility of $z\mapsto\mathbf q_{\mathrm{rel}}(z)$ therefore has no bearing on invertibility of the map $z\mapsto\mathcal Q_{\mathrm{rel}}(X(z);X(0))=\mathbf N(z)$ describing the ADM charge of the exact vacuum data set after the seed-to-solution projection.

For the Schwarzschild--Minkowski family, the equations~\eqref{eq:seed-ADM}--\eqref{eq:mass-nonlinear} give the particularly concrete relations
\be
\label{eq:Sch-total-ADM}
E_{\mathrm{ADM}}(g(m,s))=m\langle\chi\rangle+s+\mathfrak m(m,s)=\mathcal N(m,s),\qquad P_{\mathrm{ADM}}=0.
\ee
When $s=s_*$ is adjusted to ensure $\mathfrak m=0$, the total energy is $E_{\mathrm{ADM}}=m\langle\chi\rangle+s_*=\Obig((|m|+|s_*|)^2/R)$, whereas the Schwarzschild angular sector retains mass parameter~$m$. Prescribing total energy~$E_*$ would require solving $\mathcal N(m,s)=E_*$, which is at best only feasible for quadratically small~$E_*$.  We address the question of achieving a chosen total ADM charge in Section~\ref{section-ADM-prescription}, based on compact perturbations instead of the asymptotic harmonic parameters used in Theorem~\ref{thm:vector-selection}.

%============================================================================

\section{Choosing energy-momentum: total ADM charges and compact perturbations}\label{section-ADM-prescription}
In Section~\ref{section-5}, we chose the charges relative to the seed data set by varying harmonic parameters. We now choose the total ADM vector by compactly supported perturbations of a seed data set. Theorem~\ref{thm:effective-ADM} shows which ADM vectors can be obtained from compact perturbations of the Euclidean data set $(\ebf,0)$. The following propositions calculate the quadratic changes in energy and momentum; we then generalize the construction to perturbations of non-trivial data sets such as those involved in Schwarzschild gluing. Figure~\ref{fig:relative-total-charges} compares the relative and total charge constructions.

%-----------------------------------------------------------------------------------------------------------------------------------------

\begin{theorem}[Choice of sufficiently small total ADM vectors near a fixed future timelike direction]\label{thm:effective-ADM}
Fix the exponents, stable localization geometry, and radial cutoff as in Sections~\ref{section-2} and~\ref{section-3-1}. Fix a vector $v_0\in\RR^3$ with $|v_0|<1$. There exist four smooth compact perturbation profiles $(\gamma_\ell,\pi_\ell)$ for $1\leq\ell\leq 4$ with disjoint supports in the interior of~$\Omega$, and constants $\varepsilon_0,\eta>0$, such that every vector $(E_*,P_*)$ of the form
\be
\label{eq:prescribed-ADM-target}
E_*=R\varepsilon,\qquad P_*=R\varepsilon v, \qquad 0<\varepsilon<\varepsilon_0,\qquad |v-v_0|<\eta,
\ee
is the exact ADM energy-momentum of a complete vacuum data set obtained by seed-to-solution projection of a seed data set $(\ebf,0)+\sum_\ell a_\ell(\gamma_\ell,\pi_\ell)$ for some $a\in\RR^4$. The solution equals $(\ebf,0)$ outside~$\Omega$ and inside the gluing domain it admits an expansion
\be
g = \ebf + E_* \mathcal S^\lambda + \eta_g , \qquad
h = \sum_{j=1}^3 P_{*j} \mathcal T^\lambda_j + \eta_h , \qquad
(\eta_g,\eta_h) \in C_{p_\star,\underline P}^{N,\alpha} \times C_{p_\star+1,\underline P}^{N,\alpha} ,
\ee
for $1<p_\star<\min\bigl(p_{3,p}^{\lambda},1+q\bigr)$.
The solution has nonzero second fundamental form even for $v=v_0=0$, which gives non-time-symmetric initial data sets with arbitrary sufficiently small ADM mass.
\end{theorem}

\begin{figure}
\centering
\includegraphics[width=\linewidth]{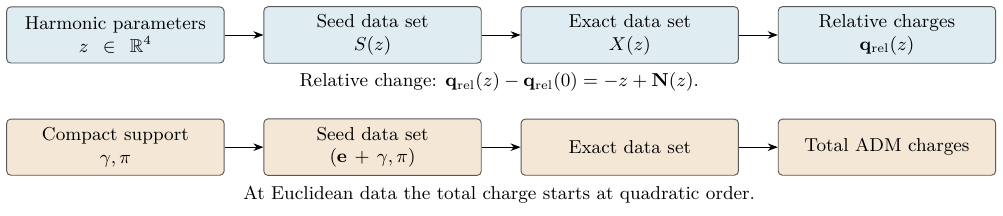}
\caption{Choosing relative and total energy-momentum. Harmonic parameters (Section~\ref{section-5}) vary the relative charges. Compactly supported, trace-free and divergence-free perturbations of Euclidean data (Section~\ref{section-ADM-prescription}) give total ADM charges whose leading change is quadratic. Both constructions apply the seed-to-solution projection to suitably chosen seed data sets; the resulting exact data sets coincide with prescribed data outside~$\Omega$ and retain the slow leading profiles inside~$\Omega$, with differences arising at harmonic radial decay or higher.}
\label{fig:relative-total-charges}
\end{figure}

%-----------------------------------------------------------------------------------------------------------------------------------------

\subsection{Compact transverse-traceless seed data sets and quadratic ADM charges}\label{section-ADM-quadratic}
The seed data sets $(\ebf,0)+\sum_\ell a_\ell(\gamma_\ell,\pi_\ell)$ have vanishing ADM charges, so the ADM charge of the exact data sets comes entirely from the energy-momentum modulators $\mathcal S^\lambda,\mathcal T^\lambda_j$.
It is conveniently obtained as~\eqref{eq:total-four-integral} in terms of quadratic terms~$Q_H,Q_M$, and we now compute these charges for general compact transverse-traceless perturbations $(\gamma,\pi)$ supported in~$\Omega$.

\begin{proposition}[Quadratic ADM energy-momentum]\label{lem:TT-quadratic}
Consider smooth compactly supported symmetric tensors $\gamma,\pi$ in~$\Omega$, both transverse and trace free with respect to the Euclidean metric~$\ebf$:
\be
\label{eq:TT-conditions}
\del_i \gamma_{ij}=0,\quad \Tr \gamma=0,\qquad \del_i\pi^{ij}=0,\quad \Tr\pi=0.
\ee
The seed data set $S(a)=(\ebf+a \gamma,a\pi)$ has constraint error $\Obig(a^2)$ in every fixed accuracy norm. For its seed-to-solution projection~$X(a)$, the total ADM vector has expansion
\be
\label{eq:TT-ADM-quadratic}
\begin{aligned}
E(a) \coloneqq E_{\mathrm{ADM}}(X(a))
& =a^2\mathscr E(\gamma,\pi)+\Obig(a^3),
&\mathscr E(\gamma,\pi)&\coloneqq\frac1{16\pi}\int\bigl(\tfrac14|\del \gamma|^2+|\pi|^2\bigr)dy,\\
P_j(a) \coloneqq P_{\mathrm{ADM}\,j}(X(a))
&=a^2\mathscr P_j(\gamma,\pi)+\Obig(a^3),
&\mathscr P_j(\gamma,\pi)&\coloneqq\frac1{16\pi}\int\pi^{i\ell}\del_j \gamma_{i\ell}\,dy.
\end{aligned}
\ee
For a seed $(\ebf,0)+\sum_\ell a_\ell(\gamma_\ell,\pi_\ell)$ with finitely many such pairs with pairwise disjoint supports, the same expansion holds with $a^2\mathscr E$ replaced by $\sum_\ell a_\ell^2\mathscr E(\gamma_\ell,\pi_\ell)$ and likewise for the momentum.
The vector $(E(a),P_j(a))$ is future-timelike up to $\Obig(|a|^3)$ corrections; in other words $\max(0,-E(a)^2+|P(a)|^2)=\Obig(|a|^5)$.
\end{proposition}
The proof is given in Appendix~\ref{appendix:quadratic-ADM}.

\begin{proposition}[Compact profiles with prescribed approximate velocity]\label{lem:TT-packets}
Fix an open Euclidean ball in~$\RR^3$ and a unit vector $n$. There exist transverse-traceless $(\gamma,\pi)$ supported in this ball, with $\mathscr E=1$ and $\mathscr P$ arbitrarily close to~$n$. For every such pair $|\mathscr P|\leq\mathscr E$.
\end{proposition}
The proof is given in Appendix~\ref{appendix:compact-TT}.

%-----------------------------------------------------------------------------------------------------------------------------------------

\subsection{Proof of Theorem~\ref{thm:effective-ADM}}\label{section-ADM-targets}

To reduce clutter we work throughout in units with $R=1$.
Rescaling by $g(x)=\widehat g(x/R)$, $h(x)=R^{-1}\widehat h(x/R)$ multiplies the constraints by $R^{-2}$ and both fluxes by~$R$, which restores the factors of~$R$ in~\eqref{eq:prescribed-ADM-target}.

We choose four disjoint compact profiles whose quadratic charges $(\mathscr E,\mathscr P_j)$ span a neighborhood of the prescribed timelike direction. We then solve the exact charge equation by contraction after rescaling the amplitudes by the square root of the target energy.

\begin{proof}[Proof of Theorem~\ref{thm:effective-ADM}]\bse
We first choose four unit vectors whose convex hull contains~$v_0$ in its interior. An explicit choice is obtained by writing $v_0=t n$ with $0\leq t<1$, taking one vertex $-n$, and three equally spaced vertices on the latitude $n\cdot y=c$ with $t<c<1$. The barycentric coefficients of~$v_0$ are $(c-t)/(1+c)$ for $-n$ and $(1+t)/(3(1+c))$ for each other vertex, all positive. For $v_0=0$ any unit~$n$ may be used.

We take four disjoint balls in the fixed interior of~$\Omega$. By Proposition~\ref{lem:TT-packets}, we choose normalized pairs $(\gamma_\ell,\pi_\ell)$ there with velocities $v_\ell$ sufficiently close to these vertices. The matrix $B$ with columns $b_\ell\coloneqq(1,v_\ell)\in\RR^4$ is invertible, and $c(v)\coloneqq B^{-1}(1,v)$ has positive components for all $v$ in a neighborhood of~$v_0$. We consider the seed data sets
\be
\label{eq:four-compact-seeds}
S(a)\coloneqq\Bigl(\ebf+\sum_{\ell=1}^4a_\ell \gamma_\ell,\ \sum_{\ell=1}^4a_\ell\pi_\ell\Bigr)
\ee
and their seed-to-solution projection~$X(a)$.
Their constraint errors are quadratic in~$a$ and compactly supported, hence $p_A$ can be taken arbitrarily large.  Proposition~\ref{lem:TT-quadratic} gives the exact charge map
\be
\label{eq:effective-charge-map}
F(a)\coloneqq(E_{\mathrm{ADM}},P_{\mathrm{ADM}})(X(a))=\sum_{\ell=1}^4a_\ell^2 b_\ell+\mathcal R(a),\qquad |\mathcal R(a)|\leq C|a|^3,\qquad \|D\mathcal R(a)\|\leq C|a|^2,
\ee
where $D$ denotes derivatives with respect to~$a$.
The choice $a_\ell^2=\varepsilon c(v)_\ell$ would make the first term equal to the desired ADM energy-momentum $\varepsilon (1,v)$.
We thus seek $a_\ell=\sqrt{\varepsilon y_\ell}$, with $y$ close to $c(v)$, characterized as a fixed point of
\be
\label{eq:effective-ADM-fixed-point}
y \mapsto c(v)-B^{-1}\varepsilon^{-1}\mathcal R(\sqrt{\varepsilon y}).
\ee
For $y$ in a ball of small ($\varepsilon$-independent) radius around~$c(v)$, all components are bounded away from zero, so this map makes sense, and the remainder $B^{-1}\varepsilon^{-1}\mathcal R(\sqrt{\varepsilon y})$ and its $y$ derivative are bounded by $C\sqrt\varepsilon$. For sufficiently small~$\varepsilon$, the map~\eqref{eq:effective-ADM-fixed-point} preserves the smaller closed ball centered at~$c(v)$ of radius $C\sqrt\varepsilon$, and contracts distances since derivatives are bounded by $C\sqrt\varepsilon<1$. Thus, there exists a fixed point~$y$, with positive components, which gives a solution with the prescribed total ADM charges~\eqref{eq:prescribed-ADM-target} when $R=1$ (and by rescaling, all $R>0$).

The momentum tensor of the projected data is $\sum_\ell a_\ell\pi_\ell+\Obig(\varepsilon)$. At a nonzero point of one disjointly supported $\pi_\ell$, its leading size is a positive multiple of $\sqrt\varepsilon$; hence the second fundamental form $K\neq0$ for small $\varepsilon>0$.

For a prescribed invariant mass $M_*>0$ and a fixed $|v|<1$ in this neighborhood, we take $E_*=M_*/\sqrt{1-|v|^2}$ and $P_*=E_*v$ in~\eqref{eq:prescribed-ADM-target}. This imposes that mass exactly whenever $E_*/R<\varepsilon_0$. The special choice $v=v_0=0$ gives $E_*=M_*$ and zero total momentum, while the second fundamental form remains nonzero.
\ese\end{proof}

%-----------------------------------------------------------------------------------------------------------------------------------------

\subsection{Prescribed ADM energy-momentum with fixed exterior data}\label{section-ADM-background}
Instead of perturbing the Euclidean data set $(\ebf,0)$ as in Theorem~\ref{thm:effective-ADM}, we now start from a small exact background solution with some decay exponent $p_G\in(0,1]$, obtained for instance from our Schwarzschild gluing in Section~\ref{section-3} or slow-decaying data sets in Section~\ref{section:slow-coupled-projection}.  We add compact perturbations to control relative ADM charges or absolute ADM charges when they are well-defined.
The relative or absolute ADM energy-momentum vectors produced in this way are necessarily future-timelike.
This construction retains the slowly decaying leading terms or the prescribed exterior Schwarzschild metric.

\begin{proposition}[Choice of small relative or total ADM charges with a background]\label{prop:fixed-background-ADM}
Fix the exponents, stable localization geometry, and radial cutoff as in Sections~\ref{section-2} and~\ref{section-3-1}, with a geometric decay exponent $p_G=q\in(0,1]$ and a projection exponent~$p$ satisfying $p+q>1$.
Also fix a vector $v_0\in\RR^3$ with $|v_0|<1$.  There exist four smooth compact perturbation profiles $(\gamma_\ell,\pi_\ell)$ for $1\leq\ell\leq 4$ with disjoint supports in the interior of~$\Omega$, and constants $c_0,\varepsilon_0,\eta>0$ with the following property.
For any complete $C^{N,\alpha}$ vacuum background $X_0=(g_0,h_0)$ that is asymptotically flat and such that the norm $\nu_0\coloneqq \|(g_0-\ebf,h_0)\|_{C^{N,\alpha}_q\times C^{N,\alpha}_{q+1}}$ is bounded as $\nu_0\leq c_0 R^q\sqrt\varepsilon_0$, any vector $(E_*,P_*)$ of the form~\eqref{eq:prescribed-ADM-target} with an additional lower bound on~$\varepsilon$, namely
\be
\label{eq:prescribed-ADM-target-bis}
E_*=R\varepsilon,\qquad P_*=R\varepsilon v, \qquad (R^{-q}\nu_0/c_0)^2 < \varepsilon<\varepsilon_0,\qquad |v-v_0|<\eta,
\ee
is the relative ADM energy-momentum of a complete vacuum data set~$X(a)$ obtained as
the seed-to-solution projection of $(g_0,h_0)+\sum_\ell a_\ell(\gamma_\ell,\pi_\ell)$ for some $a\in\RR^4$, namely
\be
\label{eq:fixed-background-relative-target}
\mathcal Q_{\mathrm{rel}}(X(a);X_0)=(E_*,P_*).
\ee
The solution $X(a)$ equals $X_0$ outside the gluing domain, and within~$\Omega$ it differs by terms that decay at harmonic rate, specifically $X(a)=(g(a),h(a))$ admits the expansion
\be
g(a) = g_0 + E_* \mathcal S^\lambda + \eta_g , \qquad
h(a) = h_0 + \sum_{j=1}^3 P_{*j} \mathcal T^\lambda_j + \eta_h , \qquad
(\eta_g,\eta_h) \in C_{p_\star,\underline P}^{N,\alpha} \times C_{p_\star+1,\underline P}^{N,\alpha} .
\ee
Moreover, if $q>1/2$, one can \emph{instead} prescribe the absolute ADM vector to be $\mathcal Q_{\mathrm{rel}}(X(a);(\ebf,0))=(E_*,P_*)$.
\end{proposition}

\begin{proof}
Without loss of generality we work at $R=1$.  We use notation from the proof of Theorem~\ref{thm:effective-ADM}.  There, we built four compact profiles $(\gamma_\ell,\pi_\ell)$ with $1\leq\ell\leq 4$ whose quadratic charges $(\mathscr E,\mathscr P)=b_\ell=(1,v_\ell)$ are such that $v_0$ belongs to the interior of the convex hull of the~$v_\ell$, we introduced the matrix~$B$ with columns~$b_\ell$ and $c(v)=B^{-1}(1,v)$, which has positive components for $v$ sufficiently close to~$v_0$ as stated in~\eqref{eq:prescribed-ADM-target-bis}.

We consider the seed data set $S(a)\coloneqq X_0+\sum_\ell a_\ell(\gamma_\ell,\pi_\ell)$ and its seed-to-solution projection~$X(a)$. Since $X_0$ is exact, the constraint error of $S(a)$ is compact and bounded by $C(\nu_0|a|+|a|^2)$ in~$\mathcal Y_p$. The projection correction $X(a)-S(a)$ has the same bound in~$\mathcal X_p$. Thus the difference $X(a)-X_0$ belongs to $\mathcal X_p$, and the integrals of the nonlinear differences are finite: cross terms with the background have order $r^{-p-q-2}$, while products of differences have order $r^{-2p-2}$. Both are integrable by $p+q>1$ and $p>1/2$.

The inverse and parameter-derivative estimates of Proposition~\ref{lem:TT-quadratic} are uniform for small backgrounds. At $a=0$, the linearized constraint $D\mathcal G_{X_0}(\gamma_\ell,\pi_\ell)$ is $\Obig(\nu_0)$, since its Euclidean value vanishes; the inverse gives an $\Obig(\nu_0)$ first projection derivative. Every change in a linear or quadratic charge coefficient contains this derivative or a background factor, and is $\Obig(\nu_0)$ by the preceding radial integrals. The same product estimates bound third derivatives uniformly. Thus the relative ADM energy-momentum of $X(a)$ with respect to~$X_0$ is
\be
\label{eq:background-charge-estimate}
F_{X_0}(a)=\sum_\ell a_\ell^2b_\ell+\mathcal R_{X_0}(a),\qquad |\mathcal R_{X_0}(a)|\leq C(\nu_0|a|+|a|^3),\quad \|D\mathcal R_{X_0}(a)\|\leq C(\nu_0+|a|^2).
\ee
With $a_\ell=\sqrt{\varepsilon y_\ell}$, we seek fixed points of $y\mapsto c(v)-B^{-1}\varepsilon^{-1}\mathcal R_{X_0}(\sqrt{\varepsilon y})$, whose second term and its first derivative are bounded by $C(\nu_0/\sqrt\varepsilon+\sqrt\varepsilon)$. Choosing $c_0$ and $\varepsilon_0$ small makes this bound small (since $\nu_0<c_0\sqrt\varepsilon$), which lets the fixed-point proof near~\eqref{eq:effective-ADM-fixed-point} apply, proving~\eqref{eq:fixed-background-relative-target}.

For $q>1/2$,~\eqref{eq:total-four-integral} bounds the background charge by $C\nu_0^2$. Subtracting it from an absolute target changes the normalized equation by at most $C\nu_0^2/\varepsilon$, small after reducing $c_0$; the same contraction proof applies. The seed data set changes only on the four compact supports, and projection adds harmonic terms and super-harmonic remainders, faster than the leading pair for $q<1$. Localization preserves the exterior data.
\end{proof}

For fixed $1/2<q\leq 1$, nonzero slow metric and momentum profiles with amplitude proportional to $\varepsilon$ give $\nu_0=\Obig(\varepsilon)<c_0\sqrt\varepsilon$. We can therefore retain these profiles and prescribe any sufficiently small ADM vector in the chosen future timelike direction. For $0<q\leq1/2$, we can still choose the relative energy and momentum as in~\eqref{eq:fixed-background-relative-target}, but convergence of the absolute ADM fluxes requires the fine-tuning explained below~\eqref{eq:flux-polynomials}.

It is also instructive to specialize the previous result to the Schwarzschild gluing of Section~\ref{section-3}.
In that context, we consider the vacuum data set $(g,h)$ obtained by seed-to-solution projection from the sum of the exact solution $g_{m_1,m_2}$ in Theorem~\ref{thm:fixed-radius} and some compact profiles.  Since $g_{m_1,m_2}$ itself is obtained by projection from the interpolating metric $\seedg(m_1,m_2)$ in~\eqref{eq:global-seed}, the exact same data set $(g,h)$ is obtained by seed-to-solution projection from the sum of $\seedg(m_1,m_2)$ and the compact profiles.  At infinity, this seed metric in turn is $g_{m_1,m_2}^{[1]}+\Obig(r^{-2})$.

\begin{proposition}[Schwarzschild gluing and total ADM mass]\label{prop:Sch-mass-prescribed}
Fix the exponents, localization geometry, cutoffs, as in Theorem~\ref{thm:fixed-radius}, including the reference radius $R>0$ and dimensionless mass threshold $\mu_*>0$.  There exists small constants $c_0,\varepsilon_0>0$ such that for any mass parameters $M_*,m_1,m_2$ with
\be
\Bigl(\frac{|m_1|+|m_2|}{c_0R}\Bigr)^2 < \frac{M_*}{R} < \varepsilon_0,
\ee
there exists a complete asymptotically flat vacuum data set~$(g,h)$ with the following properties.  It coincides with the prescribed Schwarzschild data sets $(g_{m_1}^{\mathrm{Sch}},0)$ and $(g_{m_2}^{\mathrm{Sch}},0)$ in the two exterior regions outside~$\Omega$ and having total ADM energy~$M_*$ and total momentum~$0$.
It admits an expansion
\be
g=g_{m_1,m_2}^{[1]}+\mathfrak m\,\mathcal S^\lambda+\eta_g, \qquad
\eta_g \in C^{N,\alpha}_{p_\star,\underline P} , \qquad
h \in C^{N,\alpha}_{p_\star+1,\underline P} ,
\ee
where (with $\mathcal F$ defined in~\eqref{eq:relative-flux})
\be
\mathfrak m = M_* - \mathcal F(g_{m_1,m_2}^{[1]}) .
\ee
The data set $(g,h)$ can be chosen to have nonzero second fundamental form or alternatively to be time-symmetric so that $g$ is a scalar flat metric and $h=0$.
As an example, for $M_* = \mathcal F(g_{m_1,m_2}^{[1]})$ this gives a scalar flat metric $(g,0)$ that coincides super-harmonically with the interpolation metric $(g_{m_1,m_2}^{[1]},0)$.
As another example, for sufficiently small masses $m_1,m_2$ with $m_1>0$ one can select $M_*=m_1$ so that the total ADM mass coincides with the Schwarzschild mass in one angular region.

More generally, for each fixed $v_0\in\RR^3$ with $|v_0|<1$, one can choose data sets with nonzero second fundamental form, total ADM energy~$m$, and total ADM momentum~$mv_0$, while retaining the Schwarzschild metric of mass parameter~$m$ in the prescribed exterior region.
\end{proposition}

The proof of Proposition~\ref{prop:Sch-mass-prescribed} is given in \hyperref[appendix:Sch-ADM]{Appendix~\ref*{appendix:radial-selection}} and is mostly a corollary of Proposition~\ref{prop:fixed-background-ADM}.  However, to produce a time-symmetric data set part of the proof has to be adjusted to use a single time-symmetric (transverse, trace-free) profile $(\gamma,0)$ instead of multiple profiles $(\gamma_\ell,\pi_\ell)$.
The compact amplitudes have order $\sqrt{M_*/R}$, larger than the Schwarzschild amplitude $(|m_1|+|m_2|)/R$, in such a way that their ADM charges have the required order~$M_*$ despite being quadratic as stated in~\eqref{eq:Sch-total-ADM}. The radius, profiles, and accuracy are fixed before choosing masses small.

%============================================================================

\begin{acknowledgments}
The authors are grateful to Xuantao Chen for discussions about his work, especially the backward construction from infinity in~\cite{ChenKlainerman}. The authors were partially supported by the Agence Nationale de la Recherche (ANR), project ANR-23-CE40-0010-02, Einstein-PPF, ``Einstein constraints: past, present, and future'', and by the European Union's Horizon Europe programme through the Marie Sk\l odowska-Curie Staff Exchanges project 101131233, EinsteinWaves, ``Einstein gravity and nonlinear waves: physical models, numerical simulations and data analysis''.
\end{acknowledgments}

\phantomsection
\subsubsection*{Note on AI use}

After we had written the draft of the paper, including all of the ideas and technical arguments, we used ChatGPT 6 to help reorganize the manuscript, improve its presentation, and produce all of the figures. This was followed by a human-only review to ensure mathematical validity and clarity.

\appendix
\renewcommand{\theHequation}{appendix.\Alph{section}.\arabic{equation}}

%============================================================================

\section{Local weighted projection estimates and dependence on parameters}\label{appendix:local-projection}
We use the seed-to-solution projection and the harmonic and super-harmonic decay estimates established in our companion papers~\cite{BLF--PLF,BLF--PLF-stability}. We give the details needed for two applications in this article: the projection estimates in our local weighted H\"older norms, and the dependence of the projection on the seed data sets and their parameters. The first point allows us in Sections~\ref{section-3}--\ref{section:slow-coupled-projection} to apply the projection to the seed data sets constructed in Sections~\ref{section-2}--\ref{section:full-constraints}. The second supplies the local Lipschitz bounds and parameter derivatives used to choose the relative and total ADM charges in Sections~\ref{section-5}--\ref{section-ADM-prescription}.

We normalize $0<\lambda\leq1$ and choose a constant $0<c<1/3$. We set $d(x)\coloneqq c\lambda(x)\rho(x)$, which is comparable to $\min(\rho(x),\operatorname{dist}(x,\Omega^c))$ if $\Omega\neq\RR^3$ and to $\rho(x)$ otherwise. The constant $c$ is sufficiently small that the weights on $B_{2d(x)}(x)$ are uniformly comparable to their values at its center.
As stated below~\eqref{eq:weighted-L2}, the weighted H\"older norm $C^{N,\alpha}_{b,a}(\Omega)$ of some function~$f$ includes angular and radial weights as
\bse\label{equa-norm-Hol}
\be\label{equa-norm-Hol-a}
\aligned
\| f \|_{\Omega, b,a}^{l, \alpha}
= \| f \|_{C_{b,a}^{l, \alpha}(\Omega)}
\coloneqq \sum_{|L| \leq l}   \sup_\Omega \big( \lambda^{-a+|L|} \, \rho^{b+|L|} \, |\del^L f| \big) 
+ \sum_{|L| = l}   \sup_\Omega \big( \lambda^{-a + |L|+ \alpha} \, \rho^{b+|L|+ \alpha} \, [\![ \del^L f ]\!]_{\alpha} \big).
\endaligned
\ee
Here, $L$~represents a multi-index and the H\"older bracket is
\be
[\![ f ]\!]_{\Omega,\alpha} (x) 
\coloneqq \sup_{y,z \in B_{d(x)}(x)} \frac{|f(y)-f(z)|}{|y-z|^\alpha} ,
\qquad x \in \Omega .
\ee
\ese
Restricting the supremum to $z=x$ and $y\in B_{d(x)}(x)$ as done in~\cite{BLF--PLF} gives an equivalent H\"older norm~\eqref{equa-norm-Hol-a}.

\begin{lemma}[Projection in local weighted norms]\label{lem:local-projection}
Fix exponents and localization geometry as in Sections~\ref{section-3}--\ref{section:slow-coupled-projection}, with in particular $0<p<1$, $p_G>0$, and $p_A\geq\max(p,p_G)$ with $p_A\neq 1$.
Consider $C^{N,\alpha}$ seed data sets $(\seedg,\seedh)$ such that both the $C^{N,\alpha}_{p_G}\times C^{N,\alpha}_{p_G+1}$ norm of $(\seedg-\ebf,\seedh)$ and the constraint error norm $\Eappr{p_A}[\seedg,\seedh]$ defined in~\eqref{eq:full-error-norm} are sufficiently small.  In particular, the constraint error is supported in the gluing domain~$\Omega$.
Then the seed-to-solution projection $(g,h)$ of the seed data set, defined by~\eqref{eq:seed-to-solution}, exists, with $(g-\seedg,h-\seedh)$ controlled in $C^{N,\alpha}_{p,\underline P}\times C^{N,\alpha}_{p+1,\underline P}$ by $\Eappr{p}[\seedg,\seedh]$.
This projection is locally Lipschitz in the seed data set and source norms with radial exponent~$p$ given in~\eqref{eq:all-p-norms}, below, and preserves time symmetry in the sense that $\seedh=0$ implies $h=0$. Under the Hamiltonian and momentum stability assumptions,
\begin{itemize}
\item if $p_A<1$ then the sub-harmonic conclusions of \cite[Proposition~4.7]{BLF--PLF} hold, and in particular the $C^{N,\alpha}_{p_A,\underline P}\times C^{N,\alpha}_{p_A+1,\underline P}$ norm of $(g-\seedg,h-\seedh)$ is controlled by $\Eappr{p_A}[\seedg,\seedh]$;
\item if $p_A>1$ then the super-harmonic conclusions of \cite[Theorem~4.4]{BLF--PLF} hold, and in particular there exist four parameters $\mathfrak m,\mathfrak J_1,\mathfrak J_2,\mathfrak J_3\in\RR$ such that the $C^{N,\alpha}_{p_A,\underline P}\times C^{N,\alpha}_{p_A+1,\underline P}$ norm of $(g-\seedg-\mathfrak m\,\mathcal S^\lambda,h-\seedh - \sum_j\mathfrak J_j\mathcal T^\lambda_j)$ is controlled by $\Eappr{p_A}[\seedg,\seedh]$, where $\mathcal S^\lambda$ and $\mathcal T^\lambda_j$ are defined in \eqref{eq:silhouette-tensor} and~\eqref{eq:momentum-silhouette}.
\end{itemize}
\end{lemma}
\begin{proof}
This lemma is for the most part a specialization of Theorem~4.4 and Proposition~4.7 of~\cite{BLF--PLF} to the case $p_\star=p_A$, with $p_A\neq 1$ to avoid slight complications in the harmonic case, and to the case of a reference data set $(\ebf,0)$ and localization domain~$\Omega$ with a single asymptotic end, namely a connected angular domain~$\Lambda$ (so as to only have one set of mass and momentum modulators).
Slight differences in the definition of H\"older norms, or whether one uses the raised or lowered momentum constraints $\mathcal C^j$ or~$\mathcal M_j$, only affect multiplicative constants in the smallness bounds.

There only remains to show that the seed-to-solution projection is locally Lipschitz in the appropriate norms~\eqref{eq:all-p-norms} below, with decay exponent~$p$ (as opposed to $p_A$ or the sharp decay exponent~$p_\star$).  To this end, we recall a few key aspects of the construction in~\cite{BLF--PLF}.

From~\eqref{eq:seed-to-solution} and~\eqref{eq:full-reference-adjoint}, we recall that starting from a seed data set $S=(\seedg,\seedh)$, the vacuum data set $X=(g,h)$ is sought in a form parametrized by a scalar and vector potential $U\coloneqq(u,Z)$ as
\be
X = S + \mathcal DU , \qquad \mathcal DU \coloneqq (\mathcal D_g u,\mathcal D_h Z) ,
\ee
where $\mathcal D_g,\mathcal D_h$ are weighted adjoints of the linearized Einstein constraints at the Euclidean data set~$(\ebf,0)$.
The functional spaces that are relevant for the existence of the seed-to-solution map and its Lipschitz continuity are as follows:
\begin{itemize}
\item $U\coloneqq(u,Z)$ lies in $\mathcal W_p\coloneqq(H^2_{1-p,-P}\cap C^{N+2,\alpha}_{1-p,-\overline P+2})\times(H^1_{1-p,-P}\cap C^{N+1,\alpha}_{1-p,-\overline P+1})$,
\item $X-S=(g-\seedg,h-\seedh)$ lies in $\mathcal X_p\coloneqq(L^2_{p,P}\cap C^{N,\alpha}_{p,\underline P})\times (L^2_{p+1,P}\cap C^{N,\alpha}_{p+1,\underline P})$,
\item $\mathcal G=(\mathcal H,\mathcal C)$ lies in $\mathcal Y_p\coloneqq (L^2_{p+2,P} \cap C^{N-2,\alpha}_{p+2,\underline P-2})\times(L^2_{p+2,P} \cap C^{N-1,\alpha}_{p+2,\underline P-1})$.
\end{itemize}
These spaces are equipped with the sum norms
\be\label{eq:all-p-norms}
\begin{aligned}
\|U\|_{\mathcal W_p} & \coloneqq \|u\|_{H^2_{1-p,-P}} + \|u\|_{C^{N+2,\alpha}_{1-p,-\overline P+2}} + \|Z\|_{H^1_{1-p,-P}} + \|Z\|_{C^{N+1,\alpha}_{1-p,-\overline P+1}} ,
\\
\|X-S\|_{\mathcal X_p} & \coloneqq \|g-\seedg\|_{L^2_{p,P}} + \|g-\seedg\|_{C^{N,\alpha}_{p,\underline P}} + \|h-\seedh\|_{L^2_{p+1,P}} + \|h-\seedh\|_{C^{N,\alpha}_{p+1,\underline P}} ,
\\
\|\mathcal G\|_{\mathcal Y_p} & \coloneqq \|\mathcal H\|_{L^2_{p+2,P}} + \|\mathcal H\|_{C^{N-2,\alpha}_{p+2,\underline P-2}} + \|\mathcal C\|_{L^2_{p+2,P}} + \|\mathcal C\|_{C^{N-1,\alpha}_{p+2,\underline P-1}} .
\end{aligned}
\ee
The last of these norms coincides with~$\Eappr{b}$ defined in~\eqref{eq:full-error-norm} with $b=p$ in the sense that
\be
\Eappr{p}[\seedg,\seedh] = \bigl\|\bigl(\mathcal H(\seedg,\seedh),\mathcal C(\seedg,\seedh)\bigr)\bigr\|_{\mathcal Y_p} .
\ee

For fixed small seed $S=(\seedg,\seedh)$, we define the squared Einstein constraint operator $\mathcal J_S\coloneqq D\mathcal G_S\circ\mathcal D$ by composing the constraints linearized at~$S$ with the operator~$\mathcal D$ based on the Euclidean metric.
The results of \cite[Lemmas~E.3 and~E.6]{BLF--PLF} imply that for any source $F\in\mathcal Y_p$ one has
\be
\label{eq:full-local-inverse}
\|\mathcal J_S^{-1}F\|_{\mathcal W_p}\leq C\|F\|_{\mathcal Y_p}.
\ee
A weaker norm than~$\mathcal Y_p$ is in fact sufficient, involving dual Sobolev spaces instead of Lebesgues spaces (schematically $H^{-2}\times H^{-1}$ instead of $L^2\times L^2$).
On this larger functional space, the map $S\mapsto \mathcal J_S$ is smooth, and the resolvent identity gives
\be
\label{eq:inverse-parameter-derivative}
D(\mathcal J_S^{-1})[\dot S]=-\mathcal J_S^{-1}\bigl(D\mathcal J_S[\dot S]\bigr)\mathcal J_S^{-1}.
\ee
Repeated differentiation gives three continuous derivatives with uniform bounds on a smaller neighborhood.

We also define the nonlinear remainder $\mathcal N_S$ by Taylor's formula $\mathcal G(S+\mathcal DU)=\mathcal G(S)+\mathcal J_SU+\mathcal N_S(U)$.
The exact formulas~\eqref{eq:full-constraints} and~\eqref{eq:momentum-exact}, with the local product rule, give, for sufficiently close $U,V$
\be
\label{eq:local-nonlinear-lipschitz}
\|\mathcal N_S(U)-\mathcal N_S(V)\|_{\mathcal Y_p} \leq C(\|U\|_{\mathcal W_p}+\|V\|_{\mathcal W_p})\|U-V\|_{\mathcal W_p} ,
\ee
namely $U\mapsto\mathcal N_S(U)$ is locally Lipschitz.
Then, $\mathcal G(X)=\mathcal G(S)+\mathcal J_SU+\mathcal N_S(U)=0$ gives
\be
U = - \mathcal J_S^{-1} (\mathcal G(S) + \mathcal N_S(U)) .
\ee

We set $b\coloneqq\|\mathcal G(S)\|_{\mathcal Y_p}$. By~\eqref{eq:full-local-inverse}--\eqref{eq:local-nonlinear-lipschitz}, the map
$U\longmapsto-\mathcal J_S^{-1}\bigl(\mathcal G(S)+\mathcal N_S(U)\bigr)$
maps the closed ball of radius $2Cb$ in~$\mathcal W_p$ into itself and has Lipschitz constant strictly less than one for small~$b$. Successive iterates converge by the geometric-series bound and their limit solves the constraints. The same estimate for two seed data sets, together with the resolvent identity for $\mathcal J_S^{-1}$, proves local Lipschitz dependence. If $\seedh=0$, the subspace $Z=0$ is preserved because the momentum equation vanishes there. Uniqueness in the contraction ball gives zero projected momentum.
\end{proof}

%============================================================================

\section{Radial identities and proofs for choosing relative and total energy}\label{appendix:radial-selection}
We prove the radial identities used to construct the seed data sets. We then choose the relative energy coefficient and match the total ADM energy to the retained Schwarzschild parameter.

\phantomsection\label{appendix:spherical-divergence}
\begin{proof}[Proof of Lemma~\ref{lem:double-divergence}]
We compute the radial and tangential components of the Euclidean divergence of~$A$:
\be
\label{eq:first-divergence}
(\operatorname{div}_{\ebf}A)_r=\del_ra+\frac{2a+\divs b-\trs C}{r},\qquad (\operatorname{div}_{\ebf}A)_{\mathrm{tan}}=\del_rb+\frac{3b+\divs C}{r}.
\ee
For a vector field $V=V_r\theta+V_{\mathrm{tan}}$, the divergence is $\operatorname{div}_{\ebf}V=\del_rV_r+\frac{2}{r}V_r+\frac1r\divs V_{\mathrm{tan}}$. We substitute~\eqref{eq:first-divergence} into this identity and use $D^2a=r^2\del_r^2a+r\del_ra$. The result is~\eqref{eq:spherical-L}.
\end{proof}

\phantomsection\label{appendix:radial-momentum}
\begin{proof}[Proof of Lemma~\ref{lem:momentum-radial}]
We apply~\eqref{eq:first-divergence} to $r^{-q-1}\pi_0$. The radial and tangential divergence coefficients are given by~\eqref{eq:momentum-homogeneous}. They vanish for~\eqref{eq:momentum-profile}. Since $\pi_0$ has tangential part $C_0$, it is nonzero for $C_0\neq0$ and has the same support.

For the inhomogeneous problem we set the tangential two-tensor to zero and solve~\eqref{eq:momentum-polynomial-inverse} in the indicated order. For a nonzero constant $c$, the polynomial inverse of $\del_\ell+c$ applied to a polynomial $F$ of degree $d$ is $\sum_{j=0}^d(-1)^jc^{-j-1}\del_\ell^jF$. For $c=0$, we integrate the polynomial and choose zero integration constant. These formulas prove existence at every exponent, including $q=1,2$, and preserve angular support. Substitution into~\eqref{eq:first-divergence} proves~\eqref{eq:momentum-inhomogeneous}.
\end{proof}

\phantomsection\label{appendix:scalar-selection}
\begin{proof}[Proof of Corollary~\ref{thm:modulator-selection}]\bse
We apply Lemma~\ref{lem:modulator-equation}. We set $s_0\coloneqq-m\langle\chi\rangle_{S^2}-a$ and $\eta\coloneqq |m|+|a|$. The equation $\mathfrak m=a$ with $\mathfrak m$ given in~\eqref{eq:modulator-equation} reduces to the fixed-point equation
\be
\label{eq:scalar-fixed-point}
s=s_0+\mathcal N(m,s).
\ee
For $|s-s_0|\leq K\eta^2/R$ and $\eta/R$ sufficiently small, we have $|m|+|s|\leq C_0\eta$. We choose $K>C C_0^2$, with $C$ from~\eqref{eq:modulator-equation}, before shrinking $\eta/R$. The continuous map on the right-hand side of~\eqref{eq:scalar-fixed-point} then maps this closed interval around~$s_0$ into itself. The intermediate value theorem, applied to its difference with~$s$, gives a fixed point. If $\eta=0$, we take $s_*=0$.
This proves the existence of the required scalar parameter.

The bound for $s_*$ follows from the width of this interval. All parameters remain in the neighborhood where the seed data set and projection estimates are uniform. With $a=0$,~\eqref{eq:parametric-expansion} becomes
\be
\label{eq:zero-modulator}
g(m,s_*)-\seedg(m,s_*)\in C^{N,\alpha}_{p_\star,\underline P}(\Omega).
\ee
\ese\end{proof}

\phantomsection\label{appendix:Sch-ADM}
\begin{proof}[Proof of Proposition~\ref{prop:Sch-mass-prescribed}]
The exact Schwarzschild--Minkowski solution of Theorem~\ref{thm:fixed-radius}, measured in a norm with a radial weight $q\in(1/2,1)$, has dimensionless size $\nu_0=\Obig(m/R)$ and total ADM vector $\Obig(m^2/R)$ by~\eqref{eq:total-four-integral}. With $\varepsilon=M_*/R$, the smallness assumption gives $\nu_0<c_0\sqrt\varepsilon$. Proposition~\ref{prop:fixed-background-ADM} therefore prescribes $(m,mv_0)$, retaining the exterior values. For $v_0=0$ this gives zero ADM momentum yet $K\neq0$, since the compact leading terms~$\pi_\ell$ have non-zero weights.

To achieve a time symmetric vacuum data set, we choose a single nonzero compact transverse-traceless metric profile $\gamma$ and take $\pi=0$, normalized by $\mathscr E(\gamma,0)=1$. The projection preserves $h=0$. The same quadratic calculation and small-background estimate give its dimensionless total energy as
\be
\label{eq:TS-Sch-mass-target}
\begin{aligned}
F(a)&=F(0)+a^2+\mathcal R(a),\qquad F(0)=\Obig(\varepsilon^2),\\
|\mathcal R(a)|&\leq C(\varepsilon|a|+|a|^3),\qquad |\mathcal R'(a)|\leq C(\varepsilon+a^2).
\end{aligned}
\ee
For $a=\sqrt{\varepsilon y}$, the equation $F(a)=\varepsilon$ reads $y=1-\varepsilon^{-1}F(0)-\varepsilon^{-1}\mathcal R(\sqrt{\varepsilon y})$. On a fixed interval about $y=1$ contained in $(0,\infty)$, the right-hand side differs from one by $\Obig(\sqrt\varepsilon)$ and has derivative $\Obig(\sqrt\varepsilon)$. It is a contraction there for small~$\varepsilon$, proving the exact target. Rescaling gives total energy~$M_*$. The momentum is zero because $h=0$, and the compact perturbation followed by seed-to-solution projection retains the prescribed exterior Schwarzschild and Euclidean metrics.
\end{proof}

%============================================================================

\section{Quadratic scalar-curvature and ADM energy-momentum formulas}\label{appendix:quadratic-curvature}
We give the quadratic scalar-curvature calculation used in the asymptotic construction. We then prove the ADM formulas and construct the compact transverse-traceless tensors used in Theorem~\ref{thm:effective-ADM}.

%-----------------------------------------------------------------------------------------------------------------------------------------

\subsection{First nonlinear scalar-curvature correction}\label{appendix:first-correction}
For illustration of the expansion in Section~\ref{section-2}, we compute the second-order error of Schwarzschild--Minkowski interpolation, which features a logarithmic part. The expressions \eqref{eq:Christoffel-first}--\eqref{eq:Q-explicit} are also useful later on.
For a symmetric tensor $\gamma$, we use Cartesian indices and define
\be
\label{eq:Christoffel-first}
\Gamma_1{}^k{}_{ij}[\gamma]\coloneqq\tfrac12(\del_i \gamma_{jk}+\del_j \gamma_{ik}-\del_k \gamma_{ij}),\qquad \Gamma_2{}^k{}_{ij}[\gamma]\coloneqq-\gamma_{ka}\Gamma_1{}^a{}_{ij}[\gamma].
\ee
The quadratic term in the scalar curvature is
\be
\label{eq:Q-explicit}
\Qsc(\gamma)\coloneqq\delta_{ij}\bigl(\del_a\Gamma_2{}^a{}_{ij}-\del_j\Gamma_2{}^a{}_{ia}+\Gamma_1{}^a{}_{ab}\Gamma_1{}^b{}_{ij}-\Gamma_1{}^a{}_{jb}\Gamma_1{}^b{}_{ia}\bigr)-\gamma_{ij}\bigl(\del_a\Gamma_1{}^a{}_{ij}-\del_j\Gamma_1{}^a{}_{ia}\bigr).
\ee
Thus $R_{\ebf+t \gamma}=tL \gamma+t^2\Qsc(\gamma)+\Obig(t^3)$ whenever the metric remains uniformly positive. The first two coefficients of~\eqref{eq:basic-exterior} in powers of~$m$ are $G_1$ and $G_2\mathrel{\coloneqq}3\chi^2\ebf/(2r^2)$. We define the smooth angular function
\be
\label{eq:F2}
F_2\coloneqq r^4\bigl(LG_2+\Qsc(G_1)\bigr).
\ee
The right-hand side is independent of~$r$ thanks to the scaling. Where $\chi$ is constant,~\eqref{eq:basic-exterior} is a Schwarzschild metric with constant mass $m\chi$, so $F_2=0$. In particular, $\supp F_2\Subset\Lambda$.

Since $(\del_\ell)(\del_\ell-1)\ell=-1$ and $\ell=\log(r/R)$, the second corrected exterior metric is
\be
\label{eq:second-exterior}
g_m^{[2]}=g_m^{[1]}+\frac{m^2}{r^2}\log\frac rR\,F_2\Crad.
\ee
For every fixed $k$, its error satisfies
\be
\label{eq:second-error}
|\del^kR_{g_m^{[2]}}|\leq C_k |m|^3r^{-5-k}(1+\ell)^{K_k},\qquad r\geq R,
\ee
for some finite integers $K_k$. The logarithm is compatible with every accuracy exponent $p_A<3$, but the pointwise bound alone does not give the limiting weighted $L^2$ estimate at $p_A=3$.

%-----------------------------------------------------------------------------------------------------------------------------------------

\subsection{The quadratic ADM energy-momentum}\label{appendix:quadratic-ADM}
We calculate the first nonzero ADM contribution of a compact transverse-traceless perturbation. The parameter estimates justify the finite-dimensional inversion in Theorem~\ref{thm:effective-ADM}.
\begin{proof}[Proof of Proposition~\ref{lem:TT-quadratic}]\bse
The linearized constraints vanish by~\eqref{eq:TT-conditions}. The compact quadratic error gives projection correction $\Obig(a^2)$ in $C^{N,\alpha}_{p,\underline P}\times C^{N,\alpha}_{p+1,\underline P}$, $p>1/2$. Thus~\eqref{eq:total-four-integral} converges absolutely, and substituting the linear pair in its quadratic integrand changes the integral only by $\Obig(a^3)$.

The first Christoffel tensor~\eqref{eq:Christoffel-first} satisfies $\Gamma_1{}^a{}_{ia}=0$, while $(\Ric_L \gamma)_{ij}=-\tfrac12\Delta \gamma_{ij}$. Divergence terms in~\eqref{eq:Q-explicit} integrate to zero, and
\be
\label{eq:TT-Christoffel-integral}
\int\Gamma_1{}^a{}_{ib}\Gamma_1{}^b{}_{ia}\,dy =-\tfrac14\int|\del \gamma|^2\,dy, \qquad \int \gamma_{ij}(\Ric_L \gamma)_{ij}\,dy=\tfrac12\int|\del \gamma|^2\,dy.
\ee
For the first identity, we multiply $2\Gamma_1{}^a{}_{ib}=\del_i \gamma_{ab}+\del_b \gamma_{ia}-\del_a \gamma_{ib}$ by the expression with $a,b$ exchanged. The integrated cross term $\int\del_b \gamma_{ia}\del_a \gamma_{ib}$ vanishes by integration by parts and $\del_b \gamma_{ib}=0$. The second identity is ordinary integration by parts. Consequently $\int\Qsc(\gamma)=-\tfrac14\int|\del \gamma|^2$. The quadratic momentum term in the Hamiltonian constraint is $-|\pi|^2$, since $\Tr\pi=0$. This proves the energy formula and its sign.

For the momentum~$Q_M$ in~\eqref{eq:full-Q}, the first term $\Gamma_1{}^i{}_{i\ell}\pi^{\ell j}$ vanishes. Symmetry gives
\be
\label{eq:TT-momentum-integral}
\int\Gamma_1{}^j{}_{i\ell}\pi^{i\ell}\,dy =\int\bigl(\pi^{i\ell}\del_i \gamma_{\ell j}-\tfrac12\pi^{i\ell}\del_j \gamma_{i\ell}\bigr)dy =-\tfrac12\int\pi^{i\ell}\del_j \gamma_{i\ell}\,dy,
\ee
where the first term integrates to zero by $\del_i\pi^{i\ell}=0$. The minus sign in~\eqref{eq:total-four-integral} now gives the positive sign for~$\mathscr P$.
\ese\end{proof}

%-----------------------------------------------------------------------------------------------------------------------------------------

\subsection{Compact transverse-traceless profiles}\label{appendix:compact-TT}
We construct compact profiles whose momentum-to-energy ratio approaches any prescribed unit vector. Their differential definition preserves support and allows disjoint profiles to be used independently.
\begin{proof}[Proof of Proposition~\ref{lem:TT-packets}]\bse
We define the commuting differential operators $\mathbb P_{ij}\coloneqq-\Delta\delta_{ij}+\del_i\del_j$. For a constant symmetric tensor $S$ and a compact scalar function $\phi$, we set
\be
\label{eq:compact-TT-operator}
(\mathbb T_S\phi)_{ij} \coloneqq\mathbb P_{ia}\mathbb P_{jb}(S_{ab}\phi)-\tfrac12\mathbb P_{ij}\mathbb P_{ab}(S_{ab}\phi).
\ee
The identities $\del_i\mathbb P_{ij}=0$, $\sum_i\mathbb P_{ia}\mathbb P_{ib}=-\Delta\mathbb P_{ab}$, and $\sum_i\mathbb P_{ii}=-2\Delta$ prove that this tensor is transverse and trace free. These are identities of constant-coefficient operators and follow by expansion. The operator preserves compact support.

We choose $S\neq0$ with $Sn=0$ and $\Tr S=0$, and choose a nonzero smooth cutoff $\chi$ supported in the ball. We write $\del_n\coloneqq n\cdot\nabla_y$ for differentiation in the direction~$n$. For a frequency $\omega\to\infty$, we set
\be
\label{eq:TT-oscillations}
\phi_\omega\coloneqq\omega^{-4}\chi(y)\cos(\omega n\cdot y),\qquad \gamma_\omega\coloneqq\mathbb T_S\phi_\omega,\qquad \pi_\omega\coloneqq\tfrac12\del_n \gamma_\omega.
\ee
Both tensors satisfy~\eqref{eq:TT-conditions}. Expanding the four derivatives gives $\gamma_\omega=\chi S\cos(\omega n\cdot y)+\Obig(\omega^{-1})$ and $\del_j \gamma_\omega=-\omega n_j\chi S\sin(\omega n\cdot y)+\Obig(1)$. All errors are supported on the ball. Integration by parts in the $n$ direction shows that the average of $\chi^2(\sin^2(\omega n\cdot y)-1/2)$ tends to zero. We therefore obtain
\be
\label{eq:TT-velocity-limit}
\mathscr E(\gamma_\omega,\pi_\omega)=c\omega^2+\Obig(\omega),\qquad \mathscr P(\gamma_\omega,\pi_\omega)=c\omega^2n+\Obig(\omega),\qquad c\coloneqq\frac1{64\pi}\int\chi^2|S|^2dy>0.
\ee
Division of both tensors by $\sqrt{\mathscr E(\gamma_\omega,\pi_\omega)}$ gives the normalization and limiting velocity. Finally, for any unit~$v$, the inequality $|\pi||\del_v \gamma|\leq|\pi|^2+\tfrac14|\del_v \gamma|^2$ gives $|v\cdot\mathscr P|\leq\mathscr E$. Taking the supremum over~$v$ proves the last assertion.
\ese\end{proof}


\begin{thebibliography}{99}

\bibitem{AndersonCorvinoPasqualotto}
Anderson J., Corvino J., and Pasqualotto F.,
Multi-localized time-symmetric initial data for the Einstein vacuum equations,
J. Reine Angew. Math. 808 (2024), 67--110.
% \href{https://doi.org/10.1515/crelle-2023-0088}{doi:10.1515/crelle-2023-0088}.

\bibitem{AretakisCzimekRodnianski}
Aretakis S., Czimek S., and Rodnianski I.,
The characteristic gluing problem for the Einstein equations and applications,
Duke Math. J. 174 (2025), no.~2, 355--402.
% \href{https://doi.org/10.1215/00127094-2024-0030}{doi:10.1215/00127094-2024-0030}.

\bibitem{BieriGarfinkleIsenbergMaxwellWheeler}
Bieri L., Garfinkle D., Isenberg J., Maxwell D., and Wheeler J.,
Asymptotically Euclidean solutions of the constraint equations with prescribed asymptotics,
J. Math. Phys. 67 (2026), no.~10, 102502.
% \href{https://doi.org/10.1063/5.0323682}{doi:10.1063/5.0323682}.
% arXiv:2512.21274.

\bibitem{CarlottoSchoen}
Carlotto A. and Schoen R., 
Localizing solutions of the Einstein constraint equations, 
Invent. Math. 205 (2016), 559--615.


\bibitem{ChenKlainerman}
Chen X. and Klainerman S.,
Solving the constraint equation for general free data,
Preprint \href{https://arxiv.org/abs/2512.22704}{arXiv:2512.22704}, 2025.

\bibitem{ChenKlainermanForward}
Chen X. and Klainerman S.,
Forward construction of vacuum initial data with borderline decay,
Preprint \href{https://arxiv.org/abs/2606.08716}{arXiv:2606.08716}, 2026.

\bibitem{ChruscielDelay}
Chru\'sciel P.T. and Delay E.,
On mapping properties of the general relativistic constraints operator in weighted function spaces, with applications,
M\'em. Soc. Math. France (N.S.) 94 (2003), 109~pp.
% \href{https://doi.org/10.24033/msmf.407}{doi:10.24033/msmf.407}.
% arXiv:gr-qc/0301073.

\bibitem{ChruscielDelayHyperbolic}
Chru\'sciel P.T. and Delay E.,
Exotic hyperbolic gluings,
J. Diff. Geom. 108 (2018), 243--293.
% \href{https://doi.org/10.4310/jdg/1518490818}{doi:10.4310/jdg/1518490818}.
% arXiv:1511.07858.

\bibitem{ChruscielDelayConical}
Chru\'sciel P.T. and Delay E.,
On Carlotto--Schoen-type scalar-curvature gluings,
Bull. Soc. Math. France 149 (2021), 641--662.
% \href{https://doi.org/10.24033/bsmf.2840}{doi:10.24033/bsmf.2840}.
% arXiv:1611.00893.

\bibitem{Corvino-2000} 
Corvino J.,
Scalar curvature deformation and a gluing construction for the Einstein constraint equations, 
Comm. Math. Phys. 214 (2000), 137--189.

\bibitem{CorvinoSchoen} 
Corvino J. and Schoen R.,
On the asymptotics for the vacuum Einstein constraint equations, 
J. Diff. Geom. 73 (2006), 185--217.

\bibitem{CzimekRodnianski}
Czimek S. and Rodnianski I.,
Obstruction-free gluing for the Einstein equations,
Preprint \href{https://arxiv.org/abs/2210.09663}{arXiv:2210.09663}, 2022.

\bibitem{Fang-1}
Fang A.J., Szeftel J., and Touati A.,
Initial data for Minkowski stability with arbitrary decay,
Adv. Theor. Math. Phys. 29 (2025), no.~4, 933--1043.

\bibitem{Fang-2}
Fang A.J., Szeftel J., and Touati A.,
Spacelike initial data for black hole stability,
Comm. Math. Phys. 406 (2025), no.~10, article~235.
% \href{https://doi.org/10.1007/s00220-025-05416-0}{doi:10.1007/s00220-025-05416-0}.

\bibitem{BLF--PLF} Le Floch B. and LeFloch P.G., Optimal localization for the Einstein constraints, Preprint \href{https://arxiv.org/abs/2312.17706}{arXiv:2312.17706}.

\bibitem{BLF--PLF-stability}
Le Floch B. and LeFloch P.G., Harmonic, radial, and shell stability of the weighted Einstein constraints on the sphere at infinity, Preprint \href{https://arxiv.org/abs/2607.28599}{arXiv:2607.28599}.

\bibitem{BLF--PLF-shielding}
Le Floch B. and LeFloch P.G.,
Optimal shielding for Einstein gravity,
Class. Quantum Grav. 41 (2024), 13LT02.
% \href{https://doi.org/10.1088/1361-6382/ad5137}{doi:10.1088/1361-6382/ad5137}.
% arXiv:2402.17598.

\bibitem{LeFlochNguyen}
LeFloch P.G. and Nguyen T.-C.,
The seed-to-solution method for the Einstein constraints and the asymptotic localization problem, 
J. Funct. Analysis 285 (2023), 110106.

\bibitem{MaoOhTao}
Mao Y., Oh S.-J., and Tao Z.,
Initial data gluing in the asymptotically flat regime via solution operators with prescribed support properties,
Preprint \href{https://arxiv.org/abs/2308.13031}{arXiv:2308.13031}, 2023.

\bibitem{MaoTao}
Mao Y. and Tao Z.,
Localized initial data for Einstein equations,
Preprint \href{https://arxiv.org/abs/2210.09437}{arXiv:2210.09437}, 2022.

\bibitem{MaxwellConformal}
Maxwell D.,
The conformal method and the conformal thin-sandwich method are the same,
Class. Quantum Grav. 31 (2014), no.~14, 145006.
% \href{https://doi.org/10.1088/0264-9381/31/14/145006}{doi:10.1088/0264-9381/31/14/145006}.
% arXiv:1402.5585.

\bibitem{Sansom-2024}
Sansom R.,
The $C^3$-null gluing problem: linear and nonlinear analysis,
Ann. Henri Poincar\'e (2025).
% \href{https://doi.org/10.1007/s00023-025-01619-2}{doi:10.1007/s00023-025-01619-2}.
\end{thebibliography}
\end{document}